\documentclass[11pt, oneside]{article}   	
\usepackage{amsmath,amssymb,amsthm,color,url,framed}
\usepackage{amscd}
\usepackage{enumerate}
\usepackage[margin=2cm]{geometry}                		
\usepackage{graphicx}				
\usepackage[color=yellow]{todonotes}
\usepackage{hyperref}
\hypersetup{
colorlinks=true,
}

\usepackage{fancyhdr}
\numberwithin{equation}{section}
\theoremstyle{plain}
\newtheorem{theorem}{Theorem}[section]     
\newtheorem{corollary}[theorem]{Corollary}             
\newtheorem{lemma}[theorem]{Lemma}              
\newtheorem{proposition}[theorem]{Proposition}

\theoremstyle{definition}
\newtheorem{definition}[theorem]{Definition}             
\theoremstyle{remark}
\newtheorem{remark}[theorem]{Remark}              
\newtheorem{example}[theorem]{Example}                    

\def\p{{\partial}}
\def\pa{{\partial}}
\def\bx{{\bf x}}

\def\bk{{\bf k}}
\def\bm{{\bf m}}
\def\bp{{\bf p}}

\def\bR{{\bf R}}
\def\bu{\mathbf{u}}

\def\bv{{\bf v}}

\def\bxi{{\boldsymbol \xi}}

\newcommand{\scp}[2]{{\big\langle {#1}\, , \, {#2}\big\rangle}}
\newcommand{\Scp}[2]{{\Big\langle {#1}\, , \, {#2}\Big\rangle}}

\newcommand{\ob}[1]{\overline{#1}}

\newcommand{\mb}[1]{\mbox{\boldmath{$#1$}}}
\newcommand{\sym}[1]{\boldsymbol{#1}}
\newcommand{\mbb}[1]{\mathbb{#1}}

\newcommand*\mean[1]{\overline{#1}}
\DeclareMathOperator{\curl}{curl}

\newcommand{\E}[1]{\mathbb{E}\left[{#1}\right]}

\usepackage[scr=boondoxo, scrscaled=1]{mathalfa}

\title{Stochastic Variational Formulations of Fluid Wave-Current Interaction}
\author{Darryl D Holm\\
Imperial College London\\
email: d.holm@ic.ac.uk}
\date{}							

\begin{document}
\maketitle


\begin{abstract}
We are modelling multi-scale, multi-physics uncertainty in wave-current interaction (WCI). To model uncertainty in WCI, we introduce stochasticity into the wave dynamics of two classic models of WCI; namely,  the Generalised Lagrangian Mean (GLM) model and the Craik--Leibovich (CL) model.  \smallskip

The key idea for the GLM approach is the separation of the Lagrangian (fluid) and Eulerian (wave) degrees of freedom in Hamilton's principle. This is done by coupling an Euler--Poincar\'e {\it reduced Lagrangian} for the current flow and a {\it phase-space Lagrangian} for the wave field. WCI in the GLM model involves the nonlinear Doppler shift in frequency of the Hamiltonian wave subsystem, which arises because the waves propagate in the frame of motion of the Lagrangian-mean velocity of the current. In contrast, WCI in the CL model arises because the fluid velocity is defined relative to the frame of motion of the Stokes mean drift velocity, which is usually taken to be prescribed, time independent and driven externally. 

We compare the GLM and CL theories by placing them both into the general framework of a stochastic Hamilton's principle for a 3D Euler--Boussinesq (EB) fluid in a rotating frame. In other examples, we also apply the GLM and CL methods to add wave physics and stochasticity to the familiar 1D and 2D shallow water flow models.

 The differences in the types of stochasticity which arise for GLM and CL models can be seen by comparing the Kelvin circulation theorems for the two models. The GLM model acquires stochasticity in its Lagrangian transport velocity for the currents and also in its group velocity for the waves.  However, the CL model is based on defining the Eulerian velocity in the integrand of the Kelvin circulation relative to the Stokes drift velocity induced by waves driven externally. Thus, the Kelvin theorem for the stochastic CL model can accept stochasticity in its both its integrand and in the Lagrangian transport velocity of its circulation loop. 
 
In an appendix, we also discuss dynamical systems analogues of WCI.
\end{abstract}


\section{Introduction}

The first objective of this paper is to build a consistent variational theory of the interactions of the wave and current degrees of freedom for two quite different approaches to mean wave-current interaction (WCI). The two different approaches are the generalised Lagrangian mean (GLM) model  \cite{AM1978} for wave motion in the ocean or atmosphere, and the Craik--Leibovich (CL) model \cite{CraikLeibovich1976} for air-sea interaction due to wind and waves on the sea surface.

After this first objective has been achieved, we will introduce several types of noise into these deterministic variational formulations, and develop a new theoretical basis for modelling uncertainty in a theory of WCI which \emph{combines} aspects of both GLM and CL. 

\paragraph{The generalised Lagrangian mean (GLM) model.}
In the GLM model of WCI, the current is interpreted as the Lagrangian-mean flow velocity, $\bu^L(\bx,t)$, while the wave phase, $\phi(\bx,t)$ and the wave action density, $N(\bx,t)$ are interpreted as \emph{Eulerian-mean fields}. This dual interpretation is intuitively clear, because the waves would propagate through the fluid even if it were not moving. It is also clear in Kelvin's circulation integral, in which the loop is moving in a Lagrangian sense and the integrand is an Eulerian quantity in fixed spatial coordinates. Thus, waves and currents would naturally be treated separately in applying Hamilton's variational principle, namely $\delta S = 0$ for an action integral $S = \int_{t_1}^{t_2} \ell(\bu^L,N,\phi)dt$, to generate coupled WCI dynamics in the GLM model. 

The key idea we use for deriving the Hamilton's principle for WCI analysis in the GLM model is the introduction of a {\it phase-space Lagrangian} (PSL) written as a Legendre transform $L(\phi,\partial_t{\phi})=\langle N,\partial_t\phi \rangle - H_W(N,\nabla\phi)$ for the canonically conjugate wave degrees of freedom $(N,\phi)$. Here, the brackets $\scp{\,\cdot\,}{\,\cdot\,}$ denote $L^2$ pairing of dual variables. This PSL is manifestly invariant under translations in the phase $\phi$. Noether's theorem then implies conservation of the volume integral of the conjugate momentum $N$ (the wave action density). Our approach follows the PSL formulation of quantum mechanics introduced in 1934 by Frenkel and Dirac \cite{FrenkelDirac-QM1934}. This approach has also become a mainstay of plasma physics, where it has been used to model the time-mean (ponderomotive) forces exerted \emph{externally} by rapid electromagnetic oscillations (e.g. microwaves) on the slow dynamics of a fluid plasma \cite{Dewar1970,Dewar1973,RGL1981,SKH1986,ANK-DH1984}.
For a modern application of the Frenkel--Dirac phase-space Lagrangian in the classical-quantum interaction for non-adiabatic electron dynamics in molecular chemistry, see \cite{FHT2019}. For a recent treatment of phase-space Lagrangians for fast-slow WKB dynamics of high-frequency acoustic waves interacting with a larger-scale compressible isothermal flow, see \cite{BurbyRuiz2019}.

\paragraph{GLM main result.}
For GLM with fluid variables denoted $(u^L,a,b)$ and canonically conjugate wave variables $(q,p)$, the main result of the paper is Theorem \ref{SALT-SNWP-Thm}. This theorem derives a \emph{closed} dynamical GLM theory of WCI which can be extended into \emph{stochastic} wave-current dynamics from Hamilton's principle with action integral given by the following sum of Lebesque and Stratonovich time integrals of a fluid Lagrangian with stochastic advection constraints and a phase-space Lagrangian for the wave variables with a stochastic Hamiltonian,
\begin{align}
\begin{split}
S(u^L,a,b,q,p)  =
\int_{t_1}^{t_2} & \hspace{-2mm}
\underbrace{\
\ell(u^L,a) dt\ 
}_{\hbox{Fluid Lagrangian}}
+ \int_{t_1}^{t_2} \hspace{-2mm}
\underbrace{\
\Scp{  b }{{{\color{red}\rm d} a} + \mathcal{L}_{{\color{red}\rm d}x_t} a}_V\ 
}_{\hbox{Advection Constraint}} 
\\&+ \int_{t_1}^{t_2} 
\underbrace{\ 
\Scp{  p }{ {\color{red}\rm d} q + \mathcal{L}_{{\color{red}\rm d} x_t} q }_V 
-  
\Big(\mathcal{H}(q,p) \,dt + \mathcal{K}(q,p) \circ dB_t \Big)
}_{\hbox{Legendre Transformation in Stochastic Fluid Frame}}
\,,\end{split}
\label{SVP2.2intro}
\end{align}
in which the differential ${\color{red}\rm d}$ appearing,  for example, in the stochastic Eulerian fluid velocity ${\color{red}\rm d} x_t(x)$ denotes the Stratonovich temporal integral. That is, the time integral of the semimartingale vector field 
\begin{align}
{\color{red}\rm d} x_t(x) = u^L(x,t)\,dt + \sum_i \xi_i(x)\circ dW_i(t)\,,
\label{StochVelIntro}
\end{align}
generates the stochastic Lagrangian fluid flow
\begin{align}
x_t(x) -  x_0(x) = \int_0^t u^L(x,s)\,ds + \sum_i \int_0^t \xi_i(x_s(x))\circ dW_i(s)\,,
\label{StochVelIntro}
\end{align}

with spatially dependent correlation eigenvectors $\xi_i(x)$, $i=1,2,\dots,N$, fluid drift velocity vector field $u^L$, advected fluid quantities $(a,b)$, and wave phase-space mean fields $(q,p)$.
Here, the distinct fluid and wave Stratonovich Brownian motions are denoted, respectively, by $\circ dW_i(t)$ and $\circ dB_t$. Note that the wave Hamiltonian in the last term of the action integral in \eqref{SVP2.2intro} is also a semimartingale. 

Mean quantities in this stochastic GLM model are defined as averages over the rapid phase of the wave component of the flow \emph{at fixed Lagrangian coordinate}, as done, e.g., in \cite{GH1996,ANK-DH1984}. Consequently, the PSL basis for the WCI closure derived here is natural in the GLM  approach \cite{AM1978}. In the PSL approach, the WCI closure depends on the dispersion relation, $\omega(\bk)$, which connects the wave-frequency scalar field, $\omega(\bx,t)$, with the wave-number covector field, $\bk(\bx,t)=\nabla\phi(\bx,t)$. The dispersion relation, $\omega(\bk)$, identifies the type of wave being considered in the WCI. It also will determine the Hamiltonian dynamics of the canonically conjugate variables of the wave field,  $(\phi,N)$. The wave variables $(\phi,N)$ evolve in the local reference frame moving with the GLM transport velocity of the mean current, $\bu^L(\bx,t)$. Thus, the wave and current momentum dynamics each contribute independently to the total circulation around every material loop, as interpenetrating fluid degrees of freedom. The independence of these contributions to the total circulation represent the well known non-acceleration result for GLM \cite{AM1978}.  

The GLM closure introduced here is flexible enough to treat a variety of different types of WCI, and it also allows the wave and current components of the flow to be made stochastic independently.  

\paragraph{The Craik--Leibovich (CL) model.} 

Craik and Leibovich \cite{CraikLeibovich1976} derived an expression for the wave-current interaction called the Stokes vortex force (SVF) and showed that the SVF induces roll structures similar to the Langmuir circulations (LCs) observed in the oceanic surface boundary layer driven by the wind. Today, the SVF representation of the wave-current interaction in the momentum equation is often used for numerically modelling the effects of LCs on mixed layer turbulence by using large-eddy simulations (LES), although the theoretical issues are by no means settled \cite{Fujiwara-etal2018,Fujiwara-MellorReply2019,Mellor-Fujiwara2019,Tejada-Martinez2020}. The main discussion of the CL model in this paper is treated in sections \ref{sec:det-CL} and \ref{sec:stoch-CL}.

\paragraph{CL main result.}
The main result of the paper for CL is Theorem \ref{SALT-OU-CL-Thm} which introduces Ornstein-Uhlenbeck (OU) wave dynamics into the CL equations for EB fluid flow as a system of Euler-Poincar\'e equations obtained from Hamilton's principle with the action integral \eqref{Lag-det-CL}. 
The corresponding extension of the CL model for the 3D flow of EB fluid is obtained as an Euler-Poincar\'e equation for Hamilton's principle $\delta S=0$ with action integral given by, cf. equation \eqref{Lag-det-CL},
\begin{equation}
S = \int_{t_1}^{t_2} \int  \left[ \frac12 D |\bu|^2 
 - D \bu \cdot \bu^S(\bx)N_t - g b D  z
 - p(D-1)\right]\,d^3 x\,dt\,.
\label{Lag-OU-CLintro}
\end{equation}
Here, the function $N_t $ is the solution of the Ornstein-Uhlenbeck (OU) stochastic process \cite{OU-ref}
\begin{align}
{\rm \textcolor{red}d}N_t = \theta(\mean{N} - N_t)dt + \sigma dW_t\,,
\label{OU-1}
\end{align}
with long-term mean $\overline{N}$, and real-valued constants $\theta$ and $\sigma$. 
Namely, $N_t $ is the scalar function of time, 
\begin{equation}
N_t = e^{-\theta t}N_0 + (1 - e^{-\theta t})\overline{N}
+ e^{-\theta t} \sigma \int_0^t e^{\theta s}dW_s
\,,
\label{OUsoln-1}
\end{equation}
in which one may assume an initially normal distribution, 
$N_0\approx {\cal N}(\overline{N},\sigma^2/(2\theta))$, with mean $\overline{N}$ and variance $\sigma^2/(2\theta)$. 
Thus, uncertainty in the prescribed Stokes mean drift velocity $\bu^S(\bx)$ may be modelled probabilistically in the  CL equations, by introducing $\bu^S(\bx)N_t$ in the Lagrangian of Hamilton's principle for the wave dynamics of the classic CL equations for EB fluid flow.

Having been derived as a system of Euler-Poincar\'e equations, the probabilistic CL model with OU wave dynamics (called the OU CL model) preserves all of the geometric mechanics properties of the original CL model, including its vorticity dynamics and preservation of potential vorticity. The uncertainty in the wave field in the OU CL model in Theorem \ref{SALT-OU-CL-Thm} appears as an OU term in the \emph{circulation integrand} rather than in the transport velocity of the \emph{circulation loop} as occurs in GLM. In equation \eqref{CL-circ}, the OU term in the circulation integrand  represents large scale effects through which rapidly oscillating forces of wind and waves at the surface of the domain can produce the Stokes mean drift velocity which, in turn, transmits a mean force on the current flow as a momentum shift associated with the moving reference frame. As for the GLM model, the OU CL model also admits the introduction of stochasticity in the Lagrangian transport velocity of its Kelvin circulation loop, as in equation \eqref{CL-circ}.  Thus, we  regard the OU term in the circulation integrand as the slow ponderomotive average effect of the Stokes mean drift due to the random wind and wave oscillations on the air-sea surface. At the same time, we regard the stochastic transport velocity of the circulation loop as the result of rapid, small-scale effects which perturb the Lagrangian trajectories of the CL model. This dual viewpoint is consistent with the stochastic modelling approach to Richardson's metaphor of 3-way interactions among Big, Little and Lesser Whorls whose stochastic theory was developed in \cite{Holm-RichTriple2019}.

Although our efforts here are directed to stochastic variational models of uncertainty in WCI, the fundamental ideas are based on variational derivations of the Navier--Stokes equations from stochastic equations.

\begin{remark}[Variational derivations of the Navier--Stokes equations from stochastic equations] \rm
The derivation of the Navier--Stokes equations in the context of stochastic processes has a long and well-known history. See. e.g., Constantin and Iyer \cite{CoIy2008}, Eyink \cite{Ey2010}, and references therein. Previous specifically variational treatments of stochastic fluid equations generally started from the famous remark by Arnold \cite{Arnold1966} (about Euler's equations for the incompressible flow of an ideal fluid being geodesic for kinetic energy given by the $L^2$ norm of fluid velocity) and they have mainly treated It\^o noise in this context. For more discussion of these variational derivations of stochastic fluid equations and their relation to the Navier-Stokes equations, one should consult original sources such as, in chronological order, Inoue and Funaki \cite{InFu1979}, Rapoport \cite{Ra2000,Ra2002}, Gomes \cite{Go2005}, Cipriano and Cruzeiro \cite{CiCr2007}, Constantin and Iyer \cite{CoIy2008}, Eyink \cite{Ey2010}, Gliklikh \cite{Gl2010}, Arnaudon, Chen and Cruzeiro \cite{ArChCr2012}. We emphasise that the goal of the present work is limited to the derivation of SPDEs for WCI by following the stochastic variational strategy outlined above.  For additional information, review and background references for random perturbations of PDEs and fluid dynamic models, viewed from complementary viewpoints to the present paper, see also Flandoli et al. \cite{Fl2011,FlMaNe2014}. In particular, Flandoli et al. \cite{Fl2011,FlMaNe2014} study the interesting possibility that adding stochasticity can have a regularising effect on fluid equations which might otherwise be ill-posed. 
\end{remark}

\paragraph{Plan of the paper and main content.}$\,$

Section \ref{det-GLM-bkgrnd} introduces the ideas behind \emph{phase-space Lagrangians} (PSLs) and formulates the closed set of deterministic GLM equations for WCI in \eqref{SVP3-det} which will be the basis for the introduction of stochasticity into the GLM model in section \ref{SALTrev-sec}. The closure of the deterministic GLM theory for a given fluid Lagrangian depends on the choice of wave dispersion relation, $\omega(\bk)$, for frequency as a function of wave vector, which appears in the PSL Hamiltonian for the wave field dynamics. This feature is what makes the present GLM closure flexible enough to treat a variety of interactions of the current with different types of waves. For example, upon choosing the dispersion relation for internal waves in equation \eqref{disp_exp}, the WCI equations \eqref{SVP3-det} yield the Generalised Lagrangian Mean (GLM) equations for stratified, rotating, incompressible Euler--Boussinesq (EB) fluid motion in three dimensions \cite{AM1978}. 

Section \ref{SALTrev-sec} reviews the stochastic variational principle underlying the SALT  (Stochastic Advection by Lie Transport) approach to the derivation of stochastic fluid equations which preserve the geometric structure of fluid dynamics \cite{Holm2015}. The  motivations and recent applications of the SALT approach for uncertainty quantification and data assimilation are also briefly discussed \cite{CCHPS2018,CCHPS2019a,CCHPS2019b}. 

Section \ref{SNWP-sec} combines the ideas in the first two sections to extend the SALT approach to \emph{stochastic nonlinear wave propagation} (SNWP) in deriving a new stochastic theory of Wave-Current Interaction (WCI) in which the dynamics of either or both the waves and the currents can be made stochastic. The stochastic WCI model is formulated and its main geometrical mechanics properties are established.

Section \ref{StochGLM-sec} applies the WCI model formulated in section \ref{SNWP-sec} to derive the SALT and SNWP terms for GLM in 3D stratified EB fluids, while Section \ref{SW-WCI-sec} derives  the SALT and SNWP terms for 1D and 2D shallow water WCI  (SW-WCI) equations. Section \ref{SW-WCI-sec} also derives the Hamiltonian structure for SW-WCI, which turns out to recover a type of non-canonical Lie--Poisson bracket which was first discovered for superfluid $^4He$ and $^3He$ in \cite{HK1982} and was later formulated in more general terms by Krishnaprasad and Marsden in \cite{KM1987} who applied the formulation to the dynamics of a rigid body with a flexible attachment. It seems fitting that the WCI should have such deep roots in geometric mechanics. 

Section \ref{concl-sec} summarises the paper's main results for the WCI model derived here. Namely, the WCI model derived here enables the exchange of energy through the coupling between the two different kinds of motion: Lagrangian flows and Eulerian waves. In the WCI approach which we have implemented here, the GLM fluid transport velocity in the rotating frame is determined by taking the difference of the total momentum and the wave momentum in the frame of the transport velocity. The flow velocity is measured relative to the rotation of the Earth and the wave group velocity is measured relative to the flow velocity. This is reminiscent of L. F. Richardson's well-known metaphor of ``whorls within whorls'' for fluid turbulence. For a recent discussion Richardson's metaphor in the context of stochastic parametrisation for geophysical flows, see \cite{Holm-RichTriple2019}. 

\paragraph{Non-acceleration result for wave mean-flow interaction (WMFI).}
In the WCI model derived here, the wave dynamics may create circulation, but only within the wave subsystem of the incompressible EB flow which is transported by the GLM fluid flow.  This is particularly clear in the Kelvin circulation theorem representation of WCI for the example of EB flow, when equation \eqref{Kelvin-GLM} is compared with equations \eqref{SALT-SNWP-GLM-wave} and \eqref{SALT-SNWP-GLM-total}. Thus, equation \eqref{SALT-SNWP-GLM-total} represents a dynamical version of the famous \emph{non-acceleration theorem} for GLM \cite{Vallis2017}. Namely, in the absence of dissipation, equation \eqref{SALT-SNWP-GLM-total} shows that the  presence of waves has no net effect on the mean-flow equations of the GLM WCI model.

Section \ref{sec:det-CL}  compares the deterministic features of the Craik-Leibovich (CL) model with the corresponding deterministic results for GLM discussed in the main text.

Section \ref{sec:stoch-CL}  derives a new stochastic version of CL which differs from the GLM approach both in the type and location of its probabilistic features.  In particular, the probabilistic features are Lagrangian in GLM while they are Eulerian in the CL model. However, the CL model is not subject to a non-acceleration result. This is because wave forcing in the CL model is external, while for GLM the waves represent an internal degree of freedom. The introduction of uncertainty in both the Stokes velocity and in the Lagrangian mean velocity of the CL equations reinforces the concept of \emph{multiscale uncertainty} for the WCI.

Section \ref{concl-sec} reviews the main geometric mechanics ideas underlying our approach to WCI and suggests other open problems which may be treated via this approach. 

Appendix \ref{appendix-A} discusses the gyrostat as a potential dynamical systems analogue of WCI. (The gyrostat is a rigid body with a flywheel attached along its intermediate axis.) The dynamics of the gyrostat system can also be formulated in the presence of gravity, as a heavy top with a flywheel attached. The solution behaviour of the gyrostat is close to GLM behaviour. Namely, the effect of the flywheel on the rigid body is small except near the unstable equilibria of the rigid body. However, on the slow time scale this weak effect can accumulate over time for motion along the unstable manifold of the perturbed equilibrium. The gyrostat example may even suggest some ideas about dealing with tipping points (bifurcations) in perturbed GLM systems.

Appendix \ref{appendix-A} also discusses the swinging spring, or elastic spherical pendulum, as a potential dynamical systems analogue of WCI. When a rigid spherical pendulum is made radially elastic, the possibility opens for the new oscillation degree of freedom to interact with the rotational degrees of freedom. The resulting exchange of energy can be quite dramatic if resonances between the two dynamical modes can occur \cite{HolmLynch2002}. However, the non-acceleration result for GLM implies that no such exchange of energy is generally available for GLM. 

Both the gyrostat and swinging spring also have structural similarities with WCI from the GLM viewpoint, because the Hamiltonian matrix operator in all three Hamiltonian formulations is block-diagonal.

\section{Deterministic GLM background for waves in the ocean}\label{det-GLM-bkgrnd}

Wave trains in the compound wave-current ocean flow can be excited by external forces such as the tides, as well as by the mean ponderomotive force of winds blowing along the sea-surface, or by the restoring force of buoyancy due to gravity for flows over bathymetry, or even along outcroppings in the horizontal boundaries. Then, once excited, these wave trains can propagate through the ocean, even if the ocean currents are still and calm.  This observation argues for regarding ocean wave excitations as a degree of freedom which can be distinguished from ocean currents. 

The statement of the Wave-Current Interaction (WCI) problem involves a hybrid, or compound, description in which the wave field is regarded as a separate Eulerian degree of freedom in the decomposition of the Lagrangian fluid-parcel trajectory into fast and slow components.  The generalised Lagrangian mean (GLM) fluid description \cite{AM1978} is a natural approach in this regard, because it also arises from a fast-slow dynamic decomposition of hybrid wave and current degrees of freedom which itself goes back to averaged Lagrangian methods formulated by Whitham \cite{Whitham2011}. The fast-slow averaging approach is also familiar in many other branches of physics. For example, in the guiding center and oscillation center models in plasma physics, averaging over the fast degrees of freedom for oscillation leads to ponderomotive forces of the wave envelope on the mean flow. See, e.g., \cite{Dewar1973, RGL1981, ANK-DH1984, Brizzard2009}. All of these theories in continuum mechanics separate the full flow into a composition of a slow mean flow from which rapid fluctuations depart. Their objective is to model the combined mean dynamics of the full flow at the slow time scale.


\paragraph{Comparing the GLM and CL equations.}
The Generalised Lagrangian Mean (GLM) flow theory of Andrews and McIntyre \cite{AM1978} is in principle an exact theory of nonlinear waves on a Lagrangian mean flow, within an Eulerian framework. Its potential universality has made GLM the canonical theory for investigating wave mean flow interaction (WMFI) \cite{AM1978,GH1996} or, equivalently, wave-current interaction (WCI)  \cite{Leibovich1980}.  
In certain asymptotic regimes, the GLM equations can be reduced to the Craik--Leibovich (CL) equations \cite{CraikLeibovich1976}, in particular when the wave  field is irrotational and the shear is weak \cite{Leibovich1980}. The GLM theory also affords an extension of Craik--Leibovich (CL) instability theory to admit rotational wave fields and strong shear \cite{Craik1982a,Craik1982b,Craik1985}. The CL theory of linear instability of the wave-mean flow interaction in the latter case produces longitudinal vortices which are generally expected to develop nonlinearly into Langmuir circulations \cite{Craik1982b,Thorpe2004}. 

However, regardless of these formal similarities in certain asymptotic regimes, the nonlinear  mathematical structures of the GLM equations and the CL equations derived in this paper will turn out to be quite different. This difference is not unexpected, because  the CL equations have an external forcing term which is absent in the GLM equations. Moreover, a non-acceleration result exists for the GLM equations under which the wave degree of freedom cannot influence the circulation of the current degree of freedom. These differences will emerge when the CL model is discussed in detail in sections \ref{sec:det-CL} and \ref{sec:stoch-CL}.
In fact, the differences between CL and GLM fluid dynamics reside largely in how the choice between Eulerian and Lagrangian velocity averaging affects the Kelvin circulation theorem. Eulerian averaging in the CL approach affects the Eulerian velocity 1-form in the \emph{circulation integrand} in Kelvin's theorem, while Lagrangian averaging in the GLM approach affects the material velocity of the Kelvin \emph{circulation loop}. This profound difference in how the two approaches affect circulation dynamics means that the physics of the two approaches  can differ widely.  

%

\paragraph{GLM is based on slow-fast decomposition.}
The GLM equations are based on defining fluid quantities at a
displaced, rapidly fluctuating position $\mathbf{x}^\xi := \mathbf{x}+\xi(\mathbf{x},t)$. 
In the GLM description, $\overline{\chi}$
denotes the Eulerian mean of a fluid quantity
$\chi=\overline{\chi}+\chi^{\,\prime}$ while $\overline{\chi}^L$ denotes the
\emph{Lagrangian mean} of the same quantity, defined by
\begin{equation}
\overline{\chi}^L(\mathbf{x})
\equiv
\overline{\chi^\xi(\mathbf{x})}
\,,\quad\hbox{with}\quad
\chi^\xi(\mathbf{x})
\equiv
\chi(\mathbf{x}+\xi(\mathbf{x},t))
\,.
\label{LM-def-rel}
\end{equation}
Here $\mathbf{x}^\xi\equiv\mathbf{x}+\xi(\mathbf{x},t)$ is the current
position of a Lagrangian fluid trajectory whose current mean position is 
$\mathbf{x}$. Thus, $\xi(\mathbf{x},t)$ denotes the fluctuating displacement 
 with vanishing Eulerian mean $\overline{\xi}=0$ of a Lagrangian
particle trajectory about its current mean position $\mathbf{x}$. 

GLM defines the fluid velocity at the displaced oscillating position as $ \mathbf{u}^\xi(\mathbf{x},t) := \mathbf{u} (\mathbf{x} + {\bxi}(\mathbf{x},t))$ 
where $\mathbf{x}$ is evaluated as the current position on a Lagrangian mean path and 
\begin{equation} 
\mathbf{u}^\xi := \frac{D^L}{Dt}\Big(\bx+\bxi(\bx,t)\Big) 
= \bu^L(\bx,t) + \mathbf{u}^\ell(\bx,t)
\quad\hbox{with}\quad
\frac{D^L}{Dt} = \frac{\p}{\p t}  + \mathbf{u}^L\cdot \frac{\p}{\p \bx}
\quad\hbox{and}\quad
\mathbf{u}^\ell := \frac{D^L{\bxi}}{Dt}
\,.
\label{u(l)}
\end{equation} 
One then defines the Lagrangian mean velocity as $\ob{\mathbf{u}^\xi}(\bx,t) = \bu^L(\bx,t)$, 
where $\overline{(\,\cdot\,)}$ is a time, or phase average at fixed \emph{Eulerian} coordinate $\mathbf{x}$. 

Thus, the GLM approach decomposes the Lagrangian trajectory as $\mathbf{x}^\xi := \mathbf{x}+\xi(\mathbf{x},t)$ into its current mean position, $\mathbf{x}$, plus a rapidly fluctuating displacement $\xi(\mathbf{x},t)$, then GLM investigates the Lagrangian mean dynamical implications of this decomposition. Postulating the unknown fluctuating displacement vector field $\xi(\mathbf{x},t)$ introduces an additional degree of freedom of the Lagrangian mean fluid description whose effects must be modelled. 

\paragraph{A quick derivation of the GLM motion equation by time averaging Kelvin's theorem.}
One may derive the GLM motion equation by applying Lagrangian-mean time averaging to the Kelvin circulation theorem for the ideal fluids in the form of Newton's law,  which equates the rate of change of the momentum to the force on a distribution of mass on a material loop, 
\begin{equation}
\frac{d}{dt} \overline{\oint_{c(u^\xi)} \bu(\bx^\xi,t) \cdot d\bx^\xi }^{\,L}
= \overline{\oint_{c(u^\xi)} \Big( \cdot \mathbf{f} \cdot \Big)^\xi\cdot d\bx^\xi}^{\,L}
\,,\label{Lag-det-GLM1}
\end{equation}
where $ \mathbf{f}$ denotes the sum over whatever prescribed forces per unit mass are present. 
In Kelvin's theorem \eqref{Lag-det-GLM1}, the loop moves with the flow, so the loop is a Lagrangian quantity. The integrand is fixed in space, so the integrand is Eulerian. Thus, after taking averages, the loop velocity will be the Lagrangian mean velocity, $\bu^L$, and the integrand will be given by its Eulerian mean $\overline{(\,\cdot\,)}$ at the displaced location of the Lagrangian trajectory $\mathbf{x}^\xi := \mathbf{x}+\xi(\mathbf{x},t)$. Namely, 
\begin{equation}
\frac{d}{dt}\oint_{c(u^L)}  \overline{  \bu(\bx^\xi,t)  \big)\cdot d\bx^\xi }
=
\frac{d}{dt}\oint_{c(u^L)}  \overline{ \bu^L(\bx,t) + \mathbf{u}^\ell(\bx,t) \big)\cdot d(\mathbf{x}+\xi(\mathbf{x},t)) }
=
\oint_{c(u^L)}  \overline{ \Big( \cdot \mathbf{f} \cdot  \Big)^\xi\cdot d\bx^\xi}
\label{Lag-det-GLM1}
\end{equation}
Equation \eqref{u(l)} implies an evolution equation for the fluctuating displacement $\bxi(\bx,t)$,
\begin{equation}
\frac{D^L{\xi}}{Dt} =: \bu^\ell(\bx^\xi,t)  = (\bxi\cdot\nabla) \bu^L(\bx,t) 
\Longrightarrow
\frac{\p\bxi}{\p t}  + (\mathbf{u}^L\cdot\nabla)\bxi = (\bxi\cdot\nabla) \bu^L(\bx,t) \,.
\label{Lag-det-GLM2}
\end{equation}
The last equation means the displacement vector field $\xi:=\bxi(\bx,t)\cdot\nabla$ is \emph{advected} by the Lagrangian mean velocity $\mathbf{u}^L$. That is, $\pa_t\xi + [u^L,\xi] = 0 = \pa_t\xi - {\rm ad}_{u^L} \xi $. This, in turn, means that the fluctuating vector field $\xi$ is pushed forward by the time-dependent flow $\phi^L_t$, which itself is generated by the vector field $u^L$. That is, $\xi(t)={\phi^L_t}_*\xi(0)$. Furthermore, $\pa_t \xi(t) = \pa_t {\phi^L_t}_*\xi(0) = -\, {\phi^L_t}_*\mathcal{L}_{u^L} \xi(0)=-\mathcal{L}_{u^L}\xi(t)$. For more discussion of the geometric properties of the GLM theory, see \cite{GV2018,Holm2019}.

To quadratic order in the displacement $\bxi(\bx,t)$ with zero mean $\overline{\bxi(\bx,t)}=0$ equation \eqref{Lag-det-GLM1} implies 
\begin{equation}
\frac{d}{dt}\oint_{c(u^L)}   \Big( \bu^L(\bx,t) + \overline{u^\ell_k\nabla \xi^k} \,\Big)\cdot d\bx
=
\oint_{c(u^L)}  \overline{ \Big( \cdot \mathbf{f} \cdot  \Big)^\xi\cdot d\bx^\xi}
\,.
\label{Lag-det-GLM3}
\end{equation}

At this point the Kelvin circulation theorem for the standard GLM equations may be derived by defining the \emph{pseudovelocity}, $\widetilde{\bv}(\bx,t)$, and the \emph{pseudomomentum}, $\bp(\bx,t)$, as follows
\begin{equation}
\widetilde{\bv}(\bx,t) = -\, {\overline{u^\ell_k\nabla \xi^k} } = \bp(\bx,t)/\widetilde{D}\,,
\label{Lag-det-GLM4}
\end{equation}
where $\widetilde{D}$ is the Lagrangian mean volume element. 

\paragraph{Pseudomomentum is a momentum map.}
Pseudomomentum is a 1-form density (dual to vector fields under $L^2$ pairing) which can be written in a variety of ways.
For example, pseudomomentum can be written as \cite{Holm2019}
\begin{equation}
p(\bx,t) = \bp(\bx,t)\cdot d\bx\otimes d^3x
=
\overline{u^\ell_k \,d \xi^k} \otimes d^3x
=
\overline{u^\ell_k \,\partial_\phi \xi^k}\,d\phi (\bx,t) \otimes d^3x
=:
N \,d\phi(\bx,t) \otimes d^3x\,.
\label{Lag-det-GLM5}
\end{equation}
In the next section, the last variant in \eqref{Lag-det-GLM5} will allow us to consider $\phi(\bx,t)$ and $N(\bx,t)d^3x$ as canonically conjugate wave field variables. It will follow that $p = N\,d\phi\otimes d^3x$ is a momentum map \cite{HMR1998}. This recognition will allow us to use $L^2$ pairing to couple the wave field to the Lagrangian mean fluid velocity vector field in Hamilton's principle. Upon applying this momentum map coupling, we will introduce a Hamiltonian for the wave dynamics in a phase-space Lagrangian in Hamilton's principle. This procedure will allow us close the GLM equations explicitly in terms of physically identifiable wave and fluid quantities. 

\paragraph{The approximations which reduce GLM to the CL (Craik-Leibovich) equations.}
If one introduces the Stokes mean drift velocity $\bu^S(\bx,t)$ as the difference between the Lagrangian and Eulerian mean velocities, which is defined as $\bu^S(\bx,t)=\bu^L-\bu=\overline{\xi^k\pa_k\bu^\ell}$, then equation \eqref{Lag-det-GLM3} is expressible as
\begin{equation}
\frac{d}{dt}\oint_{c(u^L)}   \Big( \big( \bu^L(\bx,t)-\bu^S(\bx,t)\big) \cdot d\bx +\overline{\mathcal{L}_\xi (\bu^\ell \cdot d\bx)}\, \Big)
=
\oint_{c(u^L)}  \overline{ \Big( \cdot \mathbf{f} \cdot  \Big)^\xi\cdot d\bx^\xi}
\,,
\label{Lag-det-CL1}
\end{equation}
where the expression
\begin{align}
\begin{split}
\overline{\mathcal{L}_\xi (\bu^\ell \cdot d\bx)} &= (\,\overline{\xi^k\pa_k\bu^\ell} + \overline{u^\ell_k\nabla \xi^k}\,)\cdot d\bx  
\\&
= \big(\,\overline{ - \bxi\times {\rm curl}\bu^\ell} + \overline{\nabla(\bxi\cdot \bu^\ell \,)} \big)\cdot d\bx  
= (\bu^S(\bx,t) - \widetilde{\bv}(\bx,t))\cdot d\bx
\end{split}
\label{Lag-det-CL2}
\end{align}
denotes the Eulerian mean $\overline{(\,\cdot\,)}$ applied to the Lie derivative  along the fluctuation vector field, $\xi=\bxi\cdot\nabla$, of the circulation 1-form of the fluctuating velocity, $\bu^\ell \cdot d\bx$. 

Upon neglecting the entire Eulerian mean fluctuation term $\overline{\mathcal{L}_\xi (\bu^\ell \cdot d\bx)}\to 0$  in equation \eqref{Lag-det-CL2} and also neglecting the time dependence of the Stokes mean drift velocity, $\bu^S(\bx,t)\to\bu^S(\bx)$, one finds the CL (Craik-Leibovich) equations. This is the sense mentioned earlier in which the GLM equations may be ``reduced'' to the CL equations. A stochastic version of the CL equations is formulated  below in sections \ref{sec:det-CL} and \ref{{sec:stoch-CL}} for the purpose of quantifying the uncertainty of their solutions.

\paragraph{Finite-dimensional GLM analogues.}
As a finite-dimensional example which has some close parallels with the GLM decomposition into currents and waves in the GLM theory one can consider the gyrostat, comprising a rigid body coupled to a flywheel. The 2D rotational effects of the flywheel on the 3D rotations of the rigid body are discussed in \ref{appendix-A}. 

In Appendix \ref{appendix-A} we also consider the finite-dimensional rotations and oscillations of an elastic spherical pendulum.  If the pendulum is only slightly elastic, then very rapid oscillations can take place, which may be negligible for small amplitude and in the absence of resonances. However, as the pendulum becomes more elastic and behaves more a like a radial spring, its oscillations and the resulting \emph{oscillation-rotation interaction} (ORI) can become an important feature of the dynamics \cite{HolmLynch2002}. For example, when resonances occur in the system, one may see regular exchanges between springing motion (oscillation) and swinging motion (rotation). Analogously, GLM introduces a new degree of wave freedom and assesses what mean effects it may have on the full fluid solution. However, the non-acceleration result for GLM discussed in section \ref{det-GLM-bkgrnd} precludes any resonant exchanges of energy between the GLM waves and currents.
  
\begin{remark}[Relation of GLM to mainstream stability methods for fluid equilibria]\rm
Fortunately, the GLM notation is also {\it standard} in the
stability analysis of fluid equilibria in the Lagrangian picture.
See, e.g., the classic works of Bernstein \cite{Bernstein-etal1958}, Frieman \&
Rotenberg \cite{FriemanRotenberg1960} and Newcomb \cite{Newcomb1962}. See Jeffrey \& Taniuti \cite{JT1964}
for a collection of reprints showing applications of this approach
in the course of controlled thermonuclear fusion research. For insightful reviews,
see Bernstein \cite{Bernstein1983}, Chandrasekhar \cite{Chandra1987} and, more recently, Hameiri
\cite{Hameiri1998}. Rather than causing confusion, this confluence of notation
encourages the transfer of ideas between traditional Lagrangian
stability analysis for fluids under perturbation and the GLM theory. Sometimes, as in the case of the 
elliptic instability, the GLM theory actually yields the nonlinear time dependent motion equations
resulting from the perturbation of the Lagrangian path, $\mathbf{x}^\xi\equiv\mathbf{x}+\xi(\mathbf{x},t)$, 
rather than merely producing the linear spectrum \cite{GH1996}. 
\end{remark}

\paragraph{Refinements of GLM.}
The GLM theory has inspired many refinements. These refinements include determination of the higher-order correction terms in the ratio of the time scales for currents and waves from a phase-averaged Hamilton's principle in Lagrangian coordinates \cite{GH1996}. This particular refinement established the noncanonical Lie-Poisson Hamiltonian formulation of GLM as the dynamics of two interpenetrating flows, with two different types of momentum, just as in \cite{HK1982} for Landau's 2-fluid theory of superfluids, \cite{London1950,Putterman1974}. 

The GLM equations at second order in an asymptotic expansion in the ratio of the time scales are called the $\mathfrak{glm}$ equations \cite{Holm1999,Holm2002a,Holm2002b}. Perhaps not surprisingly, the expression in equation \eqref{Lag-det-CL2} for the mean deviation of the GLM model from the CL model at second order is the same as the Eulerian mean of the circulation 1-form for the second order $\mathfrak{glm}$ model of fluctuation dynamics in fluid flows, as found already in equation (5.2) of \cite{Holm2002b}. Certain closures of the $\mathfrak{glm}$ models have led to a class of models of interest as computational Large-Eddy Simulations (LES) for turbulent flows \cite{FHT2001,FHT2002}. 

When a closure based on the Taylor hypothesis is imposed on the $\mathfrak{glm}$ equations for the incompressible ideal Euler fluid flow, as in equation  in equation \eqref{Lag-det-CL2}, one obtains the \emph{Euler-alpha model}, originally known as the $N$-dimensional Camassa--Holm equation \cite{HMR1998,Chen-etalPRL1998}. In the Taylor hypothesis closure of $\mathfrak{glm}$ for the Euler-alpha equations, the length-scale (alpha) is the mean correlation length of the fluctuations of the Lagrangian trajectory away from its mean. When viscosity is added in the form of momentum diffusion, one obtains the Lagrangian Averaged Navier-Stokes-alpha (LANS-alpha) turbulence model. The LANS-alpha turbulence model has been analysed deeply mathematically \cite{FHT2001,FHT2002} and its primitive equation version has been implemented successfully for global ocean circulation \cite{HHPW2008a,HHPW2008b}.

For more in-depth discussions of recent developments of GLM, see \cite{BuhlerMEM1998,Buhler2010,Buhler2014,GV2018,Holm2019,Thomas2017}. 
In particular, recent refinements of GLM include its formulation for flows on manifolds \cite{GV2018}, and its extension from deterministic to stochastic dynamics, within the Euler--Poincar'e variational framework of geometric mechanics \cite{Holm2019}. 

The present paper continues these refinements in reformulating the GLM variational principle derived in \cite{Holm2019} by introducing a phase-space Lagrangian in Hamilton's principle for the mean description of the wave field. The present result is a closed Hamiltonian theory which is shown to recover the GLM equations, and to implement the wave dynamics required for each specific application. Namely, the details of the closure for the wave physics of a given fluid application are governed by the dispersion relation for the type of wave field involved, which explicitly determines the proper Hamiltonian. Consequently, the present reformulation of GLM is potentially flexible enough to allow application to a variety of different types of waves. This flexibility is demonstrated by deriving the GLM equations explicitly for two applications. These are: internal waves in 3D rotating stratified incompressible Euler--Boussinesq flows in section \ref{StochGLM-sec}; and 1D shallow-water waves in section \ref{SW-WCI-sec}. 

\subsection{WCI for stratified EB fluids: the Generalised Lagrangian Mean (GLM)}\label{subsec-WCI-EB}

We introduce the idea of a phase-space Lagrangian for the wave components of fluid flows by applying it to study WCI in the familiar example of 3D Euler--Boussinesq (EB) fluid.  The EB fluid is a stratified, rotating, incompressible flow governed by the Euler fluid equations in the Boussinesq approximation. Here, we propose a variational formulation of WCI, with Hamilton's principle parameterised by a phase space Lagrangian which includes a wave Hamiltonian depending on the canonically conjugate phase-space variables $(q,p)=(\phi,N)$ for the collective degrees of freedom known as the wave phase field, $\phi(\bx,t)$, and its canonically conjugate momentum density, $N(\bx,t)$, which is the familiar GLM wave action density.

After computing the variational equations, we will show that choosing the wave Hamiltonian to be $H_W=-\int_\mathcal{D}N\omega(\bk)d^3x$ for $\bk=\nabla\phi(\bx,t)$ closes the GLM equations of  \cite{AM1978} and recovers the usual physical interpretations of their wave properties, including the phase dynamics in the local reference frame of the moving flow.  

\begin{remark}
One recalls that $\bp:=N\nabla \phi=: N\bk$ is called the \emph{pseudomomentum density} in the GLM theory \cite{AM1978}. Consider the functional $M_\xi(\phi,N)$ defined by the following $L^2$ pairing of the 1-form density $N\nabla \phi$ with a vector field $\bxi(\bx)$
\begin{align}
M_\xi(\phi,N)=\scp{\bxi(\bx)}{N\nabla\phi}
=\int_{\cal{D}} \bxi(\bx)\cdot N\nabla\phi\,d^3x\,.
\label{M-xi}
\end{align}
For the canonical Poisson bracket, the functional $M_\xi(\phi,N)$ in \eqref{M-xi} generates translations in space of the wave variables $\phi$ and $N$ along the characteristic curves of the vector field $\bxi(\bx)$. This can be seen by computing the canonical Poisson brackets,
\begin{align}
\begin{bmatrix}
\{\phi \,,\,M_\xi(\phi,N) \} \\
\{N\,d^3x \,,\,M_\xi(\phi,N) \}
\end{bmatrix}
=
\begin{bmatrix}
0 & 1 \\
-1 & 0
\end{bmatrix}
\begin{bmatrix}
\delta M_\xi /\delta \phi  \\
\delta M_\xi /\delta N  
\end{bmatrix}
= 
\begin{bmatrix}
\bxi\cdot\nabla \phi \\
{\rm div}(N\bxi)d^3x
\end{bmatrix}
=
\begin{bmatrix}
\mathcal{L}_\xi \phi  \\
\mathcal{L}_\xi (N\,d^3x)
\end{bmatrix}
,\label{Lie-trans}
\end{align}
\end{remark}
where $\mathcal{L}_\xi$ denotes Lie derivative with respect to the vector field $\bxi(\bx)$, which is defined as the infinitesimal transformation along the flow generated by $\bxi(\bx)$. Thus, under the Poisson bracket for the canonically conjugate, time-dependent, wave fields $\phi(\bx,t)$ and $N(\bx,t)$, the functional $M_\xi(\phi,N)$ generates a flow of the wave fields $\phi(\bx,t)$ and $N(\bx,t)$ along the characteristic curves of the vector field $\bxi(\bx)$.

\paragraph{Phase-space Lagrangian derivation of the GLM equations}
We write the WCI action integral for Hamilton's principle as the sum of the known deterministic Lagrangian for EB fluids \cite{HMR1998} coupled to a \emph{phase-space Lagrangian} (PSL) for the wave-field dynamics, as follows,
\begin{align}
\begin{split}
S = \int_{t_1}^{t_2}\ell(\bu^L,D,b,N,\phi:p)\,dt
&= \int_{t_1}^{t_2}\!\!\int_\mathcal{D} \bigg[
\frac{D}{2}\big| \bu^L \big|^2 + D\bu^L\cdot \bR(\bx) - gDbz - p(D-1) \bigg]d^3x
\\&\hspace{2cm}
- \int_{t_1}^{t_2}\!\!\int_\mathcal{D}  N(\partial_t\phi  + \bu^L\cdot\nabla \phi)\,d^3x + \int_{t_1}^{t_2} H_W(N,\bk) \,.
\end{split}
\label{Lag-det}
\end{align}
The first line of the Lagrangian in \eqref{Lag-det} is the fluid Lagrangian for EB fluids in standard vector form \cite{HMR1998}. The second line contains the PSL for the wave degrees of freedom, obtained by a partial Legendre transform $L(\phi,\partial_t{\phi})=\langle N,\partial_t\phi \rangle - H(N,\nabla\phi)$ for the canonically conjugate wave degrees of freedom $(N,\phi)$. Note that the PSL for the wave variables is manifestly invariant under translations in the phase $\phi$. This means the PSL would keep its form under phase averaging at fixed Lagrangian coordinate. Hence, one may regard the PSL as having \emph{resulted} from such an averaging process. 

The term $-\int_\mathcal{D} N\nabla \phi\cdot \bu^L\,d^3x$ in the second line has both wave and fluid components. This term serves to couple the EB Lagrangian for the fluid variables with the phase-space Lagrangian for the wave variables by pairing the wave momentum density with the fluid velocity. Equation \eqref{Lie-trans} shows that variations of this term in $(\phi,N)$ translate the wave variables along the Lagrangian trajectories of the current flow velocity $\bu^L(\bx,t)$. 

The variation of the Lagrangian in \eqref{Lag-det} with respect to the transport velocity $\bu^L(\bx,t)$ produces the total Eulerian momentum density for GLM in the presence of the wave field,
\begin{align}
\bm(\bx,t) := \frac{\delta \ell}{\delta \bu^L} = D( \bu^L + \bR(\bx)) - N\nabla \phi\,,
\label{Lag-var}
\end{align}
in which $\bp=N\nabla \phi=N\bk$ is the GLM pseudomomentum density. As we shall see, the Hamiltonian dynamics for the momentum density in equation \eqref{Lag-var} will recover the GLM velocity equation for the Lagrangian mean transport velocity $\bu^L(\bx,t)$.   

The Hamiltonian corresponding to the Lagrangian in \eqref{Lag-det} is given by the Legendre transform,
\begin{align}
\begin{split}
H(\bm,N,\phi,D,b)) &= \!\!\int_\mathcal{D}\bm\cdot\bu^L - N\p_t{\phi}\,d^3x - \ell(\bu^L,D,b,N,\phi:p)
\\&= \!\!\int_\mathcal{D} \bigg[
\frac{1}{2D}\big| \bm - N\bk - D\bR \big|^2 +  gDbz + p(D-1)  \bigg]d^3x + H_W(\phi,N)\,.
\end{split}
\label{Ham-det}
\end{align}
This is simply the sum of the material and wave energies. The variational derivatives are given by
\begin{align}
\begin{split}
\delta H(\bm,N,\phi,D,b)) &= \!\!\int_\mathcal{D} 
\bu^L\cdot \delta\bm + \delta D\Big(gbz +p - \frac12|\bu^L|^2 - \bu^L\cdot\bR(\bx)\Big) + (gDz)\delta b
\\&\hspace{2cm} - {\rm div}\bigg( N\Big(\bu^L + \frac{\delta H_W}{\delta \bk}\Big)\bigg)\delta\phi 
+ \bigg(\frac{\delta H_W}{\delta N} + \bk\cdot\bu^L\bigg)\delta N
\,d^3x\,,
\end{split}
\label{Ham-var}
\end{align}
where we have used the variational  identity $\delta H_W/\delta \phi = -\,{\rm div}(\delta H_W/\delta \bk)$ at constant $N$, which follows from  $\bk:=\nabla\phi$. 

We may write equations of motion in Hamiltonian form by 
using a block-diagonal Poisson matrix operator, as  
 \begin{align}
\partial_t
\begin{bmatrix}
m_i \\ D \\ b \\  \phi \\ N
\end{bmatrix}
= -
\begin{bmatrix}
\partial_j m_i + m_j \partial_i & D\partial_i & -b_{,i} & 0 & 0
\\
\partial_jD & 0  & 0 & 0 & 0
\\
b_{,j} & 0  & 0 & 0 & 0
\\
0 & 0 & 0 & 0 & -1
\\
0 & 0 & 0  & 1 & 0
\end{bmatrix}
\begin{bmatrix}
\delta H / \delta m_j 
\\
\delta H / \delta D  
\\
\delta H / \delta b 
\\
\delta H / \delta \phi 
\\
\delta H / \delta N 
\end{bmatrix}.
\label{m+D+b-eqns}
\end{align}

\begin{remark}[Next steps]
The key idea in proposing the action for WCI dynamics in equation \eqref{Lag-det} is the separation of the Lagrangian (current) and Eulerian (wave) degrees of freedom in Hamilton's principle. This has been accomplished by introducing a standard {\it Euler--Poincar\'e Lagrangian} for the current flow \cite{HMR1998} and a {\it phase-space Lagrangian} for the wave field. The two Lagrangians are coupled by the \emph{mechanical connection} term $-\int_\mathcal{D} N\nabla \phi\cdot \bu^L\,d^3x$ in the Lagrangian in \eqref{Lag-det} obtained by pairing the velocity of the current flow with the momentum map of the Hamiltonian wave system. Coupling by this pairing has the effect that the waves propagate in the local reference frame of the current flow. Without this connection, the current and wave degrees of freedom would evolve separately. The result is a closed dynamical theory whose wave-current dynamics can be made stochastic. To demonstrate the applicability of this hybrid approach, we first verify that our Hamilton principle does recover the \emph{deterministic} GLM equations for an Euler--Boussinesq (EB) fluid. We then add stochasticity and recover the stochastic version of the GLM equations for EB fluid derived earlier in \cite{Holm2019} which also introduces stochasticity into the GLM wave field. 
\end{remark}

\paragraph{Deterministic wave-current interaction for EB fluid flow.}
Upon expanding out the Hamiltonian equations in \eqref{m+D+b-eqns}, the dynamics of the EB fluid with these additional wave variables is found to obey the following system of equations,%
\footnote{We will derive the stochastic versions of equations \eqref{SVP3-det} in proving Theorem \ref{SALT-SNWP-Thm}. There, equations \eqref{SVP3-det} will re-emerge when the noise is absent. }
%
\begin{align}
\begin{split}
&\partial_t  \bm + (\bu^L\cdot \nabla)\, \bm + (\nabla \bu^L)^T \cdot \bm + \bm \,{\rm div}\bu^L
= 
D \nabla \pi + Dgz\nabla b
+ \bk \,{\rm div} \Big(\frac{\delta H_W}{\delta  \bk}\Big) 
+ N \nabla  \frac{\delta H_W }{\delta  N}
\,,\\&
\partial_t  D + {\rm div}(D\bu^L)  = 0
\,,\qquad
D=1
\,,\qquad
\partial_t  b + \bu^L\cdot \nabla  b = 0
\,,\\&
\partial_t  \phi + \bu^L\cdot \nabla  \phi  - \frac{\delta\, H_W}{\delta N} = 0
\,,\qquad
\partial_t  N + {\rm div}(N\bu^L) - {\rm div} \Big(\frac{\delta\, H_W}{\delta \bk}\Big) = 0
\,.\end{split}
\label{SVP3-det}
\end{align}
The Eulerian momentum density, $\bm$, and the Bernoulli function, $\pi$, in these equations are defined by the following variational derivatives of the GLM Lagrangian in \eqref{Lag-det}, 
\begin{align}
\bm := \frac{\delta  \ell}{\delta  \bu^L}  = D (\bu^L + \bR(\bx)) - N\nabla \phi
\,,\qquad
\pi :=  \frac{\delta  \ell}{\delta  D} = \frac12 |\bu^L|^2 + \bR\cdot \bu^L - gbz - p
\,.
\label{m&pi-defs}
\end{align}
The motion equation for WCI in equation \eqref{SVP3-det} implies the following Kelvin circulation dynamics for the Eulerian momentum per unit mass,
\begin{align}
\begin{split}
\frac{d}{dt} 
\oint_{c(u^L)} \frac1D\frac{\delta  \ell}{\delta  \bu^L}\cdot d\bx 
&=
\oint_{c(u^L)} (\p_t+\mathcal{L}_u^L)\bigg(\Big(\bu^L + \bR(\bx) - \frac{N}{D}\nabla \phi \Big)\cdot d\bx \bigg)
\\&\hspace{-3cm}= 
\oint_{c(u^L)} \nabla\pi \cdot d\bx
+
\oint_{c(u^L)} \underbrace{\
 gz \nabla b \cdot d\bx \
 }_{\hbox{Buoyancy}}
+ 
\oint_{c(u^L)}
\underbrace{\
\frac{1}{D}  \bigg(\bk \,{\rm div} \Big(\frac{\delta H_W}{\delta  \bk}\Big) 
+ 
N \nabla  \frac{\delta H_W }{\delta  N}\bigg) \
}_{\hbox{Wave Forcing}}\hspace{-1mm}
\cdot \,d\bx
\,.
\end{split}
\label{Kelvin-GLM}
\end{align}
Thus, the wave forcing terms in \eqref{Kelvin-GLM} could potentially generate circulation of the total Eulerian momentum per unit mass, $\bm/D$, which is dual to the Lagrangian mean velocity, $\bu^L$. Equation \eqref{Kelvin-GLM} is Newton's $2^{\rm nd}$ Law for the time rate of change of the total Eulerian mean momentum $\bm={\delta  \ell}/{\delta  \bu^L}$ of a body whose mass is distributed  on a closed loop $c(u^L)$ moving with the Lagrangian mean velocity $\bu^L$. According to equation \eqref{Kelvin-GLM}, the wave force in Newton's Law for this model depends on the following wave properties: wave action density $N$; wave vector $\bk$; gradient of dispersion relation $\nabla\omega(\bk)$;  and group velocity $\bv_G(\bk):=\p\omega/\p\bk$. 

The solution algorithm for solving the Euler-Poincar\'e system \eqref{SVP3-det} is, as follows. First, one solves the Euler-Poincar\'e motion equation to update the total Eulerian momentum density $\bm(\bx,t)$. In parallel, one updates the solutions of the four auxiliary equations for $(\phi,N)$ and $(b,D)$.  Updating $(\phi,N)$ also updates the GLM pseudomomentum density, $\bp:=N\nabla \phi$. Next, one solves the linear equation \eqref{Lag-var} to update the Lagrangian mean transport velocity in the rotating frame, $\bu^L = D^{-1}(\bm + N\nabla \phi)-\bR(\bx)$. Finally, one updates the pressure $p$ by solving a Poisson equation with Neumann boundary conditions. After these steps, one may then take the next time step in the motion equation and iterate the solution algorithm. The solution algorithm solves for the total Eulerian momentum density, $\bm(\bx,t)$, and the pseudomomentum density, $\bp:=N\nabla \phi$, independently. This separation is crucial in solving for the Lagrangian mean transport velocity, $\bu^L$, which is a diagnostic variable in this solution algorithm.  We will see that the \emph{kinematic momentum}, $D\bu^L,$ can be made prognostic instead of $\bm$ by making a change of variables in the Hamiltonian formulation of the present model. The resulting prognostic equation for $\bu^L$ turns out to be exactly the original EB motion equation, which is as it should be. In addition, though, upon completing the Hamiltonian formulation, one recovers the Poisson bracket for Landau's 2-fluid model for superfluid $^4He$. This places the present theory of WCI into the class of complex fluids with dynamical order parameters \cite{Holm2002}

\begin{example}[The WKB wave Hamiltonian]\label{SeparatedWaveHam}
Closure of the WCI system \eqref{SVP3-det} requires one to model the wave Hamiltonian, $H_W(\bk,N)$.
A compelling choice of the wave Hamiltonian can be recognised by defining the second line of equation \eqref{Lag-det} as
\begin{align}
- \int_{t_1}^{t_2}\!\!\int_\mathcal{D}  N(\partial_t\phi  + \bu^L\cdot\nabla \phi)\,d^3x + \int_{t_1}^{t_2} H_W(N,\bk) 
=
- \int_{t_1}^{t_2}\!\!\int_\mathcal{D} N\big(\partial_t\phi  + \bu^L\cdot\nabla \phi + \omega(k)\big)\,d^3x
\,.
\label{Lag-det1}
\end{align}
That is, the wave Hamiltonian is determined by regarding $N$ in the phase-space Lagrangian in \eqref{Lag-det} as a \emph{Lagrange multiplier} which enforces the constraint that $\phi$ satisfies the WKB wave phase equation \eqref{Lag-det1} in the local reference frame of the moving fluid. In this formulation, one may immediately interpret the physical meanings of the various wave terms. Namely,
\begin{align}
H_W = - \int_M N \omega(\bk) \,d^3x 
\,,\quad\hbox{with}\quad
 \frac{\delta H_W }{\delta  N}\Big|_{\bk} = - \,\omega(\bk)
 \,,\quad\hbox{and}\quad
 \frac{\delta H_W}{\delta  \bk}\Big|_{N} = - N \frac{\partial \omega(\bk) }{\partial \bk} =: -\,N \bv_G(\bk)
\,,\label{separatedWaveHam}
\end{align}
in which $\bv_G(\bk):=\partial \omega(\bk) / \partial \bk$ is the group velocity for the dispersion relation $\omega=\omega(\bk)$ between wave frequency, $\omega$, and wave number, $\bk$. For the wave Hamiltonian $H_W$, the wave variables in equation \eqref{SVP3-det} obey the following familiar \emph{WKB} relations  in the frame of the fluid motion,\footnote{The gradient of the phase dynamics in equation \eqref{separatedWaveDyn} yields 
the following equation for the wave vector, $\bk=\nabla \phi$, in the local reference frame of the fluid motion,
\[
\big(\partial_t + \mathcal{L}_{\bu^L+\bv_G(\bk)}\big)d\phi
= 0
= \Big(\partial_t\bk + \nabla \big( \omega(\bk) + \bk\cdot \bu^L\big)\Big)\cdot d\bx
\,, \quad\hbox{since}\quad{\rm curl\,}\bk=0
. 
\]}
\begin{align}
\partial_t  \phi + \bu^L\cdot \nabla  \phi  = -\,\omega(\bk)
\,,\qquad
\partial_t  N + {\rm div}\Big(N(\bu^L+ \bv_G(\bk)\Big)  = 0
\,.\label{separatedWaveDyn}
\end{align}
The wave dynamics in \eqref{separatedWaveDyn} may also be written in a suggestive canonical Hamiltonian form in the reference frame of the fluid motion as
\begin{align}
(\p_t + \mathcal{L}_{u^L})
\begin{bmatrix}\,
\phi \\ N d^3x
\end{bmatrix}
= 
\begin{bmatrix}\,
0 & 1
\\
-1 & 0
\end{bmatrix}
\begin{bmatrix}\,
\frac{\delta H_W}{\delta \phi} \\ \frac{\delta H_W}{\delta N}
\end{bmatrix}\,,
\label{canonical-form-GLM}
\end{align}
where $\mathcal{L}_{u^L}$ denotes Lie derivative with respect to the transport vector field which in GLM theory is the Lagrangian mean velocity $\bu^L$. In this form, one may identify the operation $(\p_t + \mathcal{L}_{u^L})$ as the time derivative in the frame of the moving fluid. 
The wave dynamics in equations \eqref{separatedWaveDyn} and \eqref{canonical-form-GLM} provide further insight into wave propagation in a fluid flow.  For a constant transport velocity, $\bu^L$, the equations in \eqref{separatedWaveDyn} may be immediately recognised as the WKB equations for a wave packet with slowly varying envelope propagating in a moving medium \cite{Peregine1976,Whitham2011,Buhler2014}. However, equation \eqref{Kelvin-GLM} for the total circulation dynamics raises the issue of whether wave motions may affect the Eulerian momentum per unit mass of the fluid parcels.

The two equations in \eqref{separatedWaveDyn} imply that the fluid velocity $\bu^L$ transports the wave propagation dynamics in the reference frame of the fluid flow.  Combining these two equations yields 
\begin{align}
\begin{split}
- \,\frac{d}{dt} \oint_{c(u^L)}  \frac{N}{D}\nabla \phi \cdot d\bx 
&= \oint_{c(u^L)} \frac{1}{D}  \bigg(\bk \,{\rm div}\Big(N\bv_G(\bk)\Big)
+ 
N \nabla  \omega(k)\bigg)\cdot d\bx
\\&= \oint_{c(u^L)} \frac{1}{D}  \bigg(\bk \,{\rm div} \Big(\frac{\delta H_W}{\delta  \bk}\Big) 
+ 
N \nabla  \frac{\delta H_W }{\delta  N}\bigg)\cdot d\bx
\,.
\end{split}
\label{SALT-SNWP-GLM-wave}
\end{align}
That is, the two equations in \eqref{separatedWaveDyn} imply  
\begin{align}
-\Big(\p_t+\mathcal{L}_{u^L}\Big)\bigg( \frac{N}{D}\nabla \phi \cdot d\bx \bigg)
= \frac{1}{D}  \bigg(\bk \,{\rm div}\Big(N\bv_G(\bk)\Big)
+ 
N \nabla  \omega(k)\bigg)\cdot d\bx
\,.\label{non-accel}
\end{align}
Now substituting equation \eqref{SALT-SNWP-GLM-wave} with wave Hamiltonian \eqref{separatedWaveHam} into the circulation theorem in equation \eqref{Kelvin-GLM} produces a cancellation which recovers the Kelvin circulation theorem in the \emph{same form as for the original EB equation} 
\begin{align}
\frac{d}{dt} 
\oint_{c(u^L)} \Big(\bu^L + \bR(\bx)  \Big)\cdot d\bx
= \oint_{c(u^L)}  (\nabla\pi + gz \nabla b) \cdot d\bx 
\,.
\label{SALT-SNWP-GLM-total}
\end{align}
Equations \eqref{Kelvin-GLM}, \eqref{SALT-SNWP-GLM-wave} and \eqref{SALT-SNWP-GLM-total} provide an additive decomposition the Kelvin circulation theorem representation of WCI in the example of EB flow. This result proves a dynamical version of the famous \emph{non-acceleration theorem} for WMFI \cite{Vallis2017}. That is, \eqref{SALT-SNWP-GLM-total} implies a non-acceleration theorem for the present GLM WCI model in the example of incompressible 3D EB flow with stratification and rotation, in the sense that the equation for the mean flow velocity in this model does not change, even when waves are present. In particular, the fluid potential vorticity (PV) will still be conserved on Lagrangian mean particle paths. That is,
\begin{align} 
\partial_t Q + \bu^L\cdot \nabla Q = 0\,,
\label{GLM-PV}
\end{align}
where PV is defined as $Q := D^{-1} \nabla b \cdot {\rm curl (\bu^L + \bR(\bx) )}$ with $D=1$.

Thus, after identifying the wave Hamiltonian in \eqref{separatedWaveHam}, the phase-space Lagrangian in \eqref{Lag-det1} has produced a model of wave-current interaction in the EB fluid in which the total circulation separates additively into wave and current components. In particular, the total momentum density in the model decomposes as $\bm = D (\bu^L + \bR(\bx)) - N\nabla \phi$ into the \emph{sum} of the momentum densities for the two degrees of freedom. However, in the absence of dissipation no momentum exchange occurs between the two interpenetrating fluids, as also occurs for superfluid ${}^4He$ \cite{London1950,Putterman1974, HolmHVBK2001}. Next, we will discuss how this variational description of WCI fits into the vast literature of wave mean flow interaction \cite{Peregine1976,Whitham2011,Buhler2014}.  

The next example will show that equations \eqref{SVP3-det} with $\bm$ given in \eqref{m&pi-defs} and $H_W(N,\bk)$ given in \eqref{separatedWaveHam} provide a \emph{closure} for the GLM dynamics of the EB stratified, rotating, incompressible fluid. 
\end{example}

\begin{example}[Comparing WCI equations \eqref{SVP3-det} with the Andrews and McIntyre GLM formulation \cite{AM1978}]\label{ID-GLMvar}$\,$\smallskip

For the choice of wave Hamiltonian \eqref{separatedWaveHam} of Example \ref{SeparatedWaveHam}, one may quite easily identify the WCI terms in \eqref{Kelvin-GLM} which correspond to the GLM formulation. In the GLM notation  \cite{AM1978,Holm2019}, the wave variables involve the Eulerian time mean correlations denoted as $\ob{(\,\cdot\,)}$ among terms involving  the fluctuation displacements $ \xi^i(\bx,t)$ and the   fluctuation pressure $p^\xi(\bx,t)$. These wave mean variables for GLM include the relative group velocity $\ob{v}_G^j= \overline{(p^\xi K^j_i \,\p_\phi\xi^i)}$ and an approximation of the kinematic fluctuation pressure, $-\,\ob{\pi^\ell}\approx \ob{ p^\xi_{,j}K^j_i\xi^i  }$. Here $K^j_i$ is the cofactor of the Jacobian for the fluctuating flow
\[
K^j_k :={\cal J}({\cal J}^{-1})^j_k
\quad\hbox{with}\quad
{\cal J}^k_j := \frac{\partial \big(x^k+\xi^k(\bx,t)\big) }{ \partial\,x^j}
\,,\quad\hbox{whose determinant is ${\cal J}$.}
\]
The GLM fluctuation quantities are related to the GLM time-mean wave variables $N$ and $\bp$ as
\begin{align}
\ob{(\varpi_k \partial_\phi \xi^k)} =: - N\,,\quad
\ob{(\varpi_k \nabla \xi^k)} = - N\bk =: -\, \bp\,,
\label{GLM-defs}
\end{align}
where the co-vector $\boldsymbol{\varpi} \in \mathbb{R}^3$ with components $\varpi_k$, $k=1,2,3,$ is the fluctuation momentum variable. The same variables $N$ and $\bk=\nabla\phi$ also appear in the canonical wave equations \eqref{separatedWaveDyn} for the present WCI theory.

Let us make the change of variables for the variational partial derivatives of the wave Hamiltonian in \eqref{separatedWaveHam} from dependence $H_W(\bk,N)$ to $H_W(\bp/N,N)$. We then find that the canonical Hamiltonian system for $\phi,N$ in \eqref{SVP3-det}, when rewritten in equations \eqref{separatedWaveDyn} for the WKB wave Hamiltonian $H_W$, transforms under $\bk=\nabla \phi = \bp/N$ into a Lie--Poisson Hamiltonian system which exactly recovers the standard GLM equations for EB flow. The wave field's semidirect-product Lie--Poisson Hamiltonian structure is revealed in its matrix form as  \cite{Holm2019}
\begin{align}
\p_t\!
\begin{bmatrix}\,
{p}_j \\ N 
\end{bmatrix}
= - 
   \begin{bmatrix}
   {p}_k\partial_j + \partial_k {p}_j &  N\partial_j 
   \\
   \partial_k N & 0 &  
   \end{bmatrix}
   \begin{bmatrix}
\frac{\delta H_W}{\delta {p}_k}\Big|_N = {u}^{L\,k} + {v}_G^k(\bp/N) \\
\frac{\delta H_W}{\delta N}\Big|_\bp =  \omega(\bp/N) - N^{-1}p_i \big({u}^{L\,i} + {v}_G^i(\bp/N)\big)
\end{bmatrix}
\,.
  \label{Waves-SD-LPmatrix}
\end{align}
This calculation proves the following proposition.
\begin{proposition}\label{WKB-WCI-GLM}\rm
When the wave Hamiltonian in \eqref{separatedWaveHam} is rewritten in Lie--Poisson dynamical variables in \eqref{Waves-SD-LPmatrix} as $H_W(\bp / N,N) = - \int_M N \,\omega(\bp/N) \,d^3x $, then the Euler--Poincar\'e equations in \eqref{SVP3-det} and the relationships in \eqref{separatedWaveHam} for its variational partial derivatives provide a closure for the standard equations for EB flow in the GLM representation \cite{GH1996,Holm2019}. 
\end{proposition}

\begin{remark}[WCI relation to GLM]
The result in proposition \ref{WKB-WCI-GLM} that the WCI equations for the WKB Hamiltonian in \eqref{separatedWaveHam} provide a \emph{closure} for the standard GLM equations for EB fluid motion reflects a certain equivalence among WKB, WCI and GLM. These three approaches are all deeply connected with fast-slow decompositions and the time averaging of variational principles at fixed Eulerian position. For a recent discussion of these links and their history, see \cite{BurbyRuiz2019}. 

Thus, the phase-space Lagrangian \eqref{Lag-det} for the case $H_W(\bp / N,N)$ in separated form \eqref{separatedWaveHam} provides a Hamiltonian closure for deterministic GLM. In turn, the discovery here of the \emph{exact relation} of the WCI equations \eqref{SVP3-det} to the EB fluid GLM equations for the choice of WKB Hamiltonian in \eqref{separatedWaveHam} means that the order parameters $(\phi,N)$ introduced in the phase space Lagrangian \eqref{Lag-det} may be regarded as variables obtained from the time average of a fast-slow WMFI decomposition for \emph{whatever} fluid theory is under consideration. That is, we may regard the WCI equations with the WKB Hamiltonianin \eqref{separatedWaveHam} for any Euler--Poincar\'e fluid theory as a shortcut approach for deriving the form of the corresponding GLM equations from the linear dispersion relation for that theory. We will see another example of this approach for shallow water waves in section \ref{SW-WCI-sec}.
\end{remark}

\begin{remark}[Next steps: dynamics of  uncertainty in WCI]
The remainder of the main text will discuss the \emph{dynamics of  uncertainty in WCI}, as represented by stochasticity in this variational framework in hybrid wave-current variables for GLM. The investigation of the effects of random waves on the dispersion of fluid particles is a feature of modern research  \cite{Holmes-Cerfon-etal-2011}. We hope that a hybrid variational formulation of stochastic WCI associated with a Hamiltonian closure for GLM will be interesting and useful, as well. 
\end{remark}

\end{example}

%
%
%
%

\section{Variational principles for stochastic fluid dynamics (SALT)}\label{SALTrev-sec}

The variational principles for stochastic fluid dynamics derived in \cite{Holm2015} have come to be known as \emph{stochastic advection by Lie transport}, abbreviated as SALT \cite{CCHPS2018,CCHPS2019a,CCHPS2019b}. In the SALT fluid equations, the Lagrangian parcels move along Stratonovich stochastic paths, on which the Kelvin circulation theorem still holds for closed material circulation loops. 

The variational equations for  stochastic fluids introduced in \cite{Holm2015} showed that such equations arise from Hamilton's principle for the following action integral, in which the advected quantities are constrained to move along the Stratonovich stochastic paths:
from a \emph{stochastically constrained} variational principle $\delta S = 0$, with action, $S$, given by%
\footnote{Sections \ref{SALTrev-sec} and \ref{SNWP-sec} needn't be restricted to considering only GLM. Consequently, we will drop the superscript $^L$ in these sections.}
\begin{align}
S(u,a,b)  &=
\int_{t_1}^{t_2} \bigg( \ell(u,a) dt 
+ \left\langle  b\,,\,{\color{red}\rm d} a + \mathcal{L}_{{\color{red}\rm d} x_t} a\,\right\rangle_V \bigg)
\,,\label{SVP1}
\end{align}
where $\ell(u,a)$ is the unperturbed deterministic fluid Lagrangian, written as a functional of velocity vector field, $u$, and advected quantities, $a$. The stochastic dynamics of the advected quantities imposes a constrant on the variations known as a \emph{driving martingale relation} in which the operation ${\color{red}\rm d}$ in \eqref{SVP1} may be regarded as a \emph{stochastic differential}. 
For more discussion of this notation and the concept of driving martingales, see \cite{SC2020}.

The angle brackets in
\begin{equation}
\langle \,b\,,\,a\,\rangle_V:=\int < b(x),a(x,t) > dx
\label{L2pairing}
\end{equation}
denote the spatial $L^2$ integral  over the domain of flow of the pairing $<b\,,\,a>$ between elements $a\in V$ and their dual elements $b\in V^*$. In \eqref{SVP1}, the quantity $b\in V^*$ is a Lagrange multiplier and $\mathcal{L}_{{\color{red}{\rm d}}x_t}a$ is the \emph{Lie derivative} of an {advected quantity} $a\in V$, along a vector field ${\color{red}\rm d}x_t\in\mathfrak{X}$ defined by the following sum of a drift velocity $u(x,t)$ and Stratonovich stochastic process with \emph{cylindrical noise} parameterised by spatial position $x$, \cite{Pa2007,Sc1988} 
\begin{equation}
{\color{red}\rm d} x_t(x) = u(x,t)\,dt + \sum_i \xi_i(x)\circ dW_i(t)
\,.\label{vel-vdt}
\end{equation}
\begin{remark}\rm
The quantity ${\color{red}\rm d} x_t(x)$ in \eqref{vel-vdt} may be regarded as a stochastic Eulerian vector field parameterised by the spatial position $x$ which generates a smooth invertible map in space whose parameterisation in time is stochastic. In integral form, the operation the expression ${\color{red}\rm d} x_t$ in equation \eqref{vel-vdt} represents,%
\footnote{The usual superscript $\omega$ for pathwise stochastic quantities will be understood throughout. However, this superscript will be suppressed for the sake of cleaner notation.} 
\begin{equation}
x_t = x_0 + \int_0^t u(x,t)\,dt + \sum_i \xi_i(x)\circ dW_i(t)\,.
\label{L2pairing}
\end{equation}
\end{remark}

We also will find it useful to define a map called the \emph{diamond operation} $\diamond: V^*\times V\to \mathfrak{X}^*$, as follows \cite{HMR1998}.

\begin{definition}[The diamond operation]\label{diamond-def}\rm
Let  $\mathfrak{X}^*(M)$ denote the space of (smooth) 1-form densities dual to the space of (smooth)  vector fields, $\mathfrak{X}(M)$, with respect to the $L^2$ pairing on a manifold $M$.
On the manifold $M$, the diamond operation $\diamond: V^*\times V\to \mathfrak{X}^*$ is defined for a vector space $V$ with $(b,a)\in V^*\times V$ and vector field $\xi\in \mathfrak{X}$ is given in terms of the Lie-derivative operation $\mathcal{L}_\xi$ by 
\begin{equation}
\Scp{  b\diamond a }{\xi }_\mathfrak{X}
:=
\Scp { b }{- \mathcal{L}_\xi a}_V 
\label{diamond-def}
\end{equation}
for the $L^2$ pairings $\langle  \,\cdot\,,\,\cdot\,\rangle_V: V^*\times V\to \mathbb{R}$ and
$\langle  \,\cdot\,,\,\cdot\,\rangle_\mathfrak{X}: \mathfrak{X}^*\times \mathfrak{X}\to \mathbb{R}$ with $b\diamond a\in \mathfrak{X}^*$. 
\end{definition}

\begin{theorem}[SALT dynamics via Hamilton's principle $\delta S = 0$ for action integral  \eqref{SVP1},  \cite{Holm2015}]$\,$

The SPDEs which result from the stochastically constrained variational principle $\delta S = 0$ for $S$ defined in \eqref{SVP1} were expressed in Stratonovich form in terms of the Lie-derivative operation $\mathcal{L}_{{\color{red}{\rm d}}x_t}$ as 
\begin{align}
{\color{red}\rm d}  \frac{\delta  \ell}{\delta  u} 
+ \mathcal{L}_{{\color{red}\rm d} x_t} \frac{\delta  \ell}{\delta  u}
= \frac{\delta  \ell}{\delta  a}\diamond a\,dt
\,,\quad\hbox{and}\quad
{\color{red}\rm d} a + \mathcal{L}_{{\color{red}\rm d} x_t}a 
= 0
\,,
\label{FSPDEs-Stratonovich}
\end{align}
in which ${\color{red}\rm d}x_t$ is the stochastic Eulerian vector field in equation \eqref{vel-vdt} which generates the Stratonovich stochastic Lagrangian fluid `trajectory' (flow map). 

\end{theorem}

All fluid theories advect mass, whose density $D=\rho \,d^3x$ satisfies the following Stratonovich stochastic continuity equation,
\begin{align}
({\color{red}\rm d}  + \mathcal{L}_{{\color{red}\rm d} x_t}) D = \big({\color{red}\rm d}\rho + {\rm div}(\rho\, {\color{red}\rm d} \bx_t)\big)\,d^3x = 0\,.
\label{stoch-mass-advect}
\end{align}
The motion and advection equations in \eqref{FSPDEs-Stratonovich} and the continuity equation \eqref{stoch-mass-advect} imply the following Kelvin circulation theorem. 

\begin{corollary}[Kelvin circulation theorem for SALT dynamics]
The SALT dynamics equations \eqref{FSPDEs-Stratonovich} imply
\begin{align}
{\color{red}\rm d}\oint_{c({\color{red}\rm d} x_t)}\frac{1}{D} \frac{\delta  \ell}{\delta  u} 
=
\oint_{c({\color{red}\rm d} x_t)}
\big({\color{red}\rm d} + \mathcal{L}_{{\color{red}\rm d} x_t}\big)  \left(  \frac{1}{D} \frac{\delta  \ell}{\delta  u} \right)
= \oint_{c({\color{red}\rm d} x_t)}  \left(  \frac{1}{D} \frac{\delta  \ell}{\delta  a}\diamond a  \right) \,dt
\,,
\label{SALT-Kelvin}
\end{align}

\end{corollary}

\begin{remark}[Creation of fluid circulation by advection]$\,$\\
The first step in the proof of the fluid circulation equation \eqref{SALT-Kelvin} invokes the Kunita-It\^o-Wentzell theorem whose use in the derivation of stochastic fluid dynamics is discussed in  \cite{BdLHLT2020}.  
Equation \eqref{SALT-Kelvin} extends the familiar statement that fluid circulation can be created by the dynamics of the advected fluid quantities into the realm of fluid circulation on stochastically moving material loops.
\end{remark}

Equations \eqref{vel-vdt} and \eqref{FSPDEs-Stratonovich} have already been applied with good effect for uncertainty quantification and data assimilation resulting in reduction of uncertainty by using particle filtering in several exemplar problems \cite{CCHPS2018,CCHPS2019a,CCHPS2019b}. Future steps will turn toward oceanic applications of SALT for stochastic upper ocean dynamics (STUOD). However, the upper ocean dynamics has an added feature which cannot be addressed with the present SALT theory. Namely, upper ocean dynamics depends strongly on wave-current interaction (WCI). Historically, WCI has been a fundamental issue in ocean physics itself, not to mention its effect on modelling uncertainty and its potential complications in data assimilation. Now, the outstanding problem for STUOD is, ``How to extend the SALT approach to accommodate WCI?'' 

Naturally, to face this issue, we must return to basics. The first question might be, ``How can we extend the stochastically constrained variational principle $\delta S = 0$ for SALT, with action $S$ given in equation \eqref{SVP1} to accommodate WCI?'' For this, one would need to introduce a wave degree of freedom which would allow some of the fluid variables to propagate relative to the fluid flow, rather than being simply advected. Moreover, we must ask, ``Would those wave variables have their own type of stochasticity which would be independent of stochastic advection?'' and "Would one be able to represent the uncertainty in wave dynamics through stochastic propagation?'' Fortunately, the derivation obtained by using the phase-space Lagrangian for the wave dynamics discussed in Example \ref{SVP1} of the deterministic GLM equations for the EB fluid has already revealed a potential pathway to a theory of stochastic WCI. Namely, one may be able to extend the variational principle for SALT to include a \emph{stochastic phase-space Lagrangian} for the wave variables. 

This extension will be the aim of much of the remainder of the paper. We know that SALT will affect the wave motion, because the waves propagate in the frame of the stochastic fluid motion. And we know that the evolution of the waves will not affect the circulation of the fluid. However, we would also like to know how uncertainty in the wave propagation itself might affect the uncertainty of the fluid flow. The next section lays out the general formulation. Then, in the last section, we conclude by rederiving the stochastic GLM equations of \cite{Holm2019} for stratified EB fluids by using the WCI approach of the previous section, augmented by allowing the Hamiltonian for the wave variables in the phase-space Lagrangian to be stochastic.


\section{Including stochastic nonlinear wave propagation (SNWP) for WCI} \label{SNWP-sec}

To extend the stochastically constrained variational principle \eqref{SVP1} for SALT in order to accommodate the effects of SNWP on WCI, we will need additional wave variables and an additional constraint amongst them which will correspond to stochastic nonlinear wave propagation. For this purpose, we shall work in the abstract framework sketched in the introduction to introduce a set of canonically conjugate \emph{wave variables} denoted $(q,p)$ and propose the following \emph{minimal coupling form} of the wave action integral to the SALT action integral for fluid flow  in \eqref{SVP1},
\begin{align}
\begin{split}
S(u,a,b,q,p)  =
\int_{t_1}^{t_2} & \hspace{-2mm}
\underbrace{\
\ell(u,a) dt\ 
}_{\hbox{Fluid Lagrangian}}
+ \int_{t_1}^{t_2} \hspace{-2mm}
\underbrace{\
\Scp{  b }{{{\color{red}\rm d} a} + \mathcal{L}_{{\color{red}\rm d}x_t} a}_V\ 
}_{\hbox{Advection Constraint}} \!\!\!\! dt
\\&
- \int_{t_1}^{t_2} \hspace{-3mm}
\underbrace{
\Scp{p\diamond q}{{\color{red}\rm d} x_t}_\mathfrak{X}\
}_{\hbox{Minimal Coupling}} \!\!\!\! dt
+ \int_{t_1}^{t_2} \hspace{-3mm}
\underbrace{
\Scp{ p }{{\color{red}\rm d} q}_V  -  {\color{red}\rm d} \mathcal{J}(q,p) \
}_{\hbox{Phase-space Wave Lagrangian}} \hspace{-0.7cm}dt
 \hspace{-2mm}
\,.\end{split}
\label{SVP2.1}
\end{align}
Here, the angle brackets represent real  $L^2$ integral pairing, ${\color{red}\rm d} x_t$ is given in \eqref{vel-vdt} and the stochastic Hamiltonian functional ${\color{red}\rm d}  \mathcal{J}(q,p)$ for SNWP is given in Stratonovich form by 
\begin{align}
{\color{red}\rm d}  \mathcal{J}(q,p) := 
\mathcal{H}(q,p) \,dt
+ \mathcal{K}(q,p) \circ dB_t\,,
\label{J(q,p)-def}
\end{align} 

\begin{remark}[\emph{Minimal coupling} -- an idea from quantum mechanics]
Both the minimum coupling idea and the phase-space Lagrangian in the action integral \eqref{SVP2.1} were introduced by Paul Dirac in the early days of quantum mechanics. Minimum coupling is sometimes called `jay-dot-ay' ($\mathbf{J}\cdot\mathbf{A}$) coupling because of its interpretation in coupling a solution $\psi$ of the Schr\"odinger equation for a quantum charged particle such as an electron with charge $e$ and current density $\mathbf{J}=e\Im(\psi^*\nabla\psi)$ to Maxwell's equations for an electromagnetic field with vector potential $\mathbf{A}$. This type of coupling is still invoked universally in quantum problems today. For example, one may see $\mathbf{J}\cdot\mathbf{A}$ used for coupling classical nuclei to quantum electrons in the quantum hydrodynamic theory of molecular chemistry, \cite{FHT2019}. Thus, it may be no surprise that minimal coupling might arise here again, as a natural approach for coupling the Lagrangian mean flow of fluid trajectories to the essentially Eulerian field properties of wave propagation. In fact, as for the deterministic WCI case, the minimal coupling term will add the wave momentum map $p\diamond q$ to the total momentum, $m={\delta  \ell (u,a)}/{\delta  u}$. As one might expect, the minimal coupling term will also boost the wave dynamics into the frame of the stochastic fluid flow.  
\end{remark}

Applying the definition of the diamond operation $(\diamond)$ in \eqref{diamond-def} to the minimal coupling term in the action integral with the phase-space Lagrangian in \eqref{SVP2.1} 
yields 
\begin{equation}
-\,\Scp{p\diamond q}{{\color{red}\rm d} x_t}_\mathfrak{X}
:=
\Scp{ p }{\mathcal{L}_{{\color{red}\rm d}x_t} q}_V 
\label{diamond-def-coupling}
\end{equation}
for $(p,q)\in T^*V$ and the Stratonovich stochastic vector field ${\color{red}\rm d} x_t \in \mathfrak{X}$ given in \eqref{vel-vdt} and appearing in the Lie-derivative operation $\mathcal{L}_{{\color{red}\rm d} x_t }$ for the advection constraint in the action integral  \eqref{SVP2.1}.

As a consequence of equation \eqref{diamond-def-coupling}, the minimum coupling term in the action integral \eqref{SVP2.1} may be absorbed into the phase-space Lagrangian in \eqref{SVP2.1}, as
\begin{align}
\begin{split}
S(u,a,b,q,p)  =
\int_{t_1}^{t_2} & \hspace{-2mm}
\underbrace{\
\ell(u,a) dt\ 
}_{\hbox{Fluid Lagrangian}}
+ \int_{t_1}^{t_2} \hspace{-2mm}
\underbrace{\
\Scp{  b }{{{\color{red}\rm d} a} + \mathcal{L}_{{\color{red}\rm d}x_t} a}_V\ 
}_{\hbox{Advection Constraint}} \!\!\!\! dt
\\&+ \int_{t_1}^{t_2} 
\underbrace{\ 
\Scp{  p }{ {\color{red}\rm d} q + \mathcal{L}_{{\color{red}\rm d} x_t} q }_V 
-  
\Big(\mathcal{H}(q,p) \,dt + \mathcal{K}(q,p) \circ dB_t \Big)
}_{\hbox{Stochastic Legendre Transformation in Fluid Frame}}
\,,\end{split}
\label{SVP2.2}
\end{align}
where its purpose now is to represent time derivatives along the flow of the Lagrangian trajectories of the stochastic mean flow map $\phi_t$ generated by the Stratonovich stochastic vector field ${\color{red}\rm d} x_t \in \mathfrak{X}$ given in \eqref{vel-vdt}. This statement may be proved by recalling that the pullback $\phi_t^*q_t$ of a time dependent quantity (e.g. the state variable $q_t$) by the stochastic time-dependent map $\phi_t$ generated by the stochastic vector field ${\color{red}\rm d} x_t$ satisfies the stochastic differential relation \cite{BdLHLT2020}
\begin{align}
{\color{red}\rm d}(\phi_t^*q_t) = \phi_t^*\Big({\color{red}\rm d} q + \mathcal{L}_{{\color{red}\rm d} x_t} q\Big)
\,.
\label{SVP2..2}
\end{align}
For more discussion of the mathematics of stochastic geometric mechanics, see \cite{BdLHLT2020}. 

The stochastic dynamics associated with Hamilton's principle for the action integral for the stochastic phase-space Hamiltonian in \eqref{SVP2.2} may be encapsulated in the following theorem. 

\begin{theorem}[SALT + SNWP dynamics via Hamilton's principle $\delta S = 0$ for action integral  \eqref{SVP2.2}]\label{SALT-SNWP-Thm}$\,$

The extension of SALT to include SNWP is governed by the following system of Euler-Poincar\'e equations. 
\begin{align}
\begin{split}
&{\color{red}\rm d}  \frac{\delta  \ell (u,a)}{\delta  u} 
+ \mathcal{L}_{{\color{red}\rm d} x_t} \frac{\delta  \ell}{\delta  u}
= 
\frac{\delta  \ell}{\delta  a}\diamond a\,dt
-\,\frac{\delta  \,{\color{red}\rm d}  \mathcal{J} }{\delta  q}\diamond q
+
 p \diamond \frac{\delta  \,{\color{red}\rm d}  \mathcal{J} }{\delta  p}
\,,\qquad
{\color{red}\rm d} a + \mathcal{L}_{{\color{red}\rm d} x_t}a = 0
\,,\\&
{\color{red}\rm d} q + \mathcal{L}_{{\color{red}\rm d} x_t} q - \frac{\delta\, {\color{red}\rm d}  \mathcal{J}(q,p)}{\delta p} = 0
\,,\qquad
{\color{red}\rm d} p - \mathcal{L}_{{\color{red}\rm d} x_t}^T p + \frac{\delta\, {\color{red}\rm d}  \mathcal{J}(q,p)}{\delta q} = 0
\,,\end{split}
\label{SVP3}
\end{align}
where ${\color{red}\rm d} x_t$ is given in \eqref{vel-vdt} and ${\color{red}\rm d}  \mathcal{J}(q,p)$  is given in \eqref{J(q,p)-def}. 
\end{theorem}

\begin{remark}[The diffusion part $\mathcal{K}(q,p) \circ dB_t$ of the  wave Hamiltonian $ {\color{red}\rm d}  \mathcal{J}(q,p)$ in \eqref{J(q,p)-def}]
For definiteness in applying theorem \ref{SALT-SNWP-Thm}, we shall take the diffusion part of the semimartingale wave Hamiltonian $\mathcal{K}(q,p) \circ dB_t$ in equation \eqref{J(q,p)-def} as a pairing of a vector field martingale  $\sigma(x)\circ dB_t$ with the wave momentum map $p\diamond q$, 
\begin{align}
\mathcal{K}(q,p) \circ dB_t = \Scp{p\diamond q}{\sigma(x)}_\mathfrak{X}\circ dB_t 
= \Scp{p}{-\mathcal{L}_{\sigma}q}_V\circ dB_t
= \Scp{-\mathcal{L}_{\sigma} p}{q}_V\circ dB_t
\,.
\label{SVP2-2}
\end{align}
The corresponding variational derivatives of $\mathcal{K}(q,p)=\Scp{p\diamond q}{\sigma(x)}_\mathfrak{X}$ are
\begin{align}
\delta \mathcal{K}(q,p) 
= \Scp{\delta p}{-\mathcal{L}_{\sigma}q}_V
+ \Scp{-\mathcal{L}_{\sigma}p}{\delta q}_V
\,.
\label{SVP2.3}
\end{align}
Thus, the diffusion part of the wave Hamiltonian $\mathcal{K}(q,p) \circ dB_t$ in the phase-space Lagrangian will induce an additional transport of wave properties by the vector field martingale  $\sigma(x)\circ dB_t$. 
\end{remark}

\begin{corollary}[Kelvin circulation theorem for SALT and SNWP dynamics]\label{SALT-SNWP-KelThm}
\begin{align}
{\color{red}\rm d} \oint_{c({\color{red}\rm d} x_t)}\frac{1}{D} \frac{\delta  \ell}{\delta  u} 
= \oint_{c({\color{red}\rm d} x_t)}  \frac{1}{D}  \left( \frac{\delta  \ell}{\delta  a}\diamond a  \,dt
- \frac{\delta\, {\color{red}\rm d}  \mathcal{J}(q,p)}{\delta q}\diamond q  
+ p \diamond\frac{\delta\, {\color{red}\rm d}  \mathcal{J}(q,p)}{\delta p}
\right) 
\,.
\label{SALT-SNWP-Kelvin}
\end{align}
\end{corollary}

\begin{remark}[Creation of fluid circulation by both advection and wave interaction]
The right-hand side of Kelvin's theorem in equation \eqref{SALT-SNWP-Kelvin} raises the issue of whether fluid circulation can be created by the effects of both the advected fluid quantities on the right side of equation \eqref{SALT-SNWP-Kelvin} and also by the effects of the stochastic wave dynamics generated by the wave Hamiltonian ${\color{red}\rm d}  \mathcal{J}(q,p)$ defined now as
\begin{align}
{\color{red}\rm d}  \mathcal{J}(q,p) := 
\mathcal{H}(q,p) \,dt
+ \Scp{p\diamond q}{\sigma(x)}_\mathfrak{X} \circ dB_t\,.
\label{J(q,p)-def1}
\end{align} 
As we shall see, the non-acceleration result for GLM in corollary \ref{nonaccel-result} below precludes generation of fluid circulation by wave effects. 
\end{remark}

\begin{remark}[Compatibility of stochastic terms in the loop and in the integrand of Kelvin's theorem]
The wave noise $\circ \,dB_t$ in the wave Hamiltonian ${\color{red}\rm d}  \mathcal{J}(q,p)$ in equation \eqref{J(q,p)-def1} is assumed to be independent of the fluid transport noise $\circ \,dW_t$ in the  velocity vector field ${\color{red}\rm d} x_t(x)$ in \eqref{vel-vdt}. 
Hence, the stochasticity in the integrand of the right-hand side of equation \eqref{SALT-SNWP-Kelvin} will not interfere with the stochasticity in the loop velocity ${\color{red}\rm d} x_t$ defined in equation \eqref{vel-vdt}.
\end{remark}


\begin{proof}
The first step of the proof of Theorem \ref{SALT-SNWP-Thm} is to take the elementary variational derivatives of the action integral \eqref{SVP2.1}, to find
\begin{align}
\begin{split}
\delta u:\quad
\frac{\delta \ell}{\delta u} &- b\diamond a - p\diamond q = 0
\,,\qquad
\delta b:\quad
{\color{red}\rm d}a + \mathcal{L}_{{\color{red}\rm d}x_t} a   = 0
\,,\qquad
\delta a:\quad
\frac{\delta  \ell}{\delta  a}dt - {\color{red}\rm d}b 
+ \mathcal{L}_{{\color{red}\rm d}x_t}^Tb  = 0
\,,\\&
\delta p:\quad {\color{red}\rm d} q + \mathcal{L}_{{\color{red}\rm d} x_t} q - \frac{\delta\, {\color{red}\rm d}  \mathcal{J}(q,p)}{\delta p} = 0
\,,\qquad
\delta q:\quad
{\color{red}\rm d} p - \mathcal{L}_{{\color{red}\rm d} x_t}^T p + \frac{\delta\, {\color{red}\rm d}  \mathcal{J}(q,p)}{\delta q} = 0
\,.\end{split}
\label{var-eqns-thm}
\end{align}
Using these relations and the two lemmas below will lead to the required motion equation, 
\begin{align}
{\color{red}\rm d} \frac{\delta \ell}{\delta u} + {\rm ad}^*_{{\color{red}\rm d}x_t}\frac{\delta \ell}{\delta u}
- \frac{\delta \ell}{\delta a} \diamond a \,dt 
=
- \frac{\delta\, {\color{red}\rm d}  \mathcal{J}(q,p)}{\delta q}\diamond q  
+ p \diamond\frac{\delta\, {\color{red}\rm d}  \mathcal{J}(q,p)}{\delta p}
\,,\label{motion-eqns-thm}
\end{align}
whose left-hand side recovers the SALT equations, and whose right-hand side reveals the extension of SALT to include SNWP, along with the auxiliary equations for $a$, $q$ and $p$ in \eqref{var-eqns-thm}. 
\end{proof}
\begin{remark}[The total momentum is the sum of particle and wave components]
Suppose the kinetic energy density in the Lagrangian $\ell(u,a)$ in \eqref{SALT-SNWP-Kelvin} is proportional to the square of the transport velocity, $u$. Then, the momentum density obtained from the variational derivative result for $\delta \ell / \delta u$ above will comprise the \emph{sum} of the particle and wave momentum densities ($\delta \ell / \delta u = b\diamond a + p\diamond q$). This means the total momentum density comprises the sum of the particle momentum density ($\mu=b\diamond a$), whose canonical Poisson bracket spatially translates both the advected variables $(a)$ and their conjugate dual variables $(b)$ together, as well as the wave momentum ($\nu=p\diamond q$) whose canonical Poisson bracket spatially translates both the phase of the wave ($q$) and its canonically conjugate momentum, the wave action density ($p$). In other words, the sum of the momentum densities acts via the Poisson bracket to translate the canonically conjugate field variables for both degrees of freedom of the flow together. However, the decomposition of the dynamics into separate equations for the two momentum maps $\mu := b\diamond a $ and $\nu := p\diamond q $ in the proofs of the two lemmas \ref{Lemma-m-eqn} and \ref{Lemma-n-eqn} implies the following \emph{non-acceleration theorem}. 
\end{remark}

\begin{corollary}[Non-acceleration result -- \emph{ghost waves}]\label{nonaccel-result}\rm
The motion equation \eqref{motion-eqns-thm} for $\delta \ell / \delta u = b\diamond a + p\diamond q$ decomposes into the sum of  two separate equations. One is a standard Euler-Poincar\'e fluid equation for the currents and the other one for the wave degrees of freedom. These are:
\begin{align}
\begin{split}
\hbox{For}\quad
\mu := b\diamond a 
\quad\hbox{we have}\quad
{\color{red}\rm d} \mu + \mathcal{L}_{{\color{red}\rm d}x_t} \mu &=  \frac{\delta \ell}{\delta q} \diamond q \,dt
\\
\hbox{For}\quad
\nu := p\diamond q 
\quad\hbox{we have}\quad
{\color{red}\rm d} \nu  + \mathcal{L}_{{\color{red}{\rm d}}x_t} \nu &= - \frac{\delta\, {\color{red}\rm d}  \mathcal{J}(q,p)}{\delta q} \diamond q 
+ p \diamond \frac{\delta\, {\color{red}\rm d}  \mathcal{J}(q,p)}{\delta p}   
\,.\end{split}
\label{nonaccel-eqns}
\end{align}
\end{corollary}

Corollary \ref{nonaccel-result} is a non-acceleration result for the wave and current momenta, in the sense that the waves propagate in the local reference frame of the fluid flow and the presence of waves has no net effect on the mean-flow equations in this model. 
Note that equations \eqref{nonaccel-eqns} provide a non-acceleration theorem for any choice of wave Hamiltonian.
\begin{lemma}\rm\label{Lemma-m-eqn}
Together, the variational equations arising from varying $b$ and $a$ in the first line of  \eqref{var-eqns-thm} imply the following useful identity first proved in \cite{Holm2015}.
\begin{align}
\hbox{Upon defining}\quad
\mu := b\diamond a 
\quad\hbox{we have}\quad
{\color{red}\rm d} \mu - \frac{\delta \ell}{\delta a} \diamond a \,dt = - \mathcal{L}_{{\color{red}\rm d}x_t} \mu
\,.\label{m-eqn-lem}
\end{align}
\end{lemma}

\begin{proof}
For an arbitrary vector field $w\in \mathfrak{X}(M)$, one computes  the following pairing.
\begin{align}
\begin{split}
\left\langle 
{\color{red}\rm d}\mu - \frac{\delta \ell}{\delta a} \diamond a \,dt  \,,\, w 
\right\rangle_{\mathfrak{X}}
&=   
\left\langle 
 {\color{red}\rm d}b\diamond a + b\diamond {\color{red}\rm d}a - \frac{\delta \ell}{\delta a} \diamond a\,dt  \,,\, w 
\right\rangle_{\mathfrak{X}}
\\
\hbox{By equation \eqref{var-eqns-thm} } &=   
\left\langle 
(\mathcal{L}_{{\color{red}{\rm d}}x_t}^Tb) \diamond a 
- b\diamond \mathcal{L}_ a\,,\, w 
\right\rangle_{\mathfrak{X}}
\\&=   
\left\langle 
b\,,\, (-\mathcal{L}_{{\color{red}{\rm d}}x_t} \mathcal{L}_{w} + \mathcal{L}_{w} \mathcal{L}_{{\color{red}{\rm d}}x_t} )a\,
\right\rangle_{V}
\\&=   
\left\langle 
b\,,\, ({\rm ad}_{{\color{red}{\rm d}}x_t}{w})\,a\,
\right\rangle_{V}
=
-\left\langle 
b\diamond a\,,\, {\rm ad}_{{\color{red}{\rm d}}x_t}{w}\,
\right\rangle_{\mathfrak{X}}
\\&=
-\left\langle 
 {\rm ad}^*_{{\color{red}{\rm d}}x_t}(b\diamond a)\,,\,{w}\,
\right\rangle_{\mathfrak{X}}
=
-\,\Big\langle 
 \mathcal{L}_{{\color{red}{\rm d}}x_t}\mu\,,\,{w}\,
\Big\rangle_{\mathfrak{X}}\,.
\end{split}
\label{calc-lem1}
\end{align}
Since $w\in \mathfrak{X}$ was arbitrary, the last line completes the proof of the Lemma. In the last step we have also used the coincidence that coadjoint action ${\rm ad}^*_v \mu$ is identical to Lie-derivative action $\mathcal{L}_v \mu$ when a vector field $v\in \mathfrak{X}$ acts on a 1-form density $\mu\in \mathfrak{X}^*$, where one denotes $\mathfrak{X}^*$ as the dual space of the vector fields $\mathfrak{X}$ with respect to the $L^2$ pairing defined in equation \eqref{diamond-def}.
\end{proof}

\begin{lemma}\rm\label{Lemma-n-eqn}
Likewise, the variational equations arising from varying $p$ and $q$ in the second line of  \eqref{var-eqns-thm} satisfy a similar useful identity.
\begin{align}
\hbox{Upon defining}\quad
\nu := p\diamond q 
\quad\hbox{we have}\quad
{\color{red}\rm d} \nu  =  -\mathcal{L}_{{\color{red}{\rm d}}x_t} \nu - \frac{\delta\, {\color{red}\rm d}  \mathcal{J}(q,p)}{\delta q} \diamond q 
+ p \diamond \frac{\delta\, {\color{red}\rm d}  \mathcal{J}(q,p)}{\delta p}   
\,.\label{m-eqn-lem}
\end{align}

\end{lemma}

\begin{proof}
For an arbitrary vector field $w\in \mathfrak{X}(M)$, one computes the following pairing.
\begin{align}
\begin{split}
\scp{{\color{red}\rm d}\nu }{ w }_{\mathfrak{X}}
&=   
\scp{ {\color{red}\rm d}p\diamond q + p\diamond {\color{red}\rm d}q }{w }_{\mathfrak{X}}
\\&= 
\scp{{\color{red}\rm d}p} {-\mathcal{L}_w q} + \scp{p} {-\mathcal{L}_w {\color{red}\rm d}q }
\\&= 
\Scp{\mathcal{L}_{{\color{red}\rm d} x_t}^T p - \frac{\delta\, {\color{red}\rm d}  \mathcal{J}(q,p)}{\delta q}} {-\mathcal{L}_w q} + \Scp{p} {-\mathcal{L}_w \left(- \mathcal{L}_{{\color{red}\rm d} x_t} q + \frac{\delta\, {\color{red}\rm d}  \mathcal{J}(q,p)}{\delta p}\right) }
\\&= 
\Scp{\mathcal{L}_{{\color{red}\rm d} x_t}^T p}{-\mathcal{L}_w q}
+ \Scp{p} {\mathcal{L}_w \mathcal{L}_{{\color{red}\rm d} x_t} q ) }
+ \Scp{ \frac{\delta\, {\color{red}\rm d}  \mathcal{J}(q,p)}{\delta q}} {\mathcal{L}_w q} 
+
\Scp{p} {- \mathcal{L}_w \frac{\delta\, {\color{red}\rm d}  \mathcal{J}(q,p)}{\delta p} }
\\&= 
\Scp{p}{-\mathcal{L}_{{\color{red}\rm d} x_t}\mathcal{L}_w q}
+ \Scp{p} {\mathcal{L}_w \mathcal{L}_{{\color{red}\rm d} x_t} q ) }
- \Scp{ \frac{\delta\, {\color{red}\rm d}  \mathcal{J}(q,p)}{\delta q}\diamond q} {w}_{\mathfrak{X}}
+ \Scp{p \diamond\frac{\delta\, {\color{red}\rm d}  \mathcal{J}(q,p)}{\delta p} } {w}_{\mathfrak{X}}
\\&= 
 \Scp{p} { - \mathcal{L}_{[{\color{red}\rm d} x_t,w]} q ) }
+ \Scp{ - \frac{\delta\, {\color{red}\rm d}  \mathcal{J}(q,p)}{\delta q}\diamond q  
+ p \diamond\frac{\delta\, {\color{red}\rm d}  \mathcal{J}(q,p)}{\delta p} } {w}_{\mathfrak{X}}
\\&= 
 \Scp{p \diamond q} { -{\rm ad}_{{\color{red}\rm d} x_t}w  }
+ \Scp{ - \frac{\delta\, {\color{red}\rm d}  \mathcal{J}(q,p)}{\delta q}\diamond q  
+ p \diamond\frac{\delta\, {\color{red}\rm d}  \mathcal{J}(q,p)}{\delta p} } {w}_{\mathfrak{X}}
\\
\scp{{\color{red}\rm d}\nu }{ w }_{\mathfrak{X}}
&=  
 \Scp{-{\rm ad}^*_{{\color{red}\rm d} x_t}(p \diamond q)
- \frac{\delta\, {\color{red}\rm d}  \mathcal{J}(q,p)}{\delta q}\diamond q  
+ p \diamond\frac{\delta\, {\color{red}\rm d}  \mathcal{J}(q,p)}{\delta p} } {w}_{\mathfrak{X}}
\,.\end{split}
\label{calc-lem2}
\end{align}

\end{proof}

\begin{remark}[Stochastic nonlinear wave propagation ignoring fluid flow] \label{GLM-SNWP}
Let us focus on the particular choice in \eqref{SVP2.2} of the stochastic component of the wave Hamiltonian $\mathcal{K}(q,p)= \scp{p\diamond q}{\sigma(x)}_\mathfrak{X}$ in $ {\color{red}\rm d}  \mathcal{J}(q,p)$ as in equation \eqref{J(q,p)-def1}. In the absence of fluid motion, the corresponding stochastic nonlinear wave dynamics in the second line of \eqref{var-eqns-thm} are obtained from the variational derivatives in \eqref{SVP2.3} as
\begin{align}
\begin{split}
{\color{red}\rm d} q &= \frac{\delta\, {\color{red}\rm d}  \mathcal{J}(q,p)}{\delta p}  
=  \frac{\delta\, \mathcal{H}(q,p)}{\delta p} dt - \mathcal{L}_{\sigma}q \circ dB_t  
\,,
\\
{\color{red}\rm d} p &= -\frac{\delta\, {\color{red}\rm d}  \mathcal{J}(q,p)}{\delta q}  
= - \frac{\delta\, \mathcal{H}(q,p)}{\delta q}dt + \mathcal{L}_{\sigma}p \circ dB_t
\,.
\end{split}
\label{SNWP1-wave}
\end{align}
We conclude that the role of the Lie transport operators in the last line of equation \eqref{var-eqns-thm} is simply to put the wave propagation into the frame of the stochastic fluid motion. That is, the wave propagation is passive. 

Thus, as the waves propagate in the frame of the fluid flow, they cannot transfer momentum to the fluid flow, nor can they  generate fluid circulation in Kelvin's theorem. The result is stochastic wave-current non-acceleration. When the fluid flow is added back into the wave dynamics, equations \eqref{SNWP1-wave} with the choice \eqref{SVP2.2} for the semimartingale part of the wave Hamiltonian in equation \eqref{J(q,p)-def} become 
\begin{align}
\begin{split}
{\color{red}\rm d} q + \mathcal{L}_{{\color{red}\rm d} x_t} q+ \mathcal{L}_{\sigma}q \circ dB_t 
& =  \frac{\delta\, \mathcal{H}(q,p)}{\delta p} dt 
\,,
\\
{\color{red}\rm d} p -  \mathcal{L}_{{\color{red}\rm d} x_t}^Tp - \mathcal{L}_{\sigma}^Tp \circ dB_t &= - \frac{\delta\, \mathcal{H}(q,p)}{\delta q}dt 
\,.
\end{split}
\label{SNWP1-wave+fluid}
\end{align}
in which we see that the wave properties are transported by both wave and fluid vector fields in SWCI. These relations imply the following corollary, in which the contributions of the choice of wave Hamiltonian in \eqref{SVP2.3} can be seen explicitly.

\begin{corollary}[Kelvin circulation theorem for SALT and SNWP dynamics]\label{SALT-SNWP-KelThm-x}
\begin{align}
\begin{split}
{\color{red}\rm d} \oint_{c({\color{red}\rm d} x_t)}\frac{1}{D} \frac{\delta  \ell}{\delta  u} 
= &\oint_{c({\color{red}\rm d} x_t)}  \frac{1}{D}  \left( \frac{\delta  \ell}{\delta  a}\diamond a  
- \frac{\delta\, \mathcal{H}(q,p)}{\delta q}\diamond q  
+ p \diamond\frac{\delta\, \mathcal{H}(q,p)}{\delta p}
\right) \,dt
\\
&\quad + \oint_{c({\color{red}\rm d} x_t)}  \frac{1}{D}  \left( 
- \big(\mathcal{L}_{\sigma} p\big) \diamond q  
+ p \diamond \big(\mathcal{L}_{\sigma} q\big)
\right)  \circ dB_t 
\,.
\end{split}
\label{SALT-SNWP-Kelvin-x}
\end{align}
where the material loop $c({\color{red}\rm d} x_t)$ follows the stochastic Lagrangian fluid path generated by the vector field ${\color{red}\rm d} x_t$ given in equation \eqref{vel-vdt} for a stochastic term $dW_t$, which is not correlated with $dB_t$ above.
\end{corollary}

By Corollary \ref{SALT-SNWP-KelThm-x} the Kelvin circulation theorem \eqref{SALT-SNWP-Kelvin-x} separates into two independent Kelvin circulation theorems for the separate wave and current parts.

\begin{corollary}[Separate Kelvin circulation theorems for SALT and SNWP dynamics]\label{SALT-SNWP-KelThm-xy}\rm
 The Kelvin circulation theorem \eqref{SALT-SNWP-Kelvin-x} for SALT and SNWP dynamics splits into the sum of  two separate equations circulation theorems for interpenetrating fluids with the same circulation loop. The summands are:
 \begin{align}
\begin{split}
{\color{red}\rm d} \oint_{c({\color{red}\rm d} x_t)}\frac{\mu}{D} 
= &\oint_{c({\color{red}\rm d} x_t)}  \frac{1}{D}  \left( \frac{\delta  \ell}{\delta  a}\diamond a  
\right) \,dt
\\
\quad\hbox{and}\quad
\\
{\color{red}\rm d} \oint_{c({\color{red}\rm d} x_t)}\frac{\nu}{D} 
= &\oint_{c({\color{red}\rm d} x_t)}  \frac{1}{D}  \left(
- \frac{\delta\, \mathcal{H}(q,p)}{\delta q}\diamond q  
+ p \diamond\frac{\delta\, \mathcal{H}(q,p)}{\delta p}
\right) \,dt
\\
&\quad + \oint_{c({\color{red}\rm d} x_t)}  \frac{1}{D}  \left( 
- \big(\mathcal{L}_{\sigma} p\big) \diamond q  
+ p \diamond \big(\mathcal{L}_{\sigma} q\big)
\right)  \circ dB_t 
\,.
\end{split}
\label{SALT-SNWP-Kelvin-split}
\end{align}

\end{corollary}

\end{remark}

\section{Application \#1: SALT and SNWP for GLM in stratified EB fluids}\label{StochGLM-sec}

\paragraph{The deterministic case.}
In equation \eqref{Lag-det}, we have augmented the known deterministic Lagrangian for EB fluids \cite{HMR1998} by coupling it to a phase-space Lagrangian for wave dynamics, as follows,
\begin{align}
\begin{split}
\ell(\bu^L,D,b,N,\phi:p)
&= \int_\mathcal{D} \bigg[
\frac{D}{2}\big| \bu^L \big|^2 + D\bu^L\cdot \bR(\bx) - gDbz - p(D-1) 
\\&\hspace{2cm}
- N(\partial_t\phi  + \bu^L\cdot\nabla \phi)\,d^3x + H_W(N,\bk) \,.
\end{split}
\label{Lag-stoch}
\end{align}
The first line of the Lagrangian in \eqref{Lag-stoch} is the fluid Lagrangian for EB fluids in standard form \cite{HMR1998}. The second line is the phase-space Lagrangian for the wave degrees of freedom. 
The term $-\int_\mathcal{D} N\nabla \phi\cdot \bu^L\,d^3x$ in the second line has both wave and fluid components. This term serves to couple the EB Lagrangian for the fluid variables with the phase-space Lagrangian for the wave variables by pairing the wave momentum with the fluid velocity. 

To proceed, let us rewrite the deterministic equations \eqref{SVP3-det} for the stratified EB fluid dynamics in a more geometric form so we will be able to see they lead to the stochastic Kelvin circulation theorem  more easily, 
\begin{align}
\begin{split}
&(\partial_t  + \mathcal{L}_{u^L}) \big(\bm\cdot d\bx \otimes d^3x\big) = 
\left( 
D d \pi + Dgz d b
+ {\rm div} \Big(\frac{\delta H_W}{\delta  \bk}\Big) d\phi
+ N d \Big( \frac{\delta H_W }{\delta  N}\Big)
\right)\otimes d^3x
\,,\\&
(\partial_t  + \mathcal{L}_{u^L})  (D\,d^3x)  = 0
\,,\qquad
D=1
\,,\qquad
(\partial_t  + \mathcal{L}_{u^L}) b = 0
\,,\\&
(\partial_t  + \mathcal{L}_{u^L})  \phi  - \frac{\delta\, H_W}{\delta N} = 0
\,,\qquad
(\partial_t  + \mathcal{L}_{u^L}) (N\,d^3x) - \mathcal{L}_{{\delta\, H_W}/{\delta \bk}}\,d^3x = 0
\,,\end{split}
\label{SVP3-det-redux}
\end{align}
where the Eulerian momentum density $\bm$ and pressure $\pi$ in these equations are recalled from \eqref{m&pi-defs} as, 
\begin{align}
\bm := \frac{\delta  \ell}{\delta  \bu^L}  = D (\bu^L + \bR(\bx)) - N\nabla \phi
\,,\qquad
\pi :=  \frac{\delta  \ell}{\delta  D} = \frac12 |\bu^L|^2 + \bR(\bx)\cdot \bu^L - gbz - p
\,.
\label{m&pi-defs-redux}
\end{align}
From the first two equations, one obtains the form needed for the Kelvin theorem, 
\begin{align}
(\partial_t  + \mathcal{L}_{u^L}) \Big( \frac{1}{D}\bm\cdot d\bx \Big) = 
 d \pi + gz d b
+ \frac{1}{D}{\rm div} \Big(\frac{\delta H_W}{\delta  \bk} \Big) d\phi
+ \frac{N}{D} d \Big( \frac{\delta H_W }{\delta  N}\Big)
\label{Kel-form}
\end{align}

Next, we recall the WKB wave Hamiltonian which leads to the GLM equations,
\begin{align}
H_W = - \int_M N \omega(\bk) \,d^3x 
\,,\quad\hbox{with}\quad
 \frac{\delta H_W }{\delta  N}\Big|_{\bk} = - \,\omega(\bk)
 \,,\quad\hbox{and}\quad
 \frac{\delta H_W}{\delta  \bk}\Big|_{N} = - N \frac{\partial \omega(\bk) }{\partial \bk} =: -\,N \bv_G(\bk)
\,,\label{separatedWaveHam-redux}
\end{align}
in which $\bv_G(\bk):=\partial \omega(\bk) / \partial \bk$ is the group velocity for the dispersion relation $\omega=\omega(\bk)$ between wave frequency, $\omega$, and wave number, $\bk$, given for internal waves at leading order by \cite{GH1996} as
\begin{align}
  \omega^2(\bk) = \frac{(2\sym{\Omega}\cdot\bk)^2 }{ k^2} +
  \Big(\delta^{jl}-{ \frac{k^j k^l}{k^2}}\Big)
  {\frac{\partial^2 p}{\partial x^j \partial x^l}}
  \quad\hbox{with}\quad
 2\sym{\Omega} = {\rm curl} \bR(\bx)\,.
\label{disp_exp}
\end{align}

The motion equation for WCI in equation \eqref{SVP3-det-redux} implies the following Kelvin circulation dynamics
\begin{align}
\begin{split}
\frac{d}{dt} 
&\oint_{c(u^L)} \frac{1}{D}\bm\cdot d\bx
= \oint_{c(u^L)} (\partial_t  + \mathcal{L}_{u^L})\Big(\frac{1}{D}\bm\cdot d\bx\Big)
\\&\qquad = \oint_{c(u^L)} \big(\nabla \pi +  gz \nabla b\big) \cdot d\bx -
\oint_{c(u^L)}
\underbrace{\
\frac{1}{D}  \bigg(\bk \,{\rm div} \Big( N \bv_G(\bk)\Big) 
+ 
N \nabla \omega(\bk)\bigg) \
}_{\hbox{GLM Wave Forcing}}\hspace{-1mm}
\cdot \,d\bx
\,,
\end{split}
\label{Det-GLM-Kelvin}
\end{align}
where $c(u^L)$ is a material loop moving with the flow at velocity $\bu^L(\bx,t)$. The quantities $\bm$ and $\pi$ in \eqref{Det-GLM-Kelvin} are defined in equation \eqref{m&pi-defs-redux}.

\begin{remark}[Non-acceleration is broken for non-constant $D$]\label{breaking-nnonaccel}
The presence of $D$ in the last term in \eqref{Det-GLM-Kelvin} links the  wave and fluid components of the flow when $D$ is \emph{not} constant. Thus, as we shall see in the next section, the non-acceleration result \emph{does not hold} when $D$ is a dynamical variable.
\end{remark}

\paragraph{The SALT and SNWP stochastic cases.}
To recover the SALT GLM equations derived in \cite{Holm2019} and extend them to SNWP GLM equations by following the general case in the previous section, we make two replacements. One is in the transport velocity and the other is in the wave Hamiltonian, as 
\begin{align}
\mathcal{L}_{u^L} \to \mathcal{L}_{{\color{red}\rm d}x_t} 
\quad\hbox{with ${\color{red}\rm d}x_t$ in \eqref{vel-vdt} and}\quad
H_W \to {\color{red}\rm d}h_W := H_Wdt + K_W(N,\phi)\circ dB_t
\,.\end{align}
For GLM we choose the diffusion part of the wave Hamiltonian to be $K_W(N,\phi)=\int N\nabla\phi \cdot \sym{\sigma}(\bx)\,d^3x$, as in equation \eqref{J(q,p)-def1} of Remark \ref{GLM-SNWP}. Thus, equation \eqref{Kel-form} becomes 
\begin{align}
\begin{split}
({\color{red}\rm d}  + \mathcal{L}_{{\color{red}\rm d}x_t})
\Big( \frac{1}{D}\bm\cdot d\bx \Big) = 
 (d \pi + gz d b)dt
&- \frac{1}{D}  \Big(\bk \,{\rm div} \Big( N \bv_G(\bk)\Big) 
+ 
N \nabla \omega(\bk)\Big) 
\cdot \,d\bx\, dt
\\& - \frac{1}{D}\Big(\bk \,{\rm div} \big( N \sym{\sigma}(\bx)\big)
- N\nabla\big(\bk\cdot\sym{\sigma}(\bx)\big)
\Big) \circ dB_t
\,.\end{split}
\label{Kel-form-stoch}
\end{align}
Here, the Bernoulli quantity $\pi$ as 
\begin{align}
\pi := \frac{\delta  \ell}{\delta  D} =  \frac12 |\bu^L|^2 + \bR(\bx)\cdot \bu^L - gbz - p
\,,
\label{martingale-p}
\end{align}
which is required in order to impose preservation of volume when the transport velocity $dt$ is stochastic, as discussed in \cite{SC2020}. 

The motion equation for WCI in equation \eqref{SVP3-det-redux} implies the following Kelvin circulation dynamics
\begin{align}
\begin{split}
{\color{red}\rm d}
&\oint_{c({\color{red}\rm d}x_t)} \frac{1}{D}\bm\cdot d\bx
= \oint_{c({\color{red}\rm d}x_t)} ({\color{red}\rm d} + \mathcal{L}_{{\rm dx_t}})\Big(\frac{1}{D}\bm\cdot d\bx\Big)
\\&\qquad = \oint_{c({\color{red}\rm d}x_t)} (\nabla \pi +  gz \nabla b) \cdot d\bx\,dt 
-
\oint_{c({\color{red}\rm d}x_t)}
\frac{1}{D}  \bigg(\bk \,{\rm div} \Big( N \bv_G(\bk)\Big) 
+ 
N \nabla \omega(\bk)\bigg) 
\cdot \,d\bx\,dt
\\&\qquad \hspace{5cm}
-
\oint_{c({\color{red}\rm d}x_t)}
\frac{1}{D}\Big(\bk \,{\rm div} \big( N \sym{\sigma}(\bx)\big)
- N\nabla\big(\bk\cdot\sym{\sigma}(\bx)\big)
\Big) 
\cdot \,d\bx\,\circ dB_t
\,,
\end{split}
\label{SALT-SNWP-GLM-Kelvin1}
\end{align}
where $c({\color{red}\rm d}x_t)$ is a material loop moving with the stochastic flow velocity ${\color{red}\rm d}x_t$ in \eqref{vel-vdt}. Thus, the SALT and SNWP augmentations of GLM have been derived. Future research will investigate the combination of stochastic processes appearing in these dynamics. 

\begin{corollary}[Non-acceleration result -- \emph{ghost waves}]\label{nonaccel-result1}\rm
The Kelvin circulation dynamics in \eqref{SALT-SNWP-GLM-Kelvin1} for $\bm/D := (\bu^L + \bR(\bx)) - D^{-1}N\nabla \phi$ decomposes into the sum of  two separate equations for the currents and wave degrees of freedom. These are:
\begin{align}
\begin{split}
{\color{red}\rm d}
\oint_{c({\color{red}\rm d}x_t)}(\bu^L + \bR(\bx)) \cdot d\bx
&=
\oint_{c({\color{red}\rm d}x_t)} (\nabla \pi +  gz \nabla b) \cdot d\bx\,dt 
\qquad\hbox{(SALT)}\,,
\\
{\color{red}\rm d}
\oint_{c({\color{red}\rm d}x_t)} D^{-1}N\nabla \phi \cdot d\bx
& = 
\oint_{c({\color{red}\rm d}x_t)}
\frac{1}{D}  \bigg(\bk \,{\rm div} \Big( N \bv_G(\bk)\Big) 
+
N \nabla \omega(\bk)\bigg) 
\cdot \,d\bx\,dt
\\&\hspace{-3cm}\hbox{(SALT \& SNWP)}\qquad 
+
\oint_{c({\color{red}\rm d}x_t)}
\frac{1}{D}\Big(\bk \,{\rm div} \big( N \sym{\sigma}(\bx)\big)
- N\nabla\big(\bk\cdot\sym{\sigma}(\bx)\big)
\Big) 
\cdot \,d\bx\,\circ dB_t
\,,
\end{split}
\label{SALT-SNWP-GLM-Kelvin2}
\end{align}
where $c({\color{red}\rm d}x_t)$ is a material loop moving with the stochastic flow velocity ${\color{red}\rm d}x_t$ in \eqref{vel-vdt}. \end{corollary}

Corollary \ref{nonaccel-result1} is a non-acceleration result for the wave and current momenta, in the sense that the waves propagate in the local reference frame of the fluid flow and the presence of waves has no net effect on the mean-flow equations in this model.

\section{Application \#2: SALT and SNWP for shallow water waves}\label{SW-WCI-sec}

\paragraph{Phase-space Lagrangian derivation of the Shallow water waves in 1D (SWW1D)}
Following the pattern set in \eqref{Lag-det}, we augment the known \emph{deterministic} Lagrangian for SWW1D \cite{HMR1998} by appending to it a phase-space Lagrangian for wave dynamics. We may then write the WCI action integral for Hamilton's principle as follows,
\begin{align}
S = \int_{t_1}^{t_2}\ell(u,D,N,\phi)\,dt
&= \int_{t_1}^{t_2}\!\!\int_\mathcal{D} \bigg[
\frac{D}{2} u^2 - \frac{g}{2}(D-b(x))^2 - Nu\phi_x - N\phi_t \,\bigg]  dx\,dt + H_W(N,\phi_x)dt
\,.
\label{SWLag-det}
\end{align}
Hamilton's principle gives
\begin{align}
\begin{split}
0=\delta S = & \int_{t_1}^{t_2}\!\!\int_\mathcal{D} \bigg[
\delta D\Big( \frac{u^2}{2}  - g(D-b)\Big)
+ \delta u\Big(Du - N\phi_x \Big) 
\\ \qquad &
+  \delta N \Big( - \phi_t  - u\phi_x + \frac{\delta H_W}{\delta N}\Big)
+ \delta \phi \Big( N_t + (Nu)_x - \partial_x\frac{\delta H_W}{\delta \phi_x} \Big)
 \,\bigg]  dx\,dt\,.
 \end{split}
\label{SWLag-var}
\end{align}
As before, we choose the wave Hamiltonian to be
\begin{align}
 \begin{split}
H_W(N,\phi_x) &= - \int_\mathcal{D} N \omega(k)dx
\quad\hbox{with}\quad k = \phi_x
\\
\delta H_W(N,\phi_x) &= - \int_\mathcal{D} (\delta N) \omega(k) - (\delta\phi) \,\partial_x (Nv_G(k))dx
\,.
 \end{split}
\label{SW-HamWave}
\end{align}
The canonical equations for $(\phi,N)$ are then
\begin{align}
\phi_t  +  u\phi_x + \omega(k)= 0
\quad\hbox{and}\quad 
N_t + \partial_x (N(u+v_G(k))) = 0\,.
\label{phi-N-eqns}
\end{align}
The corresponding equations for the fluid variables 
\begin{align}
\hbox{Momentum:}\quad \frac{\delta \ell}{\delta u} =:m = Du - N\phi_x
\quad\hbox{and Depth:}\quad D
\label{fluid-quant}
\end{align}
are the Euler--Poincar\'e equations \cite{HMR1998}
\begin{align}
 \begin{split}
m_t + (m\partial_x + m\partial_x m) u &= D \partial_x \Big( \frac{u^2}{2}  - g(D-b) \Big)\,,
\\
D_t + \partial_x(Du) &= 0
\,.
 \end{split}
\label{EP-eqns}
\end{align}

\paragraph{Hamiltonian derivation of the SWW1D}
To pass to the Hamiltonian side, we complete the Legendre transform in the reduced fluid variables $(m,u)$ to find
\begin{align}
H(m,D,\phi,N) &= \scp{m}{u}  - \ell(u,D,N,\phi)
\label{Leg-Ham}
\end{align}
whose variational derivatives are found from
\begin{align}
\delta H = \scp{\delta m}{u}  + \Scp{m -  \frac{\delta \ell}{\delta u} }{\delta u} 
+ \Scp{ -\,\frac{\delta \ell}{\delta D} }{\delta D}\,.
\label{Ham-var}
\end{align}
Thus, we may write the Euler--Poincar\'e equations in \eqref{EP-eqns} as Lie--Poisson Hamiltonian equations for the variables $(m,D)$,
\begin{align}
\partial_t
\begin{bmatrix}
m \\ D
\end{bmatrix}
= -
\begin{bmatrix}
\partial_x m + m \partial_x & D\partial_x
\\
\partial_xD & 0
\end{bmatrix}
\begin{bmatrix}
\delta H / \delta m = u
\\
\delta H / \delta D =   g(D-b) - {u^2}/{2} 
\end{bmatrix}.
\label{m-D-eqns}
\end{align}
Likewise, we write canonical Hamiltonian equations for the wave variables $(\phi,N)$,
\begin{align}
\partial_t
\begin{bmatrix}
\phi \\ N
\end{bmatrix}
= 
\begin{bmatrix}
0 & 1
\\
-1 & 0
\end{bmatrix}
\begin{bmatrix}
\delta H / \delta \phi =  \partial_x(N (u+v_G(k)))
\\
\delta H / \delta N =   - \omega(k) - u \phi_x
\end{bmatrix}.
\label{N+phi-eqns}
\end{align}
Thus, in the variables $(m,D,\phi,N)$ the Poisson matrix operator is block diagonal. That is, we may write equations \eqref{m-D-eqns} and \eqref{N+phi-eqns} in Hamiltonian form as  
 \begin{align}
\partial_t
\begin{bmatrix}
m \\ D \\ \phi \\ N
\end{bmatrix}
= -
\begin{bmatrix}
\partial_x m + m \partial_x & D\partial_x & 0 & 0
\\
\partial_xD & 0  & 0 & 0
\\
0 & 0 & 0 & -1
\\
0 & 0  & 1 & 0
\end{bmatrix}
\begin{bmatrix}
\delta H / \delta m = u
\\
\delta H / \delta D =   g(D-b) - {u^2}/{2} 
\\
\delta H / \delta \phi =  \partial_x(N (u+v_G(k)))
\\
\delta H / \delta N =   - \omega(k) - u \phi_x
\end{bmatrix}.
\label{m+D-eqns}
\end{align}

However, because the particle momentum $m=Du-N\phi_x$ in equation \eqref{fluid-quant} is an unfamiliar fluid variable for SWW, 
it may be easier to understand the equations for the total momentum $M=Du$, in the 
usual language of fluid velocity. Therefore, we will transform the block diagonal Poisson matrix in
 \eqref{m+D-eqns} into the kinematic momentum $M=m+N\partial_x\phi=Du$ as well as $(D,\phi,N))$. After this transformation to the kinematic momentum variable we find the following Poisson matrix in a class of Lie-Poisson operators whose fundamental properties in finite dimensions have already been discussed by Krishnaprasad and Marsden in \cite{KM1987}, for the Hamiltonian dynamics of rigid bodies with flexible attachments,
 \begin{align}
\partial_t
\begin{bmatrix}
M \\ D \\ \phi \\ N
\end{bmatrix}
= -
\begin{bmatrix}
\partial_x M + M \partial_x & D\partial_x & -\phi_x & N\partial_x
\\
\partial_xD & 0  & 0 & 0
\\
\phi_x & 0 & 0 & -1
\\
\partial_x N & 0  & 1 & 0
\end{bmatrix}
\begin{bmatrix}
\delta H / \delta M = u
\\
\delta H / \delta D =   g(D-b) - {u^2}/{2} 
\\
\delta H / \delta \phi = \partial_x \big(Nv_G(k)\big)
\\
\delta H / \delta N = -\omega(k)
\end{bmatrix}.
\label{M+D-eqns}
\end{align}
\begin{remark}[Transforming the system \eqref{M+D-eqns} to the variables $(M,D,k=\phi_x,N)$]
Under the transformation of variables $(M,D,\phi,N)\to(M,D,k=\phi_x,N)$ the system \eqref{M+D-eqns} becomes
 \begin{align}
\partial_t
\begin{bmatrix}
M \\ D \\ k \\ N
\end{bmatrix}
= -
\begin{bmatrix}
\partial_x M + M \partial_x & D\partial_x & k\partial_x & N\partial_x
\\
\partial_xD & 0  & 0 & 0
\\
\partial_xk & 0 & 0 & -\partial_x
\\
\partial_x N & 0  & -\partial_x & 0
\end{bmatrix}
\begin{bmatrix}
\delta H / \delta M = u
\\
\delta H / \delta D =   g(D-b) - {u^2}/{2} 
\\
\delta H / \delta k = - Nv_G(k)
\\
\delta H / \delta N = -\omega(k)
\end{bmatrix}.
\label{M+D_eqns}
\end{align}
This transformation takes the Poisson matrix to a class of Lie-Poisson operators in infinite dimensions whose fundamental properties have already been discussed by Holm and Kupershmidt in \cite{HK1982,HK1987}, for the Hamiltonian dynamics of superfluid ${}^4He$ without vortices. This class of Lie--Poisson brackets was also derived for complex fluids such as liquid crystals, as well as for superfluid ${}^4He$ both with and without vortices in \cite{HolmHVBK2001,Holm2002}. 

These other appearances of the same class of Hamiltonian structure as for WCI help to interpret the wave physics we are dealing with in the present paper. Namely, all of the other theories associated with this class of Lie--Poisson brackets refer to the additional physics described in terms of \emph{order parameters} whose dynamics can be regarded as occurring \emph{internally} in the frame of the moving fluid. That is, the order-parameter dynamics can be regarded as \emph{subscale} physics taking place relative to the frame of reference of the primary fluid motion. This is quite well-known for the case for the 2-fluid model of  superfluids with vortices, for example \cite{HolmHVBK2001}. Actually, it is also well-known for GLM, when one considers the fluid interpretation of the GLM pseudomomentum and wave action density as a pair of momentum maps for the actions of translations and phase shifts of a complex wave amplitude, as one does for the famous Madelung transformation of quantum mechanics \cite{Madelung1927}. 

The order-parameter interpretation of the present formulation of WCI stemming from its Hamiltonian structure makes it seem natural to introduce a stochastic version of WCI in this formulation, in order to describe the uncertainty which may arise due to \emph{unresolved effects} of the wave-current interaction. 

\end{remark}
\begin{remark}[The variational derivatives of the Hamiltonian in \eqref{M+D-eqns}]
The computation of the required variational derivatives in \eqref{M+D-eqns} is accomplished by first passing to the Hamiltonian side via the Legendre transform in the fluid and wave variables $(m,D,\phi,N)$ then rearranging to identify the Hamiltonian dependence in $(M,D,\phi,N)$ variables, as follows
\begin{align}
H(m,D,\phi,N) &= \scp{m}{u} - \scp{N}{\phi_t}  - \ell(u,D,N,\phi)
\nonumber\\ &= \int_\mathcal{D}
(mu - N\phi_t )\,dx - \int_\mathcal{D}\Big(\frac{D}{2} u^2 - \frac{g}{2}(D-b(x))^2 - Nu\phi_x - N\phi_t \Big)dx - H_W(N,\phi_x)\,,
\nonumber\\ 
H(M,D,\phi,N) &=  \int_\mathcal{D}
\bigg[\frac{M^2}{2D}  + \frac{g}{2}\big(D-b(x)\big)^2 \bigg]dx - H_W(N,\phi_x)\,.
\label{Leg-Ham}
\end{align}
The variational derivatives in $(M,D,\phi,N)$ are found from \eqref{Leg-Ham} as
\begin{align}
\delta H = \int_\mathcal{D} \bigg[& 
\frac{M}{D}\delta M 
+ \Big(-\frac{M^2}{2D^2} + g (D-b) \Big)\delta D
 - (\delta\phi) \partial_x \big(Nv_G(k)\big)
+ ( \delta N )\omega(k)
\bigg]dx\,,
\label{Ham-var}
\end{align}
where we have used equation \eqref{SW-HamWave} for the variations of the wave Hamiltonian $H_W(N,\phi_x)$. 
\end{remark}

After a bit of manipulation, one may write equations \eqref{M+D-eqns} a form which is familiar in fluid dynamics, 
\begin{align}
 \begin{split}
u_t + uu_x = -g \partial_x\big(D-b(x)\big) + &\frac1D \partial_x\big(Nkv_G(k)\big)
\,,\\
D_t + \partial_x(Du) &= 0
\,,\\
\phi_t + u\phi_x + \omega(k) &= 0
\,,\\
N_t + \partial_x\big(N(u+v_G(k))\big) &=0
\,,
 \end{split}
\label{SW-WCI-final}
\end{align}
where $k=\phi_x$ is the 1D wave vector and $v_G(k)=\partial\omega/\partial k$ is the group velocity. One may regard the additional force in the 1D motion equation which depends on the wave variables as a nonhydrostatic `ponderomotive' pressure force due to the presence of the wave degree of freedom which propagates in the local frame of reference of the fluid flow. 

In particular, surface gravity waves in shallow water of mean depth $h$ satisfy the well-known dispersion relation, 
\cite{Vallis2017}
\begin{eqnarray}
{\omega}^2(k)=gk\tanh{hk}
\,,
\label{GWdispersrelatn}
\end{eqnarray}
which admits both leftward and rightward travelling waves with group velocity $v_G=\partial \omega / \partial k$.
Substitution of the shallow water dispersion relation \eqref{GWdispersrelatn} into equation set \eqref{SW-WCI-final} yields the final equation set for WCI in 1D shallow water. 

In 2D, the SWWCI equations can be read off the Lie-Poisson form of the equations in \eqref{M+D_eqns} as
\begin{align}
\partial_t
\begin{bmatrix}
M_i \\ D \\ k_i \\ N
\end{bmatrix}
= -
\begin{bmatrix}
\partial_j M_i + M_j \partial_i & D\partial_i & - k_{j,i} + \partial_j k_i & N\partial_i
\\
\partial_jD & 0  & 0 & 0
\\
k_{i,j} + k_j\partial_i & 0 & 0 & -\partial_i
\\
\partial_j N & 0  & -\partial_j & 0
\end{bmatrix}
\begin{bmatrix}
\delta H / \delta M_j = u^j
\\
\delta H / \delta D =   g(D-b) - {|\bu|^2}/{2} 
\\
\delta H / \delta k_j = - Nv^j_G(k)
\\
\delta H / \delta N = -\,\omega(k)
\end{bmatrix}.
\label{M+D_eqns-2D}
\end{align}
The corresponding SWWCI 2D equations are
\begin{align}
 \begin{split}
\partial_t\bu - \bu\times {\rm curl} \,\bu = -g \nabla\big(D-b(x)\big) + &\frac1D \partial_j \big(N\bk v^j_G(\bk)\big)
\,,\\
\partial_t D + {\rm div}(D\bu) &= 0
\,,\\
\partial_t \bk + \nabla (\omega(\bk) + \bu\cdot\bk )   &= 0
\,,\\
N_t + {\rm div}\big(N(\bu+\bv_G(k))\big) &=0
\,,
 \end{split}
\label{SW-WCI-final2D}
\end{align}

\paragraph{The SALT and SNWP stochastic cases for the Hamiltonian version of SWW1D.}
We propose an extension to stochastic SWW1D flow on the Hamiltonian side by modifying the Hamiltonian function in equation \eqref{Leg-Ham} to make it stochastic, following equation \eqref{J(q,p)-def1} for the diffusion part of the wave Hamiltonian, as
\begin{align}
\begin{split}
{\rm d}h(M,D,\phi,N)  &=  \int_\mathcal{D} \bigg[\frac{M^2}{2D}  + \frac{g}{2}\big(D-b(x)\big)^2 \bigg]dx \,dt
+ \int_\mathcal{D} M \sum_i \xi_i(x) dx \circ dW_t^i 
\\
& \qquad - H_W(N,\phi_x)dt + \int_\mathcal{D} (N\phi_x) \sum_i \sigma_i(x) dx \circ dB_t^i 
\,.
\end{split}
\label{Leg-Ham-stoch}
\end{align}

Then, the \emph{stochastic} version of the SWW1D motion equations in \eqref{M+D-eqns} becomes
 \begin{align}
{\rm d}
\begin{bmatrix}
M \\ D \\ \phi \\ N
\end{bmatrix}
= -
\begin{bmatrix}
\partial_x M + M \partial_x & D\partial_x & -\phi_x & N\partial_x
\\
\partial_xD & 0  & 0 & 0
\\
\phi_x & 0 & 0 & -1
\\
\partial_x N & 0  & 1 & 0
\end{bmatrix}
\begin{bmatrix}
\delta ({\rm d}h) / \delta M = {{\color{red}{\rm d}}x_t}
\\
\delta ({\rm d}h) / \delta D =   \pi dt
\\
\delta ({\rm d}h) / \delta \phi = \partial_x \big(N\widetilde{\bv}_G\big)
\\
\delta ({\rm d}h) / \delta N = -\,\widetilde{\omega} 
\end{bmatrix}
\label{M+D-eqns-stoch}
\end{align}
where one defines the hydrostatic pressure $(\pi)$ and stochastic transport vector field $({{\color{red}{\rm d}}x_t})$ as,
\begin{align}
\pi:= g(D-b) - {u^2}/{2}
\quad\hbox{and}\quad
{{\color{red}{\rm d}}x_t} := u\,dt + \sum_i \xi_i(x) \circ dW_t^i
\,,
\label{Lag-traj-stoch}
\end{align}
and one introduces notation for the stochastic versions of group velocity $(\widetilde{\bv}_G)$ and frequency $(\widetilde{\omega})$ as
\begin{align}
\widetilde{\bv}_G := \bv_G(k)dt + \sum_i {\sym{\sigma}}_i(\bx) \circ dB_t^i
\quad\hbox{and}\quad
\widetilde{\omega} := \omega(\bk)dt +  \bk\cdot\sum_i {\sym{\sigma}}_i(\bx) \circ dB_t^i
\,,
\label{grpv-freq-stoch}
\end{align}
written here in vector form for clarity when generalising to  higher dimensions. 
Note that $\widetilde{\bv}_G = \partial \widetilde{\omega} / \partial \bk$.
Physically, the
noise introduced into the diffusion part of the wave Hamiltonian in equation \eqref{Leg-Ham-stoch} produces in \eqref{grpv-freq-stoch} a stochastic shift in the group velocity, accompanied by the corresponding stochastic Doppler shift in the wave frequency. 

\begin{remark}[Determining the noise eigenvectors ${\sym{\xi}}_i(\bx)$ and ${\sym{\sigma}}_i(\bx)$]
The vector fields ${\sym{\xi}}_i(\bx)$ and ${\sym{\sigma}}_i(\bx)$ would need to be specified, or obtained,  from another source, such as observation data for the velocity-velocity correlation tensor for the currents, and the effective group velocity and wave frequency of the wave field. Determining these functions will comprise the fundamental crux of applying this class of stochastic GLM equations for uncertainty quantification and data assimilation. Further discussion of this challenge is beyond the scope of the present work. However, previous work indicates that viable procedures can be developed to meet this challenge, as done already for related problems in \cite{CCHPS2018,CCHPS2019a,CCHPS2019b,CCHPS2020}.
\end{remark}

The 1D fluid dynamical form of these stochastic SW-WCI equations is 
\begin{align}
 \begin{split}
{\color{red}{\rm d}}u + {{\color{red}{\rm d}}x_t}u_x + u \partial_x \Big( \sum_i \xi_i(x) \circ dW_t^i\Big)
= -g \partial_x\big(D-b(x)\big)dt + &\frac{k}{D} \big(N\widetilde{v}_G(k)\big)
+ \frac{N}{D} \partial_x\widetilde{\omega}
\,,\\
{\color{red}{\rm d}}D + \partial_x(D{{\color{red}{\rm d}}x_t}) &= 0
\,,\\
{\color{red}{\rm d}}\phi + {{\color{red}{\rm d}}x_t}\phi_x + \widetilde{\omega}(k) &= 0
\,,\\
{\color{red}{\rm d}}N + \partial_x\big(N({{\color{red}{\rm d}}x_t}+\widetilde{v}_G)\big) &=0
\,.
 \end{split}
\label{SW-WCI-stoch}
\end{align}

In 2D, these stochastic Hamiltonian equations would be written in fluid dynamical form as
\begin{align}
\begin{split}
\big({\color{red}{\rm d}} + \mathcal{L}_ {{\color{red}{\rm d}}x_t}\big) (\bu\cdot d\bx) 
&= - d\pi\,dt + \frac1D \Big( {\rm div}(N\widetilde{\bv}_G)\,d\phi + N d \widetilde{\omega} \Big)
\,,\\
\big({\color{red}{\rm d}} + \mathcal{L}_ {{\color{red}{\rm d}}x_t}\big)(Dd^2x) &= 0
\,,\\
\big({\color{red}{\rm d}} + \mathcal{L}_ {{\color{red}{\rm d}}x_t}\big)\phi  &= -\,\widetilde{\omega}(k)
\,,\\
\big({\color{red}{\rm d}} + \mathcal{L}_ {{\color{red}{\rm d}}x_t}\big)(Nd^2x)&= 0
\,.
\end{split}
\label{2DSW-stoch}
\end{align}

Note that the non-acceleration result in corollary \ref{nonaccel-result1} for incompressible GLM flow \emph{does not hold} when $D$ is a dynamical variable. This is clear from the following Kelvin circulation theorem for  SW-WCI in 2D.

\begin{theorem}
The stochastic Kelvin circulation theorem corresponding to the stochastic SW-WCI motion equation in 2D is given by
\begin{align}
{\color{red}{\rm d}} \oint_{c({\color{red}{\rm d}}x_t)}\hspace{-2mm}\bu\cdot d\bx
=
\oint_{c({\color{red}{\rm d}}x_t)}
\hspace{-4mm}- \,d\pi\,dt + \frac1D \Big( {\rm div}(N\widetilde{\bv}_G)\,d\phi + N d \widetilde{\omega} \Big)
\,,
\label{2DSW-stoch-Kel}
\end{align}
in which the wave sources of flow circulation are evident and the two sources of circulation cannot be separated. 
\end{theorem}

\begin{remark}
For a contrasting approach to deriving stochastic shallow water models, which combines asymptotic expansions and vertical averaging with the stochastic variational framework discussed here for the formulation of new stochastic parametrisation schemes for the nonlinear wave fields, see \cite{HolmLuesink2019}.
\end{remark}

\section{Deterministic comparison of the Craik-Leibovich model with GLM}\label{sec:det-CL}

Among the many processes which occur in the oceanic surface boundary layer, Langmuir circulations (LCs) attract much of the  attention because they are believed to affect the air-sea exchanges of heat and gases through an enhancement of turbulent mixing \cite{Thorpe2004}. In the formation of LCs, the interaction between surface waves and the mean flow is believed to play a central role. Craik and Leibovich \cite{CraikLeibovich1976} derived an expression for the wave-current interaction for wind-driven waves in the oceanic mix layer called the \emph{Stokes vortex force} (SVF) and showed that the SVF induces roll structures similar to the observed LCs. Today, the SVF representation of the wave-current interaction in the momentum equation is in general use for numerically modelling the effects of LCs in mixed layer turbulence by using large-eddy simulations (LES), although the theoretical issues are by no means settled \cite{Fujiwara-etal2018,Fujiwara-MellorReply2019,Mellor-Fujiwara2019,Tejada-Martinez2020}.

\paragraph{A quick derivation of the CL motion equation obtained by time averaging Kelvin's theorem.}
One may derive the CL model by considering how averaging applies to the Kelvin circulation theorem for the EB model, 
\begin{equation}
\frac{d}{dt}\oint_{c(u)} \big( \bu(\bx,t) + \bR(\bx) \big)\cdot d\bx = \oint_{c(u)} \Big( \dots \Big)\cdot d\bx 
\label{Lag-det-CL}
\end{equation}
In Kelvin's theorem, the loop moves with the flow, so the loop is a Lagrangian quantity. The integrand is fixed in space, so the integrand is Eulerian. Thus, after taking averages, the loop velocity will be the Lagrangian mean velocity, $\bu^L$, and the integrand velocity will be the Eulerian mean velocity, $\bu=\bu^L-\bu^S$, when defined in terms of the Stokes mean drift velocity, $\bu^S(\bx)$.%
\footnote{For convenience,  the Stokes mean drift velocity, $\bu^S(\bx)$, is usually taken to be time independent and divergence-free. However, these two assumptions remain controversial in the CL literature. }
 (For the sake of simplicity, we drop the bar notation for mean quantities.) Thus, the mean Kelvin theorem will read
\begin{equation}
\frac{d}{dt}\oint_{c(u^L)} \big(  \bu^L(\bx,t) - \bu^S(\bx) + \bR(\bx) \big)\cdot d\bx = \oint_{c(u^L)} \Big( \dots \Big)\cdot d\bx 
\label{Lag-det-CL}
\end{equation}
Taking the time derivative of the loop integral then yields the motion equation in the loop-integral form,
\begin{equation}
\oint_{c(u^L)} \big( \partial_t + \mathcal{L}_{u^L} \big)\Big(\big( \bu^L(\bx,t) - \bu^S(\bx) + \bR(\bx) \big)\cdot d\bx\Big)
 = \oint_{c(u^L)} \Big( \dots \Big)\cdot d\bx \,,
\label{Lag-det-CL}
\end{equation}
where the coordinate notation for the Lie derivative $\mathcal{L}_{u^L} (\bv \cdot  d\bx)$ for a 1-form $\bv \cdot  d\bx$ may be written out conveniently in two equivalent vector forms which are familiar in fluid dynamics, 
\begin{equation}
\mathcal{L}_{u^L} (\bv \cdot  d\bx) = \big( \bu^L\cdot\nabla)\bv + v_j\nabla u^{L\,j} \big)\cdot d\bx
=  \big( - \bu^L\times {\rm curl}\bv + \nabla(\bu^L\cdot\bv) \big)\cdot d\bx\,.
\label{Lie-det-CL}
\end{equation}
These familiar vector forms of the Lie derivative of a 1-form in \eqref{Lie-det-CL} then express the CL SVF in the motion equation for Euler-Boussinesq (EB) flow in its standard vector form in equation \eqref{CL-mot-eqn} below. 

\paragraph{Another derivation of the CL motion equation using Hamilton's principle.}
The ideal CL equations arise from stationarity of a constrained
Hamilton's principle $\delta S = 0$, under variations of the fluid variables 
at constant Eulerian position. The constrained Hamilton's principle for the implementation 
of the CL model in the EB equations is given in terms of the action \cite{Holm1996}, 
\begin{equation}
S = \int \int \,dt \left[ \frac12 D |\bu^L|^2 
- b D g z 
- D \bu^L \cdot \bu^S(\bx)
+ D \bu^L \cdot {\bf R}(\bx) - p(D-1)\right]\,d^3 x\,dt\,.
\label{Lag-det-CL}
\end{equation}
Here $\bu^L(\bx,t)$ is the Lagrangian mean fluid velocity, as before, and
the ``Stokes drift velocity'' $\bu^S(\bx)$ is a prescribed time-independent function of position,
which represents the mean drift velocity caused by oscillating winds near the surface \cite{Thorpe2004}. 
The action integral \eqref{Lag-det-CL} contains the difference of the
kinetic and potential energies, plus a ``${\bf J} \cdot {\bf A}$" coupling of the mass current ${\bf J}=D \bu^L(\bx,t)$ to 
two \emph{constant,  spatially-dependent velocity fields} ${\bf A}_1(\bx)$ and ${\bf A}_2(\bx)$. The first of these is ${\bf A}_1(\bx)=-\bu^S(\bx)$, representing the spatially-dependent boost of the inertial frame into a frame moving with minus the Stokes drift velocity in WCI, from which $-D\bu^L\times {\rm curl}\bu^S(\bx)$ arises as the CL \emph{vortex force} in the fluid motion equation.
The other constant velocity field is ${\bf A}_2(\bx)=\bR(\bx)$, representing the rotation velocity relative to the inertial frame, from which $D\bu^L\times{\rm curl}\bR=D\bu^L\times (2\sym{\Omega})$ arises as the Coriolis force in the fluid motion equation. The action integral \eqref{Lag-det-CL} also contains the incompressibility constraint $D=1$ imposed by the pressure $p$ as a Lagrange multiplier.

\paragraph{Passing to the Hamiltonian side.}
The variation of the Lagrangian in \eqref{Lag-det-CL} with respect to the transport velocity $\bu^L(\bx,t)$ produces the total Eulerian momentum density for GLM in the presence of the Stokes drift, cf. equation \eqref{Lag-var},
\begin{align}
\bm(\bx,t) := \frac{\delta \ell}{\delta \bu^L} = D\big( \bu^L -\bu^S(\bx) + \bR(\bx)\big)\,,
\label{Lag-var-CL}
\end{align}
in which $\bu^S(\bx)$ is the prescribed Stokes drift velocity. 
One may also compare the momentum density shifts in equation \eqref{Lag-var-CL} with the angular momentum shift due to fixed rotation of the reference frame for a rigid body in equation \eqref{Rot-mom-shift}.

Next, we will show that the Hamiltonian dynamics for the momentum density in equation \eqref{Lag-var-CL} recovers the CL motion equation for the Lagrangian mean transport velocity, $\bu^L(\bx,t) $.   

The Hamiltonian corresponding to the Lagrangian in \eqref{Lag-det-CL} is given by the Legendre transform, cf. equation \eqref{Ham-det},
\begin{align}
\begin{split}
H(\bm,D,b)) &= \!\!\int_\mathcal{D}\bm\cdot\bu^L \,d^3x - \ell(\bu^L,D,b,N,\phi:p)
\\&= \!\!\int_\mathcal{D} \bigg[
\frac{1}{2D}\big| \bm + D\bu^S(\bx) - D\bR \big|^2 +  gDbz + p(D-1)  \bigg]d^3x\,.
\end{split}
\label{Ham-det-CL}
\end{align}
The Hamiltonian in \eqref{Ham-det-CL} is the sum of the kinetic and potential energies of the fluid. The variational derivatives are given by
\begin{align}
\delta H(\bm,D,b)) &= \!\!\int_\mathcal{D} 
\bu^L\cdot \delta\bm + \delta D\Big(gbz +p - \frac12|\bu^L|^2 +\bu^L\cdot\bu^S(\bx) - \bu^L\cdot\bR(\bx)\Big) + (gDz)\delta b
\,d^3x\,.
\label{Ham-var-CL}
\end{align}

We may now write the CL equations  in Hamiltonian form by 
using a block-diagonal Poisson matrix operator, cf. equation \eqref{m+D+b-eqns},
 \begin{align}
\partial_t
\begin{bmatrix}
m_i \\ D \\ b 
\end{bmatrix}
= -
\begin{bmatrix}
\partial_j m_i + m_j \partial_i & D\partial_i & -b_{,i} 
\\
\partial_jD & 0  & 0 
\\
b_{,j} & 0  & 0 
\end{bmatrix}
\begin{bmatrix}
\delta H / \delta m_j 
\\
\delta H / \delta D  
\\
\delta H / \delta b 
\end{bmatrix}.
\label{m+D+b-eqns-CL}
\end{align}

\paragraph{Deterministic CL equations for EB fluid flow.}
Upon expanding out the Hamiltonian equations in \eqref{m+D+b-eqns-CL}, the dynamics of the EB fluid with these additional wave variables is found to obey the following system of equations, cf. equation set \eqref{m+D+b-eqns},
%
\begin{align}
\begin{split}
&\partial_t  \bm + (\bu^L\cdot \nabla)\, \bm + (\nabla \bu^L)^T \cdot \bm + \bm \,{\rm div}\bu^L
= 
D \nabla \pi_{CL} + Dgz\nabla b
\,,\\&
\partial_t  D + {\rm div}(D\bu^L)  = 0
\,,\qquad
D=1
\,,\qquad
\partial_t  b + \bu^L\cdot \nabla  b = 0
\,.\end{split}
\label{SVP3-det-CL}
\end{align}
The Eulerian momentum density, $\bm$, and the Bernoulli function, $\pi$, in these equations are defined by the following variational derivatives of the CL Lagrangian in \eqref{Lag-det-CL}, cf. equation \eqref{m&pi-defs},
\begin{align}
\bm := \frac{\delta  \ell}{\delta  \bu^L}  = D (\bu^L - \bu^S(\bx) + \bR(\bx)) 
\,,\qquad
\pi_{CL} :=  \frac{\delta  \ell}{\delta  D} = \frac12 |\bu^L|^2 - \bu^L\cdot \bu^S +  \bu^L\cdot \bR - gbz - p
\,.
\label{m&pi-defs-CL}
\end{align}
The motion equation for WCI in equation \eqref{SVP3-det-CL} implies the following Kelvin circulation dynamics for the Eulerian momentum per unit mass,  compared with the GLM equation \eqref{Kelvin-GLM},
\begin{align}
\begin{split}
\frac{d}{dt} 
\oint_{c(u^L)} \frac1D\frac{\delta  \ell}{\delta  \bu^L}\cdot d\bx 
&=
\oint_{c(u^L)} (\p_t+\mathcal{L}_{u^L})\bigg(\Big(\bu^L - \bu^S(\bx) + \bR(\bx) \Big)\cdot d\bx \bigg)
\\&= 
\oint_{c(u^L)} \nabla\pi \cdot d\bx
+
\oint_{c(u^L)} \underbrace{\
 gz \nabla b \cdot d\bx \
 }_{\hbox{Buoyancy}}
\,.
\end{split}
\label{Kelvin-CL}
\end{align}
Equation \eqref{Kelvin-CL} is Newton's $2^{\rm nd}$ Law for the time rate of change of the total Eulerian mean momentum per unit mass $\bm/D$ of a body whose mass is distributed  on a closed loop $c(u^L)$ moving with the Lagrangian mean velocity $\bu^L$. According to equation \eqref{Kelvin-CL}, the Stokes drift velocity in Newton's Law for this model appears as an addendum to the Coriolis force. The Stokes drift velocity appears in the usual form of the fluid equations as 
\begin{align} 
\begin{split}    
\partial_t \bu^L - \bu^L\times {\rm curl}\big(\bu^L - \bu^S(\bx) + \bR(\bx)\big) 
&= -\,\frac{1}{D}\nabla \Big( p - \frac12 |\bu^L|^2 \Big) - gb \mathbf{\hat{z}}  
\, ,   \\
\partial_t D + {\rm div}(D\bu^L) &= 0\, ,   \quad \hbox{with}\quad D=1
\, ,   \\
\partial_t b + \bu^L\cdot \nabla b &= 0\, 
\end{split}   
\label{CL-eqns}   
\end{align}
Upon defining the Coriolis parameter as $2\sym{\Omega}={\rm curl} \bR(\bx))$, the motion equation becomes
\begin{align} 
\partial_t \bu^L - \bu^L\times {\rm curl}\bu^L - \bu^L\times 2\sym{\Omega}
= -\,\frac{1}{D}\nabla \Big( p - \frac12 |\bu^L|^2 \Big) 
\underbrace{\
- \, \bu^L\times {\rm curl}\bu^S(\bx) \
}_{\hbox{CL Stokes force}}
 - gb \mathbf{\hat{z}}   \,.
\label{CL-mot-eqn}   
\end{align}
To see the potential vorticity (PV) conservation for CL, we rewrite the dissipative CL motion equation \eqref{CL-mot-eqn} 
in terms of the \emph{Eulerian mean velocity} defined to be ${\bu}:= \bu^L - \bu^S(\bx)$,
\begin{align} 
\partial_t {\bu} - \bu^L \times \sym{\varpi} + \nabla (p - \frac12 |\bu^L|^2) =  - \,gb \mathbf{\hat{z}}
 \,.
\label{CL-mot-eqn2}
\end{align}
where we have set $D=1$ and defined $\sym{\varpi} = {\rm curl} \,({\bu} + \bR(\bx))$ as the total Eulerian mean vorticity. It follows that
\begin{align} 
\partial_t \sym{\varpi} - {\rm curl} (\bu^L \times \sym{\varpi} )  
=  -\,g \mathbf{\hat{z}} \times\nabla b \,.
\label{CL-mot-eqn3}
\end{align}
Consequently, the Craik-Leibovich theory conserves Eulerian mean potential vorticity (PV) on Lagrangian particles. Namely,
\begin{align} 
\partial_t Q + \bu^L\cdot \nabla Q = 0\,,
\label{CL-PV}
\end{align}
where PV is defined as $Q := \nabla b \cdot \sym{\varpi}$.

\begin{remark}
Apparently, the difference between CL and GLM fluid dynamics resides in how the modelling  choice between Eulerian and Lagrangian velocity averaging affects the Kelvin circulation theorem. Eulerian averaging affects the Eulerian velocity 1-form in the \emph{circulation integrand} in Kelvin's theorem, while Lagrangian averaging affects the material velocity of the \emph{circulation loop}. This means that the implementation of stochasticity for the CL equations will differ in the same way. In fact, there may be a related modelling choice to be made between It\^o and Stratonovich stochasticity, in choosing between Eulerian and Lagrangian implementations of stochasticity, for example, in the pursuit of uncertainty quantification \cite{Holm2020}. 
\end{remark}

\begin{remark}[Comparing the GLM and CL models for 3D EB flow]

Compared to the action integral for the corresponding GLM theory in \eqref{Lag-det}, the Lagrangian in the action integral \eqref{Lag-det-CL} for the CL theory replaces the current-boosted wave dynamics in the phase-space Lagrangian in the second line of \eqref{Lag-det} by the time-independent, prescribed boost of velocity of the inertial frame by ${\bf A}_1(\bx)=-\bu^S(\bx)$. 
That is, the action integral \eqref{Lag-det-CL} for the CL theory places the Stokes drift velocity $-\bu^S(\bx)$ and the velocity of the rotating frame  $\bR(\bx)$ onto the same footing. Namely, these two quantities are both regarded as velocity boosts into a moving reference frame relative to which the velocity of the current will be defined. As we have seen, in the GLM formulation the waves propagate in the frame of the Lagrangian mean current velocity, which itself flows relative to the rotating frame. However, in the Craik-Leibovich (CL) model, the current flows in the reference frame of the sum of the rotation velocity minus the prescribed Stokes drift velocity which is designed to model the effects of the generation of wave fluctuations on the sea surface due to the external wind. Thus, to compare the deterministic GLM and CL models one should modify the GLM action integral in \eqref{Lag-det} to include the Stokes drift velocity $-\bu^S(\bx)$ boost. In this case, the action integral for the deterministic Stokes-boosted GLM model of the 3D flow of EB fluid will be 
\begin{align}
\begin{split}
S = \int_{t_1}^{t_2}\ell(\bu^L,D,b,N,\phi:p)\,dt
&= \int_{t_1}^{t_2}\!\!\int_\mathcal{D} \bigg[
\frac{D}{2}\big| \bu^L \big|^2 + D\bu^L\cdot \bR(\bx) 
- D \bu^L \cdot \bu^S(\bx) - gDbz - p(D-1) \bigg]d^3x
\\&\hspace{2cm}
- \int_{t_1}^{t_2}\!\!\int_\mathcal{D}  N(\partial_t\phi  + \bu^L\cdot\nabla \phi)\,d^3x + \int_{t_1}^{t_2} H_W(N,\bk) \,.
\end{split}
\label{Lag-det-GLMCL}
\end{align}
Then, because of the non-acceleration result in corollary \ref{nonaccel-result} for GLM with incompressible flow, one can expect that the results of the CL model and the GLM model when Stokes drift velocity is included will largely coincide for the 3D flow of EB fluid. 

\end{remark}

\section{Deriving the OU Craik--Leibovich (OU CL) equations}\label{sec:stoch-CL}

We have seen that the CL model obtains the Eulerian mean fluid momentum density in the rotating frame by subtracting the prescribed Stokes velocity $\bu^S(\bx)$ from the Lagrangian mean transport velocity $\bu$ and adding the velocity of the rotating frame, $\bR(\bx)$. This defines the relative Eulerian momentum density of the fluid as%
\begin{align}
\bm(\bx,t) := \frac{\delta \ell}{\delta \bu} = D\big( \bu -\bu^S(\bx) + \bR(\bx)\big)\,.
\label{CL-rel-momap}
\end{align}
There are several likely sources of uncertainty in the CL model. First is the effect of the unsteady wind forcing typical of natural conditions. Second is the delay of the drift velocity in response to changes in the wind conditions \cite{Thorpe2004}. Yet another another likely source of uncertainty in the CL model lies in errors in the observational determination of the Stokes drift velocity, $\bu^S(\bx)$, \cite{Klein-OceanData2019,vdBBstokes2018}. The goal of this section is to introduce a theoretical framework for quantifying the uncertainty in the solution of the CL equations due to the uncertainty in the Stokes drift, $\bu^S$. For this, we introduce  a probabilistic aspect into the Stokes drift frame velocity in the action integral for Hamilton's principle for fluid dynamics in \eqref{Lag-det-CL}. In exploring this probabilistic aspect, we will neglect the effects of rotation, $\bR(\bx)$; so, we can focus on the effects of uncertainty in the Stokes drift velocity. We will also generally ignore the effects of stochasticity in the transport velocity $\bu$ (SALT) except for taking one passing opportunity to include it in remark \ref{op-SALT}. 

In particular, we consider ideal incompressible 3D fluid motion in the frame of motion with velocity $-\,\bu^S(\bx)N_t $, where $\bu^S(\bx)$ is the prescribed deterministic divergence-free Stokes mean drift velocity and $N_t$ is obtained as the solution path of the Ornstein-Uhlenbeck (OU) stochastic process \cite{OU-ref}
\begin{align}
{\rm \textcolor{red}d}N_t = \theta(\mean{N} - N_t)dt + \sigma dW_t\,,
\label{OU-1}
\end{align}
with long-term mean $\overline{N}$, and real-valued constants $\theta$ and $\sigma$. 
The solution path of the stationary Gaussian-Markov OU process is known to be an ordinary scalar function of time $N_t$ defined by
\begin{equation}
N_t = e^{-\theta t}N_0 + (1 - e^{-\theta t})\overline{N}
+ e^{-\theta t} \sigma \int_0^t e^{\theta s}dW_s
\,,
\label{OUsoln-1}
\end{equation}
in which one may assume an initially normal distribution, 
$N(0)\approx {\cal N}(\overline{N},\sigma^2/(2\theta))$, with mean $\overline{N}$ and variance $\sigma^2/(2\theta)$. 
Thus, we model uncertainty in the prescribed mean drift velocity by multiplying  $-\bu^S(\bx)$ by the OU process in $-\,\bu^S(\bx)N_t $. 

The corresponding extension of the CL model for the 3D flow of EB fluid is obtained as an Euler-Poincar\'e equation for Hamilton's principle $\delta S=0$ with action integral given by, cf. equation \eqref{Lag-det-CL},
\begin{equation}
S = \int_{t_1}^{t_2} \int  \left[ \frac12 D |\bu|^2 
 - D \bu \cdot \bu^S(\bx)N_t - g b D  z
 - p(D-1)\right]\,d^3 x\,dt\,.
\label{Lag-OU-CL}
\end{equation}
Here, the Euler-Poincar\'e equations are obtained from varying the action as 
\begin{align}
\begin{split}
0 = \delta S 
&=
 \int \bigg[ \left\langle D(\mathbf{u} - \bu^S(\bx) N_t,\delta \mathbf{u}\right\rangle + \left\langle \frac{1}{2}|\mathbf{u}|^2 
 - \mathbf{u}\cdot \bu^S(\bx) N_t -gbz - p,\delta D\right\rangle 
\\& \hspace{3cm}
-gDz\,\delta b + \langle 1 - D, \delta p\rangle \bigg]\,dt\,,
\end{split}
\label{HP-OUCL}
\end{align}
where the variations are given in terms of a smooth vector field $w$ by \cite{HMR1998}
\begin{align}
\delta u = \partial_t w - {\rm ad}_u w
\,,\quad
\delta D = - \mathcal{L}_wD
\,,\quad
\delta b = - \mathcal{L}_wb
 \,.
  \label{def-var}
\end{align}

Now, the momentum 1-form density is defined by
\begin{align}
m:= \frac{\delta\ell}{\delta \mathbf{u}} 
=\mathbf{m}\cdot d\mathbf{x}\otimes d^3x 
:= D\big(\mathbf{u} {-}  \bu^S(\bx) N_t \big)\cdot d\mathbf{x}\otimes d^3x
 =: D \mathbf{v}\cdot d\mathbf{x}\otimes d^3x
 \,.
  \label{def-m}
\end{align}
Thus, we have $\mathbf{m} = D (\mathbf{u} {-}  \bu^S(\bx) N_t )$, so the term involving the Stokes velocity $\bu^S(\bx) $ is to be regarded as a component of the total Eulerian momentum, $\mathbf{m}(\bx,t)$. 
The last term in \eqref{def-m} defines a counterpart to the deterministic Craik-Leibovich notation, as follows,
\begin{align}
\mathbf{v} = \mathbf{u} - \bu^S(\bx)N_t 
\iff 
\bu^E(\bx,t) = \bu^L(\bx,t) - \bu^S(\bx)
\,.
\label{def-v}
\end{align}
The Euler-Poincar\'e motion equation resulting from Hamilton's principle \eqref {HP-OUCL} with variations given in  \eqref{def-var} is
\begin{align}
{\sf \textcolor{red}d}m + \pounds_{u}m\,dt +
d\left(p + \mathbf{u}\cdot\bu^S(\bx)N_t 
 - \frac{1}{2}|\mathbf{u}|^2 \right)\otimes  Dd^3x\,dt - gbd z \otimes  Dd^3x\,dt &=0\,.
 \label{SCL-MotEqn}
\end{align}
Thus, by applying the definitions of $\bm$ in equation \eqref {def-m} and of the OU process  in \eqref{OU-1}, the co-vector quantity $\bv:=\bm/D$ in \eqref{def-v} is found from \eqref{SCL-MotEqn} to satisfy 
\begin{align}
\big({\sf \textcolor{red}d} + \pounds_{u}\,dt \big) (\bv\cdot d\bx)
+ d\left(p + \mathbf{u}\cdot\bu^S(\bx)N_t 
 - \frac{1}{2}|\mathbf{u}|^2 \right) \,dt - gbd z\,dt &=0\,.
 \label{SCL-MotEqn-v}
\end{align}
Equation \eqref{SCL-MotEqn-v} may also be written equivalently as the following OU PDE, 
\begin{align}
{\sf \textcolor{red}d}\bv + \bigg( - \bu\times {\rm curl} \bv 
+ \nabla\Big(  p + \frac{1}{2}|\mathbf{u}|^2\Big)  + gb\nabla z\bigg)dt  = 0
\,,
 \label{SCL-MotEqn-vec}
\end{align}
where we have used the continuity equation for the volume element, $D$, 
\begin{align} 
{\sf \textcolor{red}d} D + {\rm div}(D \mathbf{u})dt = 0\,,
 \label{Contin-Eqn}
\end{align}
which implies ${\rm div}\bu=0$ when the constraint $D-1=0$ is enforced by the variation of 
the Lagrange multiplier $p$, the pressure, in \eqref{HP-OUCL}.
Now from \eqref{def-v} we have 
\begin{align}
{\sf \textcolor{red}d}\bv = {\sf \textcolor{red}d} \bu  - \bu^S(\bx) {\sf \textcolor{red}d}N_t 
\,. \label{def-dm}
\end{align}
Consequently, for ${\rm div}\bu^S(\bx)=0$, the the pressure, $p$, may be found by imposing ${\sf  {\rm div} (\textcolor{red}d}\mathbf{u}) = 0$ at each time step, which by equations \eqref{SCL-MotEqn-vec} and \eqref{def-dm} implies
\begin{align}
{\rm div}\Big(- \bu\times {\rm curl} \bv 
+ \nabla\Big(p + \frac{1}{2}|\mathbf{u}|^2\Big) + gb\nabla z \Big)= 0
\,. \label{p-eqn}
\end{align}
These calculations may be summarised in the following theorem. 
\begin{theorem}[OU CL wave dynamics via Hamilton's principle $\delta S = 0$ for action integral  \eqref{Lag-det-CL}]
\label{SALT-OU-CL-Thm}$\,$\rm

The probabilistic CL model with OU wave dynamics is governed by the following Euler-Poincar\'e motion equation obtained from Hamilton's principle with the action integral \eqref{Lag-det-CL}, 
\begin{align}
{\sf \textcolor{red}d}\bv + \bigg( - \bu\times {\rm curl} \bv 
+ \nabla\Big(  p + \frac{1}{2}|\mathbf{u}|^2\Big)  + gb\nabla z\bigg)dt  = 0
\,,
 \label{SCL-MotEqn-vec-thm}
\end{align}
for $\mathbf{v} = \mathbf{u} {-} \bu^S(\bx)N_t $ defined in \eqref{def-v}, solution $N_t$ of the OU process \eqref{OU-1}, and divergence-free velocities of Lagrangian mean transport velocity $\bu$ and Stokes mean drift $ \bu^S(\bx)$. The auxiliary advection equations for the volume element, $D$, and buoyancy $b$ are
\begin{align} 
{\sf \textcolor{red}d} D + {\rm div}(D \mathbf{u})\,dt = 0
\quad\hbox{and}\quad
{\sf \textcolor{red}d} b + \mathbf{u}\cdot\nabla b\,dt = 0
\,. 
 \label{SCL-aux-eqns-thm}
 \end{align}
\end{theorem}

\begin{remark}[Determining the pressure, $p$]
In the motion equation \eqref{SCL-MotEqn-vec-thm},  $N_t$ is the solution \eqref{OUsoln-1} of the OU process \eqref{OU-1}, and the prescribed Stokes mean drift velocity $\bu^S(\bx)$ is divergence-free. Hence, the Lagrangian mean transport velocity $\bu$ remains divergence-free, $\nabla\cdot \mathbf{u} = 0$, as a result of the constraint $D=1$ imposed by the Lagrange multiplier $p$, the pressure, in Hamilton's principle \eqref{Lag-det-CL} in combination with the auxiliary continuity equation in \eqref{SCL-aux-eqns-thm} for the volume element, $D$. The equation for the pressure, $p$, arises from the divergence of the motion equation \eqref{SCL-MotEqn-vec-thm}, as the Poisson equation, 
\begin{align} 
-\,\Delta \Big(p + \frac{1}{2}|\mathbf{u}|^2\Big) 
=
{\rm div} \big( \bu\times {\rm curl} \bv  - gb\nabla z \big)
\,, \label{p-eqn}
 \end{align}
 with Neumann boundary conditions obtained by evaluating the normal component of the motion equation \eqref{SCL-MotEqn-vec-thm} at the boundary of the flow domain. 
Then, since $\mathbf{v} = \mathbf{u} - \bu^S(\bx)N_t $ by \eqref{def-v}, in principle, the pressure should be written as $p=p_0\,dt + p_1 N_t$ and the quantities $p_0$ and $p_1$ should be determined separately in the Poisson equation, \eqref{p-eqn}. However, we shall forgo this technical feature in favour of keeping the notation transparent. For a full explanation of semimartingale-driven variational principles, see \cite{SC2020}.
\end{remark}

\paragraph{Three equivalent forms of the OU CL equations.}
We will write the motion equation \eqref{SCL-MotEqn} in three equivalent vector forms and 
discuss each form separately to extract the information it presents most conveniently. 
A different parsing of the information in the \emph{deterministic} version of the CL motion equation \eqref{SCL-MotEqn} has been 
proposed in \cite{SuzF-K-StokesForces2016} by regarding the Stokes terms as components of \emph{forces} 
rather than contributions to the \emph{momentum} and the circulation, as done here. 

(1) The first of these three equivalent forms of the motion equation \eqref{SCL-MotEqn} is already in equation \eqref{SCL-MotEqn-vec-thm} in Theorem \ref{SALT-OU-CL-Thm}. This form is
\begin{align} 
{\sf \textcolor{red}d}\mathbf{v} - \mathbf{u}\times\curl\mathbf{v}  \,dt
+ \nabla\left(p + \frac{1}{2}|\mathbf{u}|^2\right) \,dt = -\,gb\nabla z \,dt
\,,
\label{SCL-MotEqn1}
\end{align}
written in terms of the pressure $p$ and the Lagrangian mean transport velocity $\mathbf{u}$ and the Eulerian mean velocity $\mathbf{v} = \mathbf{u} {-} \bu^S(\bx)N_t $ in a reference frame moving with velocity $-\,\bu^S(\bx)N_t $. The form \eqref{SCL-MotEqn1} expresses the Kelvin circulation theorem for OU CL wave dynamics as
\begin{align}
{\sf \textcolor{red}d} \oint_{c(\mathbf{u})} \mathbf{v}\cdot d\mathbf{x} = -\oint_{c(\mathbf{u})} gbd z\,dt \,,  
\label{Kel-thm-OU}
\end{align}
with transport (Lagrangian) velocity $\mathbf{u}$ and transported (Eulerian) velocity $\mathbf{v}= \mathbf{u} - \bu^S(\bx)N_t $. Equation \eqref{Kel-thm-OU} recovers the Kelvin circulation theorem for the Craik-Leibovich theory, upon identifying $\mathbf{u}=\mean{\mathbf{u}}_L$ as the Lagrangian mean velocity and $\mathbf{v}=\mean{\mathbf{u}}_E$ as the Eulerian mean velocity. 

Alternatively, one may take the curl of equation \eqref{SCL-MotEqn1}, to write the pathwise equation for total vorticity $\sym{\omega}:={\rm curl}\bv$ as
\begin{align}
{\sf \textcolor{red}d} \sym{\omega} + (\bu\cdot\nabla \sym{\omega} - \sym{\omega}\cdot\nabla \bu)\,dt
= -\,g\nabla b\times\nabla z\,dt \,, 
\label{vort-eqn-OU}
\end{align}
as well as the conservation of potential vorticity (PV) on Lagrangian particles. Namely,
\begin{align} 
{\sf \textcolor{red}d} Q + \bu \cdot \nabla Q\,dt = 0\,,
\label{CL-PV}
\end{align}
where PV is defined as $Q := \nabla b \cdot \sym{\omega}$.

\begin{remark}[OU CL with SALT]\label{op-SALT}
This first form of the OU CL equations \eqref{SCL-MotEqn1} admits stochastic advection by Lie transport (SALT), as well as the OU process. Namely, by following \cite{Holm2015} we find a modification of the Kelvin circulation theorem in \eqref{Kel-thm-OU} given by
\begin{align}
{\rm \textcolor{red}d} \oint_{c(\mathbf{\widetilde{u}})} \mathbf{v}\cdot d\mathbf{x} =  -\oint_{c(\mathbf{\widetilde{u}})} gbd z\,dt \,,  
\label{CL-circ}
\end{align}
with stochastic transport velocity in the SALT form \cite{Holm2015},
\begin{align*}
\mathbf{\widetilde{u}}:= \mathbf{u}(\bx,t)dt + \sum {\boldsymbol \xi}(\mathbf{x})\circ dW_t
\,.
\end{align*}
The corresponding SALT CL motion equation is given by
\begin{align} 
{\rm \textcolor{red}d}\mathbf{v} - \mathbf{\widetilde{u}}\times\curl\mathbf{v} 
+ \nabla\left( {\rm \textcolor{red}d}p {-} \frac{1}{2}|\mathbf{u}|^2dt + \mathbf{u}\cdot \mathbf{\widetilde{u}} \right) 
&= -\,g\nabla b\times\nabla z\,dt\,,
\label{SCL-MotEqn+SALT}
\end{align}
in which now the pressure ${\rm \textcolor{red}d}p$ is a semimartingale, see \cite{SC2020}. 
The corresponding vorticity equation keeps its form, as in \eqref{vort-eqn-OU},
\begin{align}
{\rm \textcolor{red}d}\sym{\omega} + \mathbf{\widetilde{u}}\cdot\nabla \sym{\omega}
 - \sym{\omega}\cdot\nabla \mathbf{\widetilde{u}}= 0= -\,g\nabla b\times\nabla z\,dt\,,
\label{vort-eqn-SALT}
\end{align}
as well as the conservation of potential vorticity (PV) on Lagrangian particles. Namely,
\begin{align} 
{\rm \textcolor{red}d} Q + \mathbf{\widetilde{u}}\cdot \nabla Q\,dt = 0\,,
\label{CL-PV}
\end{align}
where PV is still defined as $Q := \nabla b \cdot \sym{\omega}$.

The introduction of SALT in \eqref{SCL-MotEqn+SALT} yields the same equations as for the Richardson triple discussed in \cite{{Holm-RichTriple2019}}. Consequently, we may refer to \cite{{Holm-RichTriple2019}} for more discussion and further analysis of OU CL dynamics with SALT.   
\end{remark}

(2) The second of the three equivalent forms of the motion equation \eqref{SCL-MotEqn} we consider is reminiscent of an OU version of the electromagnetic Lorentz force on a fluid plasma, in the ``hydrodynamic'' gauge, $\phi + \bu\cdot \mathbf{A}=0$, 
which is the Coulomb gauge in the frame comoving with the fluid. Namely, 
\begin{align}
\big(\mathbf{E} + \mathbf{u}\times \mathbf{B} \big)\cdot d\bx
= \Big(-\,{\sf \textcolor{red}d}\mathbf{A} + \mathbf{u}\times\curl \mathbf{A} - \nabla (\bu\cdot \mathbf{A}) \Big)\cdot d\bx
= - \Big({\sf \textcolor{red}d} + \mathcal{L}_u \Big)(\mathbf{A}\cdot d\bx)
\label{SCL-MotEqn2}
\end{align}
with $\mathbf{A} = -\,\bu^S(\bx)N_t $ and  ${\sf \textcolor{red}d}\mathbf{A} = -\,\bu^S(\bx) {\sf \textcolor{red}d} N_t $.
Consequently, 
\begin{align} 
\Big({\sf \textcolor{red}d} + \mathcal{L}_u \,dt \Big)(\bu\cdot d\bx)
+  \big(\nabla p + g b\nabla z\big)\cdot d\bx \,dt
&=  
\Big({\sf \textcolor{red}d} + \mathcal{L}_u \Big)\big(\bu^S(\bx)N_t\cdot d\bx\big)
\,,
\label{SCL-MotEqn2}
\end{align}
which is written only in terms of velocity $\mathbf{u}$ and the OU frame velocity $\bu^S(\bx)N_t =\mathbf{u} - \mathbf{v} $. 
Here, it is not necessary for the pressure to be a semimartingale, because the semimartingale term on the right side of equation \eqref{SCL-MotEqn2} vanishes when the divergence is taken, since ${\rm div}\bu^S(\bx)=0$.

\begin{remark}[OU CL \emph{non-acceleration} theorem]
Vanishing of the right-hand side of the OU CL motion equation \eqref{SCL-MotEqn2} would correspond to the non-acceleration theorem  \eqref{SALT-SNWP-GLM-total}  for GLM. The condition corresponding to equation \eqref{non-accel} for GLM for non-acceleration to occur in the case of OU CL,  is that
\begin{equation}
\Big({\sf \textcolor{red}d} + \mathcal{L}_u \Big)\big(\bu^S(\bx)N_t\cdot d\bx\big) = 0
\,.
\label{frozen-in-uS}
\end{equation}
Thus, enforcing the non-acceleration condition \eqref{frozen-in-uS} on the Craik-Leibovich model would impose conservation of circulation of the Stokes mean drift velocity around a material loop moving  the flow of the Lagrangian mean velocity, $\bu$. That is, 
\[
{\sf \textcolor{red}d} \oint_{c(u)} \bu^S(\bx)N_t\cdot d\bx = 0\,.
\]
On the other hand, the CL model has been derived as an \emph{external} ponderomotive force exerted on the flow as a result of averaging over rapid oscillations imposed near the upper boundary. In contrast, as discussed at the beginning of section \ref{det-GLM-bkgrnd}, the GLM model has been derived by seeking an \emph{internal} ponderomotive force force, which is generated by a fluctuating component of the Lagrangian trajectory in equation \eqref{LM-def-rel}.  Thus, the non-acceleration result for GLM would not be expected to apply either to the CL model, or to its probabilistic counterpart, the OU CL model. 
\end{remark}

(3) The third of the three equivalent forms of the motion equation \eqref{SCL-MotEqn} discussed here introduces an OU version of the usual expression for the Craik--Leibovich `vortex force' given by
\begin{align}
{\sf \textcolor{red}d} \mathbf{v} - \mathbf{v}\times\curl \mathbf{v}\,dt 
+ \nabla\left( p + \frac{1}{2}|\mathbf{v}-\bu^S(\bx)N_t|^2\right) \,dt 
&=- \, \bu^S(\bx)N_t\times\curl\mathbf{v}\,dt
\,,
\label{SCL-MotEqn3}
\end{align}
written only in terms of the \emph{Eulerian} mean velocity $\mathbf{v}$ and the OU frame velocity, $-\,\bu^S(\bx)N_t$.

\begin{remark}[Potential caveat]
The square of the Ornstein-Uhlenbeck process in \eqref{SCL-MotEqn3} is the solution to the following pathwise differential equation
\begin{align}
\begin{split}
\frac12{\sf \textcolor{red}d}N_t^2 &= N_t {\sf \textcolor{red}d} N_t + \sigma^2 \,dt\\
&= \big(\theta N_t(\mean{N} - N_t) + \sigma^2\big)dt + \sigma N_t \, dW_t\,.
\end{split}
\label{OUprocess-squared}
\end{align}
The Stokes drift velocity $\bu^S(\bx)N_t$ in the third equivalent form of the OU CL equations \eqref {SCL-MotEqn3} satisfies a Bernouilli type ODE, which has finite time blow up, when either $\mean{N}$ is negative, or if $\mean{N}$ is positive, and the initial datum is sufficiently negative.
To avoid this issue, one can consider replacing the OU process by a time-integrated OU process. This option will be investigated elsewhere.
\end{remark}

\begin{remark}[Physical interpretation]
From the viewpoint of Kelvin's theorem, fluid parcels are transported by a Lagrangian mean velocity, which is an average of the fluid parcel velocity taken at fixed Lagrangian label. Kelvin's theorem in \eqref{Kel-thm-OU} for the OU CL model states that the circulation integral of the total Eulerian mean velocity $\mathbf{v} = \mathbf{u} - \bu^S(\bx)N_t $ around material loops moving with the Lagrangian mean velocity is generated only by non-vertical buoyancy gradients.  

In the proposed formulation of the OU CL model considered here for modelling uncertainty in the Stokes mean drift velocity on the CL solution, we have replaced the standard CL Stokes mean drift velocity $\bu^S(\bx)$ by an OU process in the integrand of the Kelvin circulation.  In future work, we will explore this direction farther, since uncertainty in the Stokes mean drift velocity is bound to be an issue in the calibration and assimilation of ocean data observed from space \cite{Klein-OceanData2019}. 

\end{remark}

\section{Conclusion}\label{concl-sec}

In this paper, we have modelled multi-scale, multi-physics uncertainty in wave-current interaction (WCI), by introducing stochasticity into the wave dynamics of two classic models of WCI; namely,  the Generalised Lagrangian Mean (GLM) model and the Craik--Leibovich (CL) model.  The two models acquire different types of stochasticity through their respective derivations from Hamilton's principle for different types of Lagrangians.

One main result of the present work is the derivation via Hamilton's principle of a closed dynamical model of wave-current interaction (WCI) which applies to GLM and can be extended into stochastic wave-current dynamics. The closure is governed by the choice of wave Hamiltonian in the phase-space Lagrangian. The wave Hamiltonian is chosen to match the WKB dynamics of the phase and wave action density in the local reference frame of the moving fluid. The model is flexible enough to include a variety of different wave fields, and for the waves and currents to be made stochastic in different ways for testing various causes of uncertainty. The model would apply, for example, as an efficient way of adding nonlinear wave physics which may not have been considered, needed or resolved in a previous model, or in a different regime of operation. 

Many open questions for future research have arisen in developing the stochastic Hamilton's principle framework, which was created primarily for uncertainty quantification in the hybrid wave-current interaction.  
On the Hamiltonian side, for example, the framework developed here for GLM leads to a type of non-canonical Lie--Poisson bracket discovered for superfluid $^4He$ and $^3He$ in \cite{HK1982} which was also observed by Krishnaprasad and Marsden in \cite{KM1987} for the motion of a rigid body with a flexible attachment.   The KM87 theory formulated a Lie group structure which had already been the basis for several useful theories of hybrid plasma-fluid interaction dynamics on the Hamiltonian side \cite{Tronci2010} and its formulation on the Hamilton's principle side has been accomplished in \cite{Close2019,CBT2018,HolmTronci2012}.  Thus, KM87 is a natural Hamiltonian partner for the stochastic WCI hybrid theory which has been developed here on the Lagrangian, or Hamilton's principle, side. Conversely, one may consider passing from the known KM87 Hamiltonian descriptions of hybrid kinetic theory and fluid plasma systems, either to the corresponding derivation on the Hamilton's principle side for additional modelling purposes, or directly to a stochastic Hamiltonian model as in section \ref{SW-WCI-sec}. 

Regarding specific topics for further research, one may consider testing the effectiveness of the model in different situations by investigating other types of WCI for a variety different types of wave physics.  For example, one could develop a self-consistent WCI theory for Kelvin waves propagating on superfluid vortices which are being transported by the surrounding flow. A WCI theory of Alfv\'en waves propagating on dynamics of magnetic field lines in magnetohydrodynamics (MHD) could also be developed, perhaps by following ideas for time-mean oscillation center dynamics for MHD in \cite{SKH1986}. One may also consider introducing this approach for probing the effects of submesoscale physics in oceanography. 

As we have stressed, the main geometric mechanics ideas for our approach to WCI for GLM (including the phase-space Lagrangian approach, of course) were already developed and in effective use in particle-fluid plasma physics at least forty years ago \cite{Dewar1973, RGL1981, ANK-DH1984,SKH1986}. For additional background in this matter, see, e.g., \cite{Brizzard2009,BurbyRuiz2019}. However, the connection of these mainstream geometric mechanics ideas to stochastic methods for uncertainty quantification in WCI for GLM with potential applications in oceanography, for example, has been waiting until now to be made. 

In contrast, the Hamilton's principle for addressing stochastic WCI for the CL model has no wave contribution. However, a change of frame is possible to model the effects of Stokes drift which can be probabilistic with an OU process. The CL model also admits stochastic advection by Lie transport (SALT). In the case of incompressible 3D EB fluid dynamics, the two theories can be made to converge because of the non-acceleration result for the GLM case. 

\subsection*{Acknowledgements}
I am enormously grateful for years of thoughtful discussions and correspondence about continuum dynamics of waves and currents with my friends and colleagues. Particularly in discussing the present work let me mention, B. Chapron, C. J. Cotter, D. Crisan, T. D. Drivas, F. Gay-Balmaz, M. Ghil, J. D. Gibbon, P. Korn, V. Lucarini, E. Luesink, J. C. McWilliams, E. M\'emin, O. Street, S. Takao and C. Tronci. This work was partially supported by EPSRC Standard grant EP/N023781/1 and by ERC Synergy Grant 856408 - STUOD (Stochastic Transport in Upper Ocean Dynamics).


\appendix

\section{Dynamical systems analogues of WCI}\label{appendix-A}

\subsection{Gyrostat: Rigid body with flywheel}

As we will see, the rigid body with flywheel along the intermediate principle axis in the body seems to be a closer  analogue to deterministic WCI than the swinging spring does.
Just as for the isolated rigid body, the energy is purely kinetic; so one may define the kinetic energy Lagrangian for this system $L:\,TSO(3)/SO(3)\times TS^1\to \mathbb{R}^3$ as
\begin{eqnarray}
L(\boldsymbol{\Omega},\,\dot{\phi})
=
\frac{1}{2}\lambda_1\Omega_1^2 
+
\frac{1}{2}I_2\Omega_2^2 
+
\frac{1}{2}\lambda_3\Omega_3^2 
+
\frac{1}{2}J_2(\dot{\phi}+\Omega_2)^2 
\,,
\end{eqnarray}
where 
${\boldsymbol{\Omega}}=(\Omega_1,\Omega_2,\Omega_3)$ is the angular velocity vector of the rigid body, $\dot{\phi}$ is the rotational frequency of the flywheel about the  intermediate principal axis of the rigid body,, and 
$\lambda_1$, $I_2$, $J_2$, $\lambda_3$ are positive constants corresponding to the principal moments of inertia, including the presence of the flywheel. Because the Lagrangian is independent of the angle $\phi$, its canonically conjugate angular momentum 
$
{N}:=\partial L/\partial \dot{\phi}
$ will be conserved. 
This suggests a move into the Hamiltonian picture, where the conserved ${N}$ will become a constant parameter. 
\begin{itemize}
\item
If we perform a partial Legendre transform in the flywheel variables $(\phi,\dot{\phi})\in TSO(2)$, we will obtain
\begin{eqnarray}
L(\boldsymbol{\Omega},\,\dot{\phi})
=
\frac{1}{2}\lambda_1\Omega_1^2 
+
\frac{1}{2}I_2\Omega_2^2 
+
\frac{1}{2}\lambda_3\Omega_3^2 
+
N\bigg(\dot{\phi}+\Omega_2 -\frac{N}{2J_2}\bigg) 
\,,
\end{eqnarray}
which is analogous to the phase space Lagrangians in  equations \eqref{Lag-det} and \eqref{Lag-det1}.
\item
Legendre-transforming this Lagrangian allows us to express its  Hamiltonian in terms of the angular momenta 
${\boldsymbol{\Pi}}=\pa L/\pa \boldsymbol{\Omega}\in\mathbb{R}^3$ and ${N}=\pa L/\pa \dot{\phi}  \in\mathbb{R}^1$ of the rigid body and flywheel, respectively,
\begin{eqnarray}
H({\boldsymbol{\Pi}},{N})
&=&
{\boldsymbol{\Pi}}\cdot{\boldsymbol{\Omega}}
+
{N}\dot{\phi}
- L(\boldsymbol{\Omega},\,\dot{\phi})
\label{offset-erg}\\[3pt]
&=&
\frac{\Pi_1^2}{2\lambda_1}
+
\frac{\Pi_3^2}{2\lambda_3}
+
\underbrace{\
\frac{1}{2I_2} \big(\Pi_2 - {N} \big)^2
}_{\hbox{offset along $\Pi_2$}}
 +\ 
\frac{{N}^2}{2}\Big(\frac{1}{I_2}+\frac{1}{J_2}\Big)
\,.
\nonumber
\end{eqnarray}
This Hamiltonian is an ellipsoid in coordinates $\sym{\Pi}\in\mathbb{R}^3$, whose centre is \emph{offset} in the $\Pi_2$-direction by an amount equal to the conserved angular momentum ${N}$ of the flywheel. 

The offset of the energy ellipsoid by ${N}$ along the $\Pi_2$-axis radically alters its intersections with the angular momentum sphere $|\sym{\Pi}|=const$. Its dynamical behaviour, given by motion along these altered intersections is quite different from that of the rigid body, which has no offset of its energy ellipsoid. In particular, the offset due to presence of the flywheel induces an intricate sequence of bifurcations of the equilibrium solutions which do not occur for the rigid body, for $N=0$  \cite{Elipe1997}.


\item
The Poisson bracket in the variables $\sym{\Pi},{N},\phi\in so(3)^*\times T^*S^1$ is a direct sum of the 
rigid-body bracket for $\sym{\Pi}\in so(3)^*\simeq\mathbb{R}^3$ and the canonical bracket for the flywheel phase-space coordinates $({N},\phi)\in T^*S^1$:  
\begin{eqnarray}
\{F\,,\,H\}
&=&
 -\,{\boldsymbol  \Pi} \cdot 
\bigg(\frac{\partial  F }{\partial\boldsymbol{\Pi}}
 \times
\frac{\partial H }{\partial\boldsymbol{\Pi}}\bigg)
+
\frac{\partial F }{\partial \phi}\frac{\partial H }{\partial {N}}
-
\frac{\partial H }{\partial \phi}\frac{\partial F }{\partial {N}}
\,.
\end{eqnarray}
The corresponding Hamiltonian equations may be written in a block-diagonal Poisson matrix form 
which is similar to that in equation \eqref{m+D+b-eqns} for WCI in Euler-Boussinesq equations and in \eqref{m+D-eqns} for WCI in 1D shallow water,
\begin{align}
\frac{d}{dt}
\begin{bmatrix}
\sym{\Pi} \\ \phi \\ N
\end{bmatrix}
= -
\begin{bmatrix}
\sym{\Pi} \times  & 0 & 0
\\
 0  & 0 & -1
 \\
  0  & 1 & 0
\end{bmatrix}
\begin{bmatrix}
\partial H / \partial \sym{\Pi} = \sym{\Omega}
\\
\partial H / \partial \phi =   0
\\
\partial H / \partial N = {J_2}^{-1}N - {I_2}^{-1} \big(\Pi_2 - {N} \big)
\end{bmatrix}.
\label{gyrostat-eqns}
\end{align}

\end{itemize}

\subsection{The deterministic swinging spring} 

A dynamical systems analogue of WCI arises in the oscillation-rotation interaction (ORI) seen in the elastic spherical pendulum, or swinging spring \cite{HolmLynch2002}. In this system, one may see regular exchanges between springing motion (oscillation) and swinging motion (rotation). The Lagrangian for the swinging spring is 
\begin{align}
L(\bx,{\bf{\dot{x}}}; {\bf{\hat{e}}}_3) = \frac{m}{2} |{\bf{\dot{x}}}|^2 - mg\,{\bf{\hat{e}}}_3\cdot {\bx} -  \frac{k}{2} \big( |\bx|^2-|\bx_0|^2\big)
\,,
\label{Lag-ESP1}
\end{align}
with notation  $(\bx,{\bf{\dot{x}}})\in T\mbb{R}^3$, vertical unit vector ${\bf{\hat{e}}}_3$ and constants of gravity $(g)$, mass of the bob $(m)$, isotropic spring constant $(k)$ and initial position $\bx_0\in \mbb{R}^3$.

The time dependent solution path for the Euler--Lagrange equations which follow from Hamilton's principle for the Lagrangian \eqref{Lag-ESP1} is denoted as $\bx(t)\in \mbb{R}^3$. One may lift the solution path $\bx(t)\in \mbb{R}^3$ into the Lie group $\mbb{R}_+\times SO(3)$ of scaling and rotation of vectors in $\mbb{R}^3$ by specifying its direct-product  action on an initial position $\bx_0\in\mbb{R}^3$ as 
\begin{align}
\bx(t) = R(t)O(t)\bx_0
\quad\hbox{for}\quad
\big(R(t),O(t)\big)\in \mbb{R}_+\times SO(3)
\,.
\label{RXO-act}
\end{align}
Under the scaling and rotation action in \eqref{RXO-act}, the terms in the Lagrangian \eqref{Lag-ESP1} transform as
\begin{align}
\begin{split}
|\bx|^2 &= |O(R\bx_0)|^2 = |R\bx_0|^2
\\
|{\bf{\dot{x}}}|^2 &= |\dot{R}\bx_0 + \mb{\Omega}\times R\bx_0|^2 
= |\dot{R}\bx_0|^2 + |\mb{\Omega}\times R\bx_0|^2
\\
{\bf{\hat{e}}}_3\cdot {\bx} &= (O^{-1} (t){\bf{\hat{e}}}_3)\cdot R(t)\bx_0 =: \mb{\Gamma}(t)\cdot R\bx_0
\,.
\end{split}
\label{RXO-action-TR3}
\end{align}
Here $O^{-1}\dot{O} =: \widehat{\Omega} =: \mb{\Omega}\times$, where $\widehat{\Omega}_{ij}=-\epsilon_{ijk}\Omega^k$ is the \emph{hat map} isomorphism which represents the angular frequency of rotation as induced by either the skew symmetric $3\times3$ matrix Lie algebra $\mathfrak{so}(3)$, or the cross product of vectors in 3D Euclidean space $\mbb{R}^3$. From its definition $\mb{\Gamma}(t):=O^{-1} {\bf{\hat{e}}}_3 \in \mbb{R}^3$, one finds the evolution equation
\begin{align}
\mb{\dot{\Gamma}} + \mb{\Omega}\times \mb{\Gamma} = 0\,,
\label{Gamma-eqn}
\end{align}
and one notes that $|\mb{\Gamma}|^2=|{\bf{\hat{e}}}_3|^2 =1$.

Under the scaling and rotation action in \eqref{RXO-act} the Lagrangian \eqref{Lag-ESP1} with $ \mb{S}:=R(t)\bx_0$ transforms as
\begin{align}
L(\mb{\Omega}, \mb{\Gamma}; \mb{S},\mb{\dot{S}}) 
=&\, \frac{m}{2} |\mb{\dot{S}} + \mb{\Omega}\times \mb{S}|^2 - mg\,\mb{\Gamma}\cdot \mb{S} 
-  \frac{k}{2} \big(|\mb{S}|^2 - |\bx_0|^2\big)
\nonumber\\
=&\, \frac{m}{2} |\mb{\dot{S}}|^2 + \frac{m}{2} |\mb{\Omega}\times \mb{S}|^2 - mg\,\mb{\Gamma}\cdot \mb{S} 
-  \frac{k}{2} \big(|\mb{S}|^2 - |\bx_0|^2\big)\label{Lag-ESP2}
\\
L(\mb{\Omega}, \mb{\Gamma}; \mb{S},\mb{P}) 
=&\underbrace{ \frac{m}{2} |\mb{\Omega}\times \mb{S}|^2 - mg\,\mb{\Gamma}\cdot \mb{S} }
_{\hbox{Rotations \& Gravity}}
+ \underbrace{ \Big(\mb{\Omega}\cdot \mb{S}\times \mb{P}\Big)
}_{\hbox{Coupling term}}
+ \underbrace{ 
 \mb{P}\cdot \mb{\dot{S}} 
- \Big(\frac{1}{2m} |\mb{P}|^2 
+  \frac{k}{2} \Big( |\mb{S}|^2 - |\bx_0|^2 \Big)\Big)
}_{\hbox{Oscillation phase-space Lagrangian}}
\,,\nonumber
\end{align}
where $\mb{P}:=\partial L/ \partial \mb{\dot{S}}=m\mb{\dot{S}}$ from the second line. The cross term in the square of the total velocity in the first line has vanished, because the swinging velocity $\mb{\Omega}\times \mb{S}$ and springing velocity $\mb{\dot{S}} $ are orthogonal. That is, $2\mb{\dot{S}} \cdot \mb{\Omega}\times \mb{S}=2\dot{R}\bx_0\cdot\mb{\Omega}\times R\bx_0=0$.
We see that the coupling term which boosts the spatial oscillations into the rotating frame also vanishes, i.e., the total angular momentum is given by
\begin{align}
\mb{\Pi}:=\frac{\partial L}{ \partial \mb{\Omega}} 
= m\mb{S}\times(\mb{\Omega}\times\mb{S}) + \mb{S}\times \mb{P}
= m\mb{S}\times(\mb{\Omega}\times\mb{S})
\quad\hbox{since}\quad
\mb{S}\times \mb{P} = 0
\,.
\label{Lag-Pi-mom}
\end{align}
With the vanishing of the coupling term, these manipulations have separated the original Lagrangian into an $SO(3)$-reduced Lagrangian for rotations ($\mb{\Omega}$) and translations ($\mb{\Gamma}$) in the \emph{body} frame, plus an $(\mb{S},\mb{P})$ phase-space Lagrangian for oscillations in the \emph{spatial} frame. This is consistent with our intuition that the springing motion could occur even if the spherical pendulum were not swinging. In this system, the springing oscillations are the analogues of the waves in WCI dynamics. Likewise, the swinging motion due to exchange of kinetic and gravitational energies in the rotating frame are the analogues of the currents interacting via exchanges in physics and energetics with their advected quantities in WCI. 

From their definitions, $O^{-1}\dot{O} =: \widehat{\Omega} =: \mb{\Omega}\times$ and $\mb{\Gamma}(t):=O^{-1} {\bf{\hat{e}}}_3$, a manipulation using the hat map delivers the variations of $\mb{\Omega}$ and $\mb{\Gamma}$ arising from variations of $O(t)\in SO(3)$ in $\mbb{R}^3$ vector form as
\begin{align}
\delta \mb{\Omega} = \mb{\dot{\Sigma}} +  \mb{\Omega}\times  \mb{\Sigma}
\quad\hbox{and}\quad
\delta \mb{\Gamma} = - \, \mb{\Sigma} \times \mb{\Gamma}
\quad\hbox{for}\quad
\mb{\Sigma} \times = O^{-1}\delta O
\,.
\label{Var-forms}
\end{align}
Upon substituting these variational formulas into Hamilton's principle, $\delta S = 0$, with action integral $S=\int_a^b L(\mb{\Omega}, \mb{\Gamma}; \mb{S},\mb{P}) \,dt$,  we have
\begin{align}
\begin{split}
0 = \delta S=\int_a^b & \big(-\mb{\dot{\Pi}} - \mb{\Omega}\times\mb{\Pi} + mg \mb{\Gamma}\times \mb{S}\big)\cdot  \mb{\Sigma}
+ \delta\mb{P} \cdot \big(\mb{\dot{S}} - \mb{P}/m\big)
\\& + \big(- \mb{\dot{P}} - k \mb{S} - mg \mb{\Gamma} - m \mb{\Omega}\times(\mb{\Omega}\times\mb{S})\big)\cdot \delta \mb{S}\,dt + \big[\mb{\Pi}\cdot \mb{\Sigma} \big]_a^b + \big[\mb{P} \cdot \delta\mb{S}\big]_a^b
\,.
\end{split}
\label{Var-forms}
\end{align}
Hamilton's principle now implies the dynamics for the elastic spherical pendulum, provided the variations $\mb{\Sigma}$ and $\delta\mb{S}$ vanish at the endpoints in time. These dynamics comprise two Euler--Poincar\'e equations for the rotations and librations,
\begin{align}
\begin{split}
\mb{\dot{\Pi}} + \mb{\Omega}\times\mb{\Pi} &= mg \mb{\Gamma}\times \mb{S}
\quad\hbox{and}\quad
\mb{\dot{\Gamma}} + \mb{\Omega}\times\mb{\Gamma} = 0
\,,\quad\hbox{for}\quad
\mb{\Pi}:=  m\mb{S}\times(\mb{\Omega}\times\mb{S})
\,,\end{split}
\label{Ham-princ-eqns1}
\end{align}
and two canonical Hamiltonian equations for the springing degree of freedom,
\begin{align}
\mb{\dot{S}} = \mb{P}/m
\quad\hbox{and}\quad
 \mb{\dot{P}} =  m\mb{\ddot{S}} = - k \mb{S} - mg \mb{\Gamma} - m \mb{\Omega}\times(\mb{\Omega}\times\mb{S})
\,.\label{Ham-princ-eqns2}
\end{align}
The first set of these equations has the form of a \emph{heavy top} whose vector $\mb{S}$ from the support to the centre of mass has its own dynamics. The second set reveals the $\mb{S}$ dynamics to be Newtonian with a sum of three forces, the spring restoring force, gravity and the centrifugal force.  Clearly, the oscillations in $\mb{S}$ will drive rotational motion in $\mb{\Pi}$ and $\mb{\Gamma}$, which will feed back to $\mb{S}$, provided the initial condition is not oriented vertically.\footnote{A vertical initial condition $\bx_0$ would make the initial gravitational torque vanish ($\mb{\Gamma}(0)\times \mb{S}(0)=0$), since $\mb{\Gamma}(0)={\bf{\hat{e}}}_3$. This would allow purely vertical oscillations which would not induce rotation starting from a stationary initial condition.} The discussion here of the swinging spring dynamics in which oscillations can drive rotations supports the analogous conclusions in the text (such as Corollary \ref{SALT-SNWP-KelThm}) that waves could drive currents.

\subsection{The stochastic swinging spring}

We might hope to use the same methods as in the text for stochastic WCI to include SALT noise in the swinging rotations of the elastic spherical pendulum. However, the analogy is not complete. In fact, the condition $\mb{S}\times \mb{P} = 0$ precludes introducing the analog of SNWP noise into the springing motions of the elastic spherical pendulum in the same way as we have done for the wave propagation in the EB fluid case in the text. The lesser task of including SALT noise only in the swinging rotations of the elastic spherical pendulum will not be pursued here, though, because the results would be too similar to the case of SALT noise for the rigid heavy top which has already been investigated in \cite{ACH-JNLS2018}.

\subsection{Gyroscopic analogy of non-inertial reference frames}

The Lagrangian for the free rotation of a rigid body at body angular frequency $\sym{\Omega}$
relative to a frame which is \emph{already} rotating about the same origin at a fixed angular frequency $\sym{\Upsilon}$ is given by
\[
\ell(\sym{\Omega};\sym{\Upsilon}) = \frac12\sym{\Omega}\cdot I\sym{\Omega} + \sym{\Omega}\cdot I\sym{\Upsilon}
\,,\]
where $I$ is the moment of inertia of the body. 
Hence, the total body angular momentum in the rotating frame is given by
\begin{equation}
\frac{\partial \ell( \sym{\Omega}; \sym{\Upsilon} ) } { \partial \sym{\Omega} } 
=: I(\sym{\Omega} +  \sym{\Upsilon})
\,.
\label{Rot-mom-shift}
\end{equation}

The corresponding Hamilton's principle
\[
0 = \delta S(\sym{\Omega};\sym{\Upsilon}) = \delta \int_{t_1}^{t_2}\frac12\sym{\Omega}\cdot I\sym{\Omega} + \sym{\Omega}\cdot I\sym{\Upsilon}\,dt
\,,\]
yields
\begin{equation}
\frac{d \sym{\Pi}}{dt} + \sym{\Omega}\times\sym{\Pi} = 0
\quad\hbox{with}\quad \sym{\Pi} := I(\sym{\Omega} +  \sym{\Upsilon})\,.
\label{Mom-shifted-eqn}
\end{equation}
Thus, moving into a rotating frame preserves the form of the rigid body equations in \eqref{Mom-shifted-eqn}. Consequently, this frame change preserves the conservation of $|\sym{\Pi}|^2= |I(\sym{\Omega} +  \sym{\Upsilon})|^2$. However, transforming into a rotating frame changes the definition of the angular momentum  $\sym{\Pi}$ to include the momentum associated with the additional angular velocity of the rotating frame, $\sym{\Upsilon}$. 



\end{document}